\documentclass[11pt]{article}
\pdfoutput=1

\usepackage{amsmath,amsthm}
\usepackage{preprint}

\usepackage[utf8]{inputenc}
\usepackage[T1]{fontenc}
\usepackage{graphicx}
\usepackage{float}
\usepackage{booktabs}
\usepackage{array}
\usepackage{hyperref}
\usepackage{cleveref}
\crefname{lstlisting}{listing}{listings}
\Crefname{lstlisting}{Listing}{Listings}
\usepackage{xcolor}
\usepackage{tikz}
\usetikzlibrary{arrows.meta, positioning, shapes.geometric, fit, calc, backgrounds}
\usepackage{pgfplots}
\pgfplotsset{compat=1.18}
\newcommand{\Mmark}{\ensuremath{\blacksquare}}
\newcommand{\Pmark}{\ensuremath{\square}}
\usepackage{listings}
\usepackage{enumitem}
\usepackage{caption}
\usepackage{subcaption}
\usepackage{needspace}
\usepackage{placeins}

\lstdefinelanguage{Isabelle}{
  morekeywords={locale, fixes, assumes, shows, lemma, theorem, definition,
    interpretation, datatype, fun, where, record, type_synonym,
    by, auto, cases, metis, unfold_locales, rule, proof, qed,
    Some, None, let, in, case, of, begin, end, section, text,
    imports, theory, corollary, using, from, obtain, then, have,
    show, with, next, also, finally, hence, thus, moreover,
    assumes, and, not, if, else, True, False},
  sensitive=true,
  morecomment=[s]{(*}{*)},
  morestring=[b]",
}

\usepackage{aliascnt}
\newtheorem{theorem}{Theorem}[section]
\newaliascnt{lemma}{theorem}\newtheorem{lemma}[lemma]{Lemma}\aliascntresetthe{lemma}
\newaliascnt{corollary}{theorem}\newtheorem{corollary}[corollary]{Corollary}\aliascntresetthe{corollary}
\newaliascnt{proposition}{theorem}\newtheorem{proposition}[proposition]{Proposition}\aliascntresetthe{proposition}
\theoremstyle{definition}
\newaliascnt{definition}{theorem}\newtheorem{definition}[definition]{Definition}\aliascntresetthe{definition}
\newaliascnt{remark}{theorem}\newtheorem{remark}[remark]{Remark}\aliascntresetthe{remark}
\theoremstyle{plain}
\crefname{theorem}{Theorem}{Theorems}
\crefname{lemma}{Lemma}{Lemmas}
\crefname{corollary}{Corollary}{Corollaries}
\crefname{proposition}{Proposition}{Propositions}
\crefname{definition}{Definition}{Definitions}
\crefname{remark}{Remark}{Remarks}
\crefname{section}{Section}{Sections}
\crefname{appendix}{Appendix}{Appendices}
\crefname{table}{Table}{Tables}
\crefname{figure}{Figure}{Figures}

\hypersetup{
  colorlinks=true,
  linkcolor=blue!70!black,
  citecolor=blue!70!black,
  urlcolor=blue!70!black,
  pdftitle={Soundness of the GKR Protocol for Sparse-Merkle Updates, with a Machine-Checked Core},
  pdfauthor={Jinwook Kim},
  pdfsubject={Soundness of the GKR protocol with a cube gate, its instantiation over a quartic extension field, exact arithmetization of sparse-Merkle transitions relative to the witness encoding, and their composition, with complete proofs and an Isabelle/HOL mechanization of the core},
  pdfkeywords={GKR protocol, sumcheck protocol, interactive proofs, computer-aided cryptography, Isabelle/HOL, extension fields},
}

\newcommand{\isaf}[1]{\texttt{#1}}
\newcommand{\F}{\mathbb{F}}
\newcommand{\cube}[1]{\{0,1\}^{#1}}
\newcommand{\smtvalid}{\isaf{smt\_valid}}
\newcommand{\AIdesignation}{Substantial AI contribution}
\newcommand{\eqk}{\mathrm{eq}}
\newcommand{\bpt}[2]{\langle #1\rangle_{#2}}
\newcommand{\cSC}{\mathrm{SC}}
\newcommand{\cBad}{\mathrm{Bad}}
\newcommand{\cRepr}{\mathrm{Repr}}
\newcommand{\ccs}{\mathrm{cs}}
\newcommand{\sout}{s_{\mathrm{out}}}
\newcommand{\dbnd}{\mathrm{db}}
\newcommand{\gLin}{\mathsf{Lin}}
\newcommand{\gMul}{\mathsf{Mul}}
\newcommand{\gPow}{\mathsf{Pow3}}
\newcommand{\extE}{\mathbb{E}}
\newcommand{\extK}{\mathbb{K}}
\newcommand{\hleaf}{h_{\mathrm{leaf}}}
\newcommand{\hnode}{h_{\mathrm{node}}}
\newcommand{\hwl}{\widehat{h}_{\mathrm{leaf}}}
\newcommand{\hwn}{\widehat{h}_{\mathrm{node}}}
\newcommand{\smSponge}{\mathsf{Sponge}}
\newcommand{\smEncode}{\mathsf{encode}}
\newcommand{\smPenc}{\mathsf{penc}}
\newcommand{\smFold}{\mathsf{fold}}
\newcommand{\smTag}{\mathrm{tag}}
\newcommand{\smEmpty}{\mathsf{Empty}}
\newcommand{\smOcc}{\mathsf{Occupied}}
\newcommand{\smTomb}{\mathsf{Tombstone}}
\newcommand{\smMem}{\mathsf{Membership}}
\newcommand{\smNon}{\mathsf{NonMembership}}
\newcommand{\smUpd}{\mathsf{Update}}
\newcommand{\smBit}{\mathrm{kbit}}
\newcommand{\smPathRoot}{\mathsf{PathRoot}}
\newcommand{\smPathOK}{\mathsf{PathOK}}
\newcommand{\smWl}{W_{\mathrm{leaf}}}
\newcommand{\smRep}{\mathrm{rep}}
\newcommand{\smPre}{\mathrm{pre}}
\newcommand{\smEnc}{\mathsf{Enc}}
\newcommand{\smCons}{\mathsf{Cons}}
\newcommand{\smRes}{\mathrm{Res}}
\newcommand{\smStack}{\mathrm{Stack}}
\newcommand{\smEv}{\mathrm{ev}}
\newcommand{\smAcc}{\mathsf{Acc}}
\newcommand{\smVAcc}{\mathsf{VAccept}}
\newcommand{\smPub}{\mathsf{PubOK}}
\newcommand{\smVd}{\mathsf{VdOK}}
\newcommand{\smEcho}{\mathsf{EchoOK}}
\newcommand{\smCanon}{\mathrm{canon}}
\newcommand{\smVdc}{\mathit{vd}^{\ast}}
\newcommand{\smRout}{r_{\mathrm{out}}}
\newcommand{\smMuxL}{\mathrm{mux}_{\mathrm{L}}}
\newcommand{\smMuxR}{\mathrm{mux}_{\mathrm{R}}}
\newcommand{\smKind}{\mathrm{kind}}
\newcommand{\smKM}{\mathsf{M}}
\newcommand{\smKN}{\mathsf{N}}
\newcommand{\smKU}{\mathsf{U}}
\newcommand{\smVecAcc}{\mathsf{VecAccept}}
\newcommand{\smWit}{\mathrm{wit}}
\usepackage{fancyhdr}
\title{Soundness of the GKR Protocol for Sparse-Merkle Updates,\\ with a Machine-Checked Core}

\author{
  Jinwook Kim\textsuperscript{1,2}~\href{https://orcid.org/0009-0004-4993-8005}{\mbox{\scalerel*{\begin{tikzpicture}[yscale=-1,transform shape]\pic{orcidlogo};\end{tikzpicture}}{|}}} \\
  \textsuperscript{1}Oraclizer Labs, Inc., Delaware, USA \\
  \textsuperscript{2}Oraclizer Labs Korea, Seoul, Korea \\
  \texttt{jay@oraclizer.io}
}

\date{November 2026}

\begin{document}

\maketitle

\begin{abstract}
The GKR soundness theorem bounds the acceptance of a false output claim for a fixed circuit and input. It does not say what the outputs of a compiled circuit mean for that input, and that has to be proved for each family of circuits. For a sparse Merkle tree over a compressing hash, heuristically almost every root admits a valid path for almost any leaf, so the needed statement concerns the witness as well as the claim. We prove per-input soundness for membership, non-membership and update claims, for trees of every depth, in the delegation setting of GKR: the verifier reads the circuit input in the clear and leaves the node hashes to the circuit. Assuming correct hash sub-circuits, every accepted input, honestly encoded or not, determines a valid witness, and the chain constraint cannot be dropped. If the witness read from an input fixed before the challenges is not valid, this verifier accepts with probability at most $\bigl(s_{\mathrm{out}}+\sum_i(8s_i+1)\bigr)/|\mathbb{F}|$ for a circuit with $2^{s_{\mathrm{out}}}$ outputs and $2^{s_i}$ wires read at layer $i$. This covers cube gates for the Poseidon2 S-box and, with $|\mathbb{F}| = p^4$, circuits over the 31-bit KoalaBear field with challenges from its quartic extension. The implemented node hash, a truncated permutation, is not collision resistant, so the statement that accepted runs for one key and root assert the same leaf assumes position binding of the tree. All our results are proved in the paper. Isabelle/HOL also checks a layer-chain bound, the exactness of the verifier's checks on encoded witnesses in a constraint-level model with abstract hash sub-circuits and, over one field, their composition for claims with no valid witness. All bounds concern the interactive public-coin protocol.
\end{abstract}

\noindent\textbf{AI-use designation (IACR AI policy):} \AIdesignation. AI systems drafted the manuscript, including the paper proofs, wrote the Isabelle and Rust code, and were used to check and review them, directed by the author, who defined the research questions, the system design and the scope of the verification, decided which statements are claimed, reviewed the manuscript, and takes full responsibility for the work. No AI system is an author. The Declarations give the details.


\section{Introduction}
\label{sec:introduction}

\paragraph{Soundness of GKR for authenticated-data-structure updates.} Interactive proofs of the Goldwasser--Kalai--Rothblum (GKR) family~\cite{gkr2015,thaler2013} reduce a claim about the output of a layered arithmetic circuit, layer by layer, to a claim about its input, running one sumcheck protocol~\cite{lfkn1992} per layer. Their soundness theorem bounds the probability that the verifier accepts a false output claim by a small multiple of $1/|\F|$. A system that accepts updates of a sparse Merkle tree on the strength of such a verifier needs a statement about updates: an invalid update is accepted only with small probability. For a compressing hash that statement has to concern the witness as well as the claim, because, heuristically, almost every root admits a valid path for almost any leaf (\cref{sec:goal}). It rests on three arguments.
\begin{enumerate}[leftmargin=*,label=(\roman*)]
\item \emph{Protocol soundness}: the GKR bound, for the gate set and the field that the system uses.
\item \emph{Exact arithmetization}: the compiled checks accept exactly the valid updates. A missing check lets the protocol prove a weaker statement than intended with probability one, and a superfluous check makes honest updates unprovable.
\item \emph{Commitment binding}: the root binds the state. For the hash functions used here this is position binding of the tree, because the node hash is a truncated permutation and is not collision resistant (\cref{sec:semantics}).
\end{enumerate}
Small-field provers add a further step. They evaluate circuits over a small prime field, 31 bits here, and draw the verifier's challenges from an extension field, so the bound must hold over the field the challenges come from and the circuit must be carried into that field.

\paragraph{Machine-checked cryptography for proof systems.} Machine-checked security proofs have recently reached proof systems: the sumcheck protocol in Isabelle/HOL~\cite{bosshard2024,sumcheckafp2024} and interactive oracle reductions in Lean~4, with sumcheck and further protocols under active formalization~\cite{arklib}. The authors of the Isabelle sumcheck development name the analysis of protocols such as GKR, which rely directly on sumcheck, as future work~\cite{bosshard2024}. Lean developments covering GKR have since appeared. complexitylib proves perfect completeness and soundness of the GKR protocol for layered circuits with addition and multiplication gates over a finite field, composing round-by-round soundness of its stages~\cite{complexitylib2026}. The Lean library of the GordieHowe proof system machine-checks the security of a GKR-based argument and proves, from certificates over the exported circuits that the Lean kernel checks, that each compiled circuit of its applications enforces the checks assigned to it~\cite{sharpgrubbs2026}. The Lean development of Longfellow proves conditional extraction and false-statement soundness for a model of its version-7 non-interactive verifier, starting from a trusted circuit~\cite{longfellowformal2026}. leanerVM proves round-by-round knowledge soundness of the oracle protocol of leanVM, in which a grand-product GKR checks the table bus, and open pull requests extend this to a Lean implementation of its non-interactive verifier~\cite{leanervm2026}.

\paragraph{This paper.} Of the three arguments, the first is standard up to the gate set and the field. Our subject is the second, what the outputs of the compiled circuit say about its input, which has to be proved for each family of circuits: the GordieHowe library proves it for the circuits of its signature-verification applications, open pull requests of leanerVM add it for a virtual machine, and the main chain of the Longfellow development starts from a trusted circuit. For sparse-Merkle updates, with the hash chains carried by the circuit, cube gates, the extension field of a small-field prover and a node hash that is not collision resistant, we work in a constraint-level model with abstract hash sub-circuits, for trees of every depth. We prove, on paper from a mechanized characterization of the constraints, that every accepted input determines a valid witness. The proof needs the chain constraint, which a statement about encoded witnesses cannot see: encoded witnesses satisfy it by construction, so that statement would hold for a circuit without it (\cref{rem:vector-constraints}). Composed with the protocol bound, this gives per-input soundness (\cref{sec:smt-vector}). The verifier there reads the input in the clear, which is the delegation setting of GKR; an argument that commits to the input needs a bound of this form at the input it extracts, and we do not treat the commitment. Our Isabelle/HOL development takes up the analysis of GKR that the authors of the Isabelle sumcheck development name as future work. \Cref{sec:gkr-faithful} also shows where a mechanized chain bound stops: instantiated with the multilinear representative of the layer integrand, it does not bound the acceptance of the verifier, which evaluates a different polynomial at its random point. The step to the verifier and the per-witness form, which carries security (\cref{sec:goal}), are proved on paper; the composition with the protocol bound is machine-checked, over one field, for claims that no witness makes valid. \Cref{sec:related} compares the results.

\paragraph{Contributions.} For sparse-Merkle membership, non-membership and update we prove per-input soundness of a GKR verifier, for a verifier that leaves the node hashes to the circuit and for one that holds the witness. The proofs rest on the first two arguments, which we state and prove, with the protocol bound instantiated over the exact extension field of a small-field prover. Of the third argument we prove the leaf side and state the rest as an assumption. Isabelle/HOL~\cite{isabellehol2002} checks the chain bound, its instance over the extension, the arithmetization theorem for encoded witnesses and, over one field, the composition for claims with no valid witness; the items below say which steps are proved on paper. The development accompanies the StateSync-GKR prover, whose engineering is reported separately~\cite{kim2026engine}; this paper reports no measurement.
\begin{enumerate}[leftmargin=*]
\item \textbf{A verifier that leaves the node hashes to the circuit} (\cref{sec:smt-vector}). The verifier receives the circuit input from the prover, checks its public entries and computes no node hash, so the protocol checks the hash chains. Every accepted input determines a valid witness, provided digests are represented by their own field elements, the node sub-circuit computes the compression function on every input and, for updates, the leaf sub-circuit computes the leaf hash of the old leaf (\cref{thm:a-vector}, a paper proof from a mechanized characterization of the constraints). If the witness read from the prover's first message is not valid, the verifier accepts with probability at most $\bigl(s_{\mathrm{out}} + \sum_i (8 s_i + 1)\bigr)/|\F|$, over the base field (\cref{cor:perinput-vector}) and over the extension (\cref{cor:perinput-vector-ext}); both are paper proofs. \Cref{rem:vector-constraints} records which constraints the proof uses and shows that the chain constraint cannot be dropped.
\item \textbf{Per-input soundness for a verifier that holds the witness} (\cref{sec:composition}). For every witness whose canonical form does not make the transition valid, the verifier accepts with probability at most the same bound, over the base field (\cref{cor:perinput-verifier}) and over the extension (\cref{cor:perinput-ext}). These are paper proofs built on the mechanized results; the composition itself is mechanized for claims that no witness makes valid (\cref{thm:c}), a form that is weak for a compressing hash (\cref{sec:goal}).
\item \textbf{One leaf per key and root} (\cref{sec:semantics,sec:one-leaf}). An injectivity lemma for the leaf encoding reduces leaf-hash collisions to sponge collisions (\cref{lem:leaffold}), and two different leaves with valid paths at the same key under one root yield a collision of the leaf or the node hash (\cref{prop:binding}, a paper step). The implemented node hash has cheap collisions, so we assume position binding of the tree (\cref{def:pb}) and bound the probability that two accepted runs assert different leaves for one key and root (\cref{cor:one-leaf}, a paper proof).
\item \textbf{GKR soundness with a cube gate} (\cref{sec:gkr}). For layered circuits with linear, product and cube gates, the cube gate serving the $x^3$ S-box of Poseidon2~\cite{poseidon2}, a false output claim survives the layer chain with probability at most $\bigl(s_{\mathrm{out}} + \sum_i (2 s_i d_i + 1)\bigr)/|\F|$, where layer $i$'s sumcheck runs over $2s_i$ variables with round polynomials of degree at most $d_i$; for the verifier analyzed here $d_i = 4$, because the layer integrand has degree at most $4$ in each variable (\cref{lem:integrand-poly}). The chain bound (\cref{thm:b}) is mechanized for every polynomial that agrees with the layer integrand on the Boolean hypercube and has total degree at most $d_i$; its application to the integrand that the verifier evaluates (\cref{lem:accept-bad,prop:b-faithful}) is proved on paper.
\item \textbf{Instantiation over the quartic extension} (\cref{sec:smallfield}). We construct $\F_p[X]/(X^4-3)$ for the KoalaBear prime $p = 2^{31}-2^{24}+1$ as a tower of two quadratic extensions, prove that it is a field with $p^4$ elements, and obtain the bound over it with denominator exactly $p^4$ (\cref{thm:b-ext4}). The embedding of the base field preserves layer evaluation, well-formedness and false outputs layer by layer (\cref{lem:lift-layer}); lifting a whole chain and composing across the two fields are proved on paper (\cref{prop:lift-chain,cor:perinput-ext}).
\item \textbf{Exact arithmetization relative to the witness encoding} (\cref{sec:semantics,sec:compiler}). In a constraint-level model of the compiler, with the hash sub-circuits as abstract functions, the verifier model, which performs public checks on the claim, encodes a given witness as the circuit input and evaluates the compiled checks on it, accepts exactly the valid transitions (\cref{thm:a-sound,thm:a-complete}).
\end{enumerate}
\Cref{sec:prelim} fixes the notation and states the standard results we use; every result is proved in the text or in the appendices, and the only statements used without proof are textbook facts of algebra and number theory, among them the Schwartz--Zippel lemma and quadratic reciprocity. The one cryptographic assumption, position binding (\cref{def:pb}), is stated as such and enters only \cref{cor:one-leaf}. \Cref{sec:mechanization} with \cref{app:mechanization} lists for each result the Isabelle theorem that checks it or marks it as a paper step.

\paragraph{Scope.} All probability statements are about the interactive public-coin protocol with challenges drawn uniformly and independently from the field. \Cref{sec:smt-vector} treats the setting in which the protocol carries the hash chains: the verifier receives the circuit input, reads it in the clear, compares its public entries with the claim and leaves the node hashes to the circuit. The verifier of \cref{sec:goal} receives the witness, computes the circuit input from it before any challenge is drawn, and holds that input in the clear. Such a verifier could check the authentication path by hashing it directly, so in that setting the protocol replaces no hash computation of the verifier, and no succinctness in the witness is claimed; the verifier that ships with the artifact is of this kind. \Cref{sec:limitations} states the boundary toward the non-interactive executable.

\paragraph{Organization.} \Cref{sec:goal} states the security goal. \Cref{sec:prelim} collects preliminaries. \Cref{sec:gkr} proves the protocol bound and \cref{sec:smallfield} instantiates it over the extension. \Cref{sec:semantics,sec:compiler} give the sparse-Merkle semantics, the leaf binding and the arithmetization, and \cref{sec:composition} composes them. \Cref{sec:mechanization} describes what is machine-checked, \cref{sec:limitations} the boundary, and \cref{sec:related} related work. The appendices contain the proofs that are deferred from the main text.

\section{Security Goal}
\label{sec:goal}

\paragraph{Claims and adversary.} A transition claim is a triple $(\mathit{op}, r, r')$: an operation on a sparse Merkle tree (membership of a key with a payload, non-membership of a key, or update of a key's leaf), the claimed root $r$, and, for updates, the claimed post-root $r'$; the claim also fixes a value digest $\mathit{vd}$, the leaf digest that the operation asserts (\cref{def:verifier-model}). The prover holds a witness $w$, made of the leaf data and the authentication path, and sends it with the claim. The verifier performs public checks on the claim, checks that the sibling list of $w$ has length $d$ and that the leaf of $w$ agrees with the claim (\cref{def:verifier-model}), encodes $w$ as the input of the compiled circuit of the claim (\cref{def:encoding}), and runs the GKR protocol with the prover to check that every output wire of that circuit is zero. The prover's strategy is arbitrary; probabilities are over the verifier's public-coin challenges.

\paragraph{The statement that carries security.} The existential goal, that a claim for which no witness makes the transition valid is accepted only with small probability, is weak for a compressing hash. An authentication path has far more free field elements than a root, so, heuristically, for a compression function that behaves like a random function almost every root admits some valid path for almost any leaf, and a false claim is then rarely invalid in the existential sense. The usual remedy is knowledge soundness. In this paper the verifier holds the circuit input in the clear, so the witness can be read from the prover's messages and no extractor is needed. The statement that carries security fixes the witness: for every witness whose canonical form does not make the transition valid, the verifier accepts with small probability (\cref{cor:perinput-verifier}); for the verifier of \cref{sec:smt-vector} it fixes the prover's first message (\cref{cor:perinput-vector,cor:perinput-vector-ext}). An accepted run therefore certifies, except with that probability, that the siblings of the witness form a valid authentication path for the leaf that the claim determines and for the claimed roots. Binding then rests on position binding of the tree (\cref{def:pb}): two different leaves for the same key cannot both be given valid paths under one root. Such a pair would yield a collision of the leaf hash or of the node hash (\cref{prop:binding}), but collision resistance is not available here, because the node hash of the implementation is a truncated permutation (\cref{sec:semantics}).

\paragraph{Decomposition.} The goal splits along the three arguments above (\cref{fig:roadmap}). The arithmetization theorem is deterministic: if the verifier model accepts, that is, the public-side checks and the witness echo hold and the compiled checks pass on the encoding of the witness, then the canonical form of the witness makes the transition valid (\cref{thm:a-sound}), and every valid witness is accepted with its canonical value digest (\cref{thm:a-complete}); \cref{thm:a-vector} gives the soundness direction for inputs supplied by the prover. The protocol bound (\cref{thm:b,prop:b-faithful}) is probabilistic and knows nothing about Merkle trees: for any well-formed layered circuit and any false output claim, the bad event is unlikely. The composition (\cref{thm:c,cor:perinput,cor:perinput-verifier,cor:perinput-vector}) transfers semantic invalidity to the bound, and \cref{sec:smallfield} carries the bound to the field from which the challenges are drawn.

\begin{figure}[t]
\centering
\begin{tikzpicture}[font=\small,
  arg/.style={draw, rounded corners=2pt, align=left, text width=5.6cm, inner sep=4pt, fill=black!3},
  res/.style={draw, rounded corners=2pt, align=left, text width=6.2cm, inner sep=4pt, fill=black!3},
  arr/.style={-{Stealth[length=2mm]}, semithick, black!65}]
\node[arg] (A) at (0,0) {\textbf{Protocol soundness} (\S\ref{sec:gkr})\\
  \Mmark\ chain bound, \cref{thm:b}\\
  \Pmark\ bound for the verifier, \cref{prop:b-faithful}};
\node[arg, below=2.5mm of A] (B) {\textbf{Small-field instance} (\S\ref{sec:smallfield})\\
  \Mmark\ chain bound over $\F_{p^4}$, \cref{thm:b-ext4}\\
  \Pmark\ lifting circuits over $\F_p$, \cref{prop:lift-chain}};
\node[arg, below=2.5mm of B] (C) {\textbf{Exact arithmetization} (\S\ref{sec:compiler}, \S\ref{sec:smt-vector})\\
  \Mmark\ verifier model, \cref{thm:a-sound,thm:a-complete}\\
  \Pmark\ inputs from the prover, \cref{thm:a-vector}};
\node[arg, below=2.5mm of C] (D) {\textbf{Leaf binding} (\S\ref{sec:semantics})\\
  \Mmark\ leaf encoding, \cref{lem:leaffold}\\
  \Pmark\ reduction to collisions, \cref{prop:binding}\\
  assumption: position binding, \cref{def:pb}};
\node[res, anchor=west] (F) at ($(B.east) + (2.0,0)$) {\textbf{Composition} (\S\ref{sec:composition})\\
  \Mmark\ no valid witness, \cref{thm:c}\\
  \Pmark\ per witness, over $\F$, \cref{cor:perinput-verifier}\\
  \Pmark\ per witness, over $\F_{p^4}$, \cref{cor:perinput-ext}\\
  \Pmark\ inputs from the prover, \cref{cor:perinput-vector,cor:perinput-vector-ext}};
\node[res, anchor=west, fill=white] (G) at ($(D.east) + (2.0,0)$) {\textbf{An accepted update} (\S\ref{sec:composition})\\
  \Pmark\ leaves authenticated, error $\le$ the bound\\
  \Pmark\ one leaf per key and root, under position binding, \cref{cor:one-leaf}};
\node[anchor=north west, font=\footnotesize] at ($(G.south west) + (0,-0.2)$) {\Mmark\ machine-checked\qquad \Pmark\ proved on paper};
\draw[arr] (A.east) -- ([yshift=6mm]F.west);
\draw[arr] (B.east) -- (F.west);
\draw[arr] (C.east) -- ([yshift=-6mm]F.west);
\draw[arr] (F.south) -- (G.north);
\draw[arr] (D.east) -- (G.west);
\end{tikzpicture}
\caption{How the results combine.}
\label{fig:roadmap}
\end{figure}
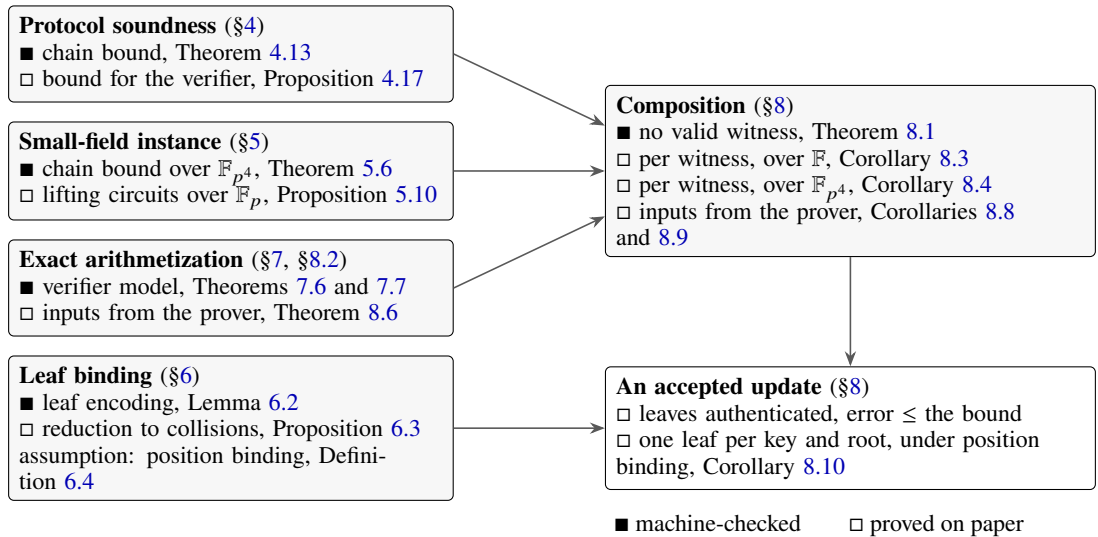

\section{Preliminaries}
\label{sec:prelim}

This section fixes notation, recalls multilinear extensions, the Schwartz--Zippel lemma and the sumcheck protocol in the form in which the imported formalization states them, and describes the probability space of every soundness statement in the paper.

\paragraph{Notation.} $\F$ is a finite field with $|\F|$ elements, and $[n] = \{0,\dots,n-1\}$. A point is a tuple $z = (z_0,\dots,z_{n-1}) \in \F^n$. Challenge tuples are written in bold: for $\mathbf{r} = (\mathbf{r}_0,\dots,\mathbf{r}_{N-1})$ we write $\mathbf{r}_{[a,b)} = (\mathbf{r}_a,\dots,\mathbf{r}_{b-1})$ and $\mathbf{r}_{\geq a} = (\mathbf{r}_a,\dots,\mathbf{r}_{N-1})$, and $\mathbf{u} \,\|\, \mathbf{v}$ is concatenation. A \emph{table} with $n$ index bits is a function $f\colon [2^n]\to\F$; the formalization stores tables as lists and reads them by index, and no statement below reads an index outside $[2^n]$. Index bits are read most significant first: for $a\in[2^n]$ and $j\in[n]$, $\mathrm{bit}_j(a)$ is bit $n-1-j$ of the binary expansion of $a$, and $\bpt{a}{n} = (\mathrm{bit}_0(a),\dots,\mathrm{bit}_{n-1}(a))$ is the Boolean point of $a$. The map $a\mapsto\bpt{a}{n}$ is a bijection from $[2^n]$ onto $\cube{n}$. For a polynomial $g$ over $\F$, $\deg g$ is its total degree and $\deg_X g$ its degree in the variable $X$.

\subsection{Multilinear extensions}

\begin{definition}[Equality kernel, multilinear extension]
\label{def:mle}
For $z\in\F^n$ and $a\in[2^n]$ let
\[
\eqk_n(z,a) \;=\; \prod_{j\in[n]} e_j,
\qquad
e_j = \begin{cases} z_j & \text{if } \mathrm{bit}_j(a) = 1,\\ 1 - z_j & \text{if } \mathrm{bit}_j(a) = 0.\end{cases}
\]
The multilinear extension of a table $f\colon[2^n]\to\F$ is $\widetilde f(z) = \sum_{a\in[2^n]} \eqk_n(z,a)\, f(a)$, and $\beta_n(z) = \eqk_n(z,0) = \prod_{j\in[n]}(1-z_j)$.
\end{definition}

\noindent Mechanized as \isaf{eq\_pi}, \isaf{mle} and \isaf{beta} in \texttt{Multilinear\_Extension.thy}. With indeterminates $Z_0,\dots,Z_{n-1}$ in place of $z$, each $\eqk_n(Z,a)$ is a product of $n$ factors, the $j$-th of degree at most one in $Z_j$ and free of the other variables. Hence $\widetilde f$ is a polynomial of degree at most one in each variable and of total degree at most $n$.

\begin{lemma}[Facts on multilinear extensions; mechanized]
\label{lem:mle-facts}
Let $f,g\colon[2^n]\to\F$, $c\in\F$ and $a,a'\in[2^n]$.
\begin{enumerate}[label=(\alph*),leftmargin=*,nosep]
\item $\eqk_n(\bpt{a}{n},a') = 1$ if $a = a'$, and $0$ otherwise.
\item $\widetilde f(\bpt{a}{n}) = f(a)$.
\item $\widetilde{f+g} = \widetilde f + \widetilde g$ and $\widetilde{c f} = c\,\widetilde f$.
\item $\beta_n(\bpt{a}{n}) = 1$ if $a = 0$, and $0$ otherwise; hence $\sum_{a\in[2^n]}\beta_n(\bpt{a}{n}) = 1$.
\item If $f$ has $n+1$ index bits, $p\in\F$ and $z'\in\F^n$, then $\widetilde f(p, z') = (1-p)\,\widetilde{f_0}(z') + p\,\widetilde{f_1}(z')$, where $f_0(a) = f(a)$ and $f_1(a) = f(a+2^n)$ for $a\in[2^n]$.
\end{enumerate}
\end{lemma}

\begin{proof}
(a) If $a\neq a'$, pick $j$ with $\mathrm{bit}_j(a)\neq\mathrm{bit}_j(a')$; the $j$-th factor of $\eqk_n(\bpt{a}{n},a')$ is $\mathrm{bit}_j(a) = 0$ if $\mathrm{bit}_j(a') = 1$ and $1-\mathrm{bit}_j(a) = 0$ otherwise. If $a = a'$, every factor equals $1$. (b) follows from (a), (c) from the definition, and (d) is (a) with $a' = 0$. (e) Split the sum defining $\widetilde f(p,z')$ into the indices $a<2^n$, whose first bit is $0$ and whose first factor is $1-p$, and the indices $a+2^n$ with $a<2^n$, whose first bit is $1$ and whose first factor is $p$; in both cases the remaining $n$ factors form $\eqk_n(z',a)$.
\end{proof}

\begin{lemma}[Schwartz--Zippel~\cite{schwartz1980,zippel1979}]
\label{lem:schwartz-zippel}
Let $g\in\F[X_1,\dots,X_m]$ be a nonzero polynomial of total degree at most $d$. Then $\Pr_{\mathbf{r}\leftarrow\F^m}\bigl[g(\mathbf{r}) = 0\bigr] \le d/|\F|$.
\end{lemma}

\noindent We use this standard lemma in paper proofs only. The mechanized proofs do not use it: they establish the multilinear case they need directly (Lemma~\ref{lem:seed}), and the univariate root bound behind the root axiom below is proved in the imported AFP instance.

\subsection{The sumcheck protocol}

\paragraph{Abstract polynomials.} The imported formalization of the sumcheck protocol~\cite{sumcheckafp2024,bosshard2024} is stated over an abstract structure of multivariate polynomials, and our development inherits it. Variables are natural numbers, and an \emph{assignment} $\sigma$ is a finite partial map from variables to $\F$; $[x\mapsto r]$ is the assignment with domain $\{x\}$, $\emptyset$ the empty one, and for $H\subseteq\F$ and a finite set $V$ of variables, $H^V$ is the set of assignments with domain $V$ and values in $H$. A structure of abstract polynomials is a set $\mathcal{P}$ with an addition and a zero, together with a finite variable set $\mathrm{vars}(p)$, a degree $\deg(p)\in\mathbb{N}$, an evaluation $p(\sigma)\in\F$ and an instantiation $p[\sigma]\in\mathcal{P}$ for each $p\in\mathcal{P}$, subject to the axioms of the locale \isaf{multi\_variate\_polynomial} (theory \texttt{Abstract\_Multivariate\_Polynomials} of~\cite{sumcheckafp2024}):
\begin{itemize}[leftmargin=*,nosep]
\item $\mathrm{vars}(0) = \emptyset$, $\mathrm{vars}(p+q)\subseteq\mathrm{vars}(p)\cup\mathrm{vars}(q)$, $\mathrm{vars}(p[\sigma])\subseteq\mathrm{vars}(p)\setminus\mathrm{dom}\,\sigma$;
\item $\deg(0) = 0$, $\deg(p+q)\le\max(\deg p,\deg q)$, $\deg(p[\sigma])\le\deg p$;
\item $0(\sigma) = 0$; $(p+q)(\sigma) = p(\sigma)+q(\sigma)$ when $\mathrm{vars}(p)\cup\mathrm{vars}(q)\subseteq\mathrm{dom}\,\sigma$; and $(p[\sigma])(\rho) = p(\rho \oplus \sigma)$ when $\mathrm{vars}(p)\subseteq\mathrm{dom}\,\sigma\cup\mathrm{dom}\,\rho$, where $\rho\oplus\sigma$ agrees with $\sigma$ on $\mathrm{dom}\,\sigma$ and with $\rho$ elsewhere;
\item (root bound) for every $d$, every variable $x$ and all $p\neq q$ with $\mathrm{vars}(p),\mathrm{vars}(q)\subseteq\{x\}$ and $\deg p,\deg q\le d$, at most $d$ elements $r\in\F$ satisfy $p([x\mapsto r]) = q([x\mapsto r])$.
\end{itemize}
In our development values and challenges lie in the same finite field $\F$ (locale \isaf{gkr\_layer\_sumcheck} in \texttt{Sumcheck\_Instance.thy}). The AFP entry provides a concrete instance: multivariate polynomials over $\F$ with $\deg$ the total degree, for which the root bound is the univariate root bound. Our development interprets its layer locale at this instance (\isaf{gkr\_mpoly}).

\begin{definition}[Sumcheck acceptance~\cite{sumcheckafp2024}]
\label{def:sumcheck}
A \emph{sumcheck prover} is a function $\mathrm{pr}$ that, given the current instance $(H,p,v)$, the current variable $x$, the list of remaining variables, the previous challenge $r'$ and a state $\mathit{st}$, returns a polynomial $q$ and a new state. For an instance $(H,p,v)$, where $H\subseteq\F$ is the summation domain, $p\in\mathcal{P}$ and $v\in\F$ is the claimed sum, the acceptance predicate $\cSC(\mathrm{pr},\mathit{st},(H,p,v),r',((x_1,r_1),\dots,(x_m,r_m)))$ is defined by recursion on $m$.
\begin{itemize}[leftmargin=*,nosep]
\item For $m = 0$ it holds iff $v = p(\emptyset)$.
\item For $m\geq 1$, let $(q,\mathit{st}') = \mathrm{pr}((H,p,v), x_1, (x_2,\dots,x_m), r', \mathit{st})$. It holds iff $\mathrm{vars}(q)\subseteq\{x_1\}$, $\deg q\le\deg p$, $v = \sum_{h\in H} q([x_1\mapsto h])$, and $\cSC(\mathrm{pr},\mathit{st}',(H,\,p[x_1\mapsto r_1],\,q([x_1\mapsto r_1])), r_1, ((x_2,r_2),\dots,(x_m,r_m)))$.
\end{itemize}
\end{definition}

\noindent This is the function \isaf{sumcheck} of the AFP theory \texttt{Sumcheck\_Protocol}. When $\mathrm{vars}(p)\subseteq\{x_1,\dots,x_m\}$, the final check compares the last round value with the value of $p$ at $(r_1,\dots,r_m)$, obtained by instantiating every variable. The round check $\deg q\le\deg p$ compares the prover's polynomial with the residual polynomial $p$ of the current round, under the abstract degree measure.

\begin{theorem}[Sumcheck soundness, AFP form~\cite{lfkn1992,bosshard2024,sumcheckafp2024}; imported; proof in \cref{app:gkr-proofs}]
\label{thm:sumcheck-afp}
Let $H\neq\emptyset$, let $\mathbf{x} = (x_1,\dots,x_m)$ be distinct variables with $\mathrm{vars}(p)\subseteq\{x_1,\dots,x_m\}$, and let $\deg p\le d$. For every sumcheck prover $\mathrm{pr}$, state $\mathit{st}$ and previous challenge $r'$,
\[
\Pr_{\mathbf{r}\leftarrow\F^m}\Bigl[\cSC\bigl(\mathrm{pr},\mathit{st},(H,p,v),r',((x_j,\mathbf{r}_{j-1}))_{j=1}^{m}\bigr) \;\wedge\; v\neq \textstyle\sum_{\sigma\in H^{\{x_1,\dots,x_m\}}} p(\sigma)\Bigr] \;\le\; \frac{m\,d}{|\F|}.
\]
\end{theorem}

\noindent Our development imports this theorem from the AFP entry~\cite{sumcheckafp2024} and does not re-prove it; \cref{app:gkr-proofs} gives a proof from the axioms above.

Our verifier checks round degrees against a fixed bound. For a concrete polynomial $g\in\F[X_0,\dots,X_{m-1}]$ we therefore also use the textbook variant of the protocol.

\begin{definition}[Sumcheck with round degree bounds]
\label{def:sumcheck-D}
Let $g\in\F[X_0,\dots,X_{m-1}]$, $v\in\F$ and $\mathbf{D} = (D_0,\dots,D_{m-1})\in\mathbb{N}^m$. A strategy $\psi$ maps each challenge prefix $\mathbf{r}_{[0,t)}$, $t\in[m]$, to a univariate polynomial $g_t = \psi(\mathbf{r}_{[0,t)})\in\F[X]$. Put $v_0 = v$ and $v_{t+1} = g_t(\mathbf{r}_t)$. Then $\cSC^{\mathbf{D}}(g,v,\psi,\mathbf{r})$ holds iff $\deg g_t\le D_t$ and $v_t = g_t(0)+g_t(1)$ for every $t\in[m]$, and $v_m = g(\mathbf{r}_0,\dots,\mathbf{r}_{m-1})$. For a constant bound $D_t = D$ we write $\cSC^{D}$.
\end{definition}

\noindent Definition~\ref{def:sumcheck-D} differs from Definition~\ref{def:sumcheck} with $H = \{0,1\}$ only in the round degree check. A sumcheck prover with a state can compute any function of the challenge prefix, since it receives the previous challenge in every round and may store it, so the two notions of strategy are interchangeable.

\begin{theorem}[Sumcheck soundness with per-variable degrees~\cite{lfkn1992,thaler2013}; standard; proof in \cref{app:gkr-proofs}]
\label{thm:sumcheck-std}
Let $g\in\F[X_0,\dots,X_{m-1}]$ with $\deg_{X_t} g\le D_t$ for every $t\in[m]$, and $v\neq\sum_{\mathbf{b}\in\cube{m}} g(\mathbf{b})$. For every strategy $\psi$, $\Pr_{\mathbf{r}\leftarrow\F^m}\bigl[\cSC^{\mathbf{D}}(g,v,\psi,\mathbf{r})\bigr]\le \sum_{t\in[m]} D_t/|\F|$.
\end{theorem}

\subsection{Probability space}
\label{sec:probspace}

Every soundness statement in this paper has the following form. The circuit, its input, the claimed output and the prover strategy are fixed first. The verifier's challenges are then $N$ field elements drawn uniformly and independently, and the statement bounds the probability, over the challenges alone, of an event determined by the strategy and the challenges. The prover strategy is arbitrary and adaptive: each of its messages may depend on all challenges drawn before it. We state the results for deterministic strategies. A randomized strategy whose coins are independent of the challenges is a probability mixture of deterministic ones, so every bound below holds for it by averaging. Following the imported development, the mechanized statements express the challenge distribution as the uniform distribution on the set $\F^N$ of $N$-tuples (\isaf{pmf\_of\_set} (\isaf{tuples} \isaf{UNIV} $N$)). Two elementary facts about this distribution carry all independence arguments of Section~\ref{sec:gkr}.

\begin{lemma}[Splitting uniform tuples; mechanized]
\label{lem:split}
Let $E\subseteq\F^{k+m}$.
\begin{enumerate}[label=(\alph*),leftmargin=*,nosep]
\item $\Pr_{\mathbf{r}\leftarrow\F^{k+m}}[\mathbf{r}\in E] = |\F|^{-k}\sum_{\mathbf{u}\in\F^k}\Pr_{\mathbf{v}\leftarrow\F^m}[\mathbf{u}\,\|\,\mathbf{v}\in E]$.
\item If $E = \{\mathbf{r} : \mathbf{r}_{[0,k)}\in E'\}$ for some $E'\subseteq\F^k$, then $\Pr_{\mathbf{r}\leftarrow\F^{k+m}}[\mathbf{r}\in E] = \Pr_{\mathbf{u}\leftarrow\F^k}[\mathbf{u}\in E']$.
\end{enumerate}
\end{lemma}

\begin{proof}
The map $(\mathbf{u},\mathbf{v})\mapsto\mathbf{u}\,\|\,\mathbf{v}$ is a bijection from $\F^k\times\F^m$ onto $\F^{k+m}$, so $|E| = \sum_{\mathbf{u}}|\{\mathbf{v} : \mathbf{u}\,\|\,\mathbf{v}\in E\}|$; dividing by $|\F|^{k+m}$ gives (a). In (b) the inner probability is $1$ if $\mathbf{u}\in E'$ and $0$ otherwise.
\end{proof}

\section{Soundness of the GKR Protocol}
\label{sec:gkr}

The GKR protocol~\cite{gkr2015,thaler2013} reduces a claim about the output of a layered arithmetic circuit to a claim about its input, one layer at a time, with one sumcheck per layer. This section defines the circuits (Section~\ref{sec:gkr-circuits}), proves the identity on which each layer's sumcheck rests (Section~\ref{sec:gkr-identity}), defines the verifier we analyze (Section~\ref{sec:gkr-verifier}) and the event that the mechanized theorem bounds (Section~\ref{sec:gkr-chain}), and proves that bound (Section~\ref{sec:gkr-thmb}). Section~\ref{sec:gkr-faithful} relates the event to the verifier and proves the verifier's soundness bound on paper. Complete proofs of the probabilistic steps are in Appendix~\ref{app:gkr-proofs}.

\subsection{Layered circuits}
\label{sec:gkr-circuits}

\begin{definition}[Layered circuit]
\label{def:circuit}
A \emph{gate} is a tuple $g = (\kappa,\mathrm{out},\mathrm{in}_1,\mathrm{in}_2,c)$ of a kind $\kappa\in\{\gLin,\gMul,\gPow\}$, wire indices $\mathrm{out},\mathrm{in}_1,\mathrm{in}_2\in\mathbb{N}$ and a coefficient $c\in\F$. The kinds act by $\mathrm{sem}(\gLin,a,a') = a$, $\mathrm{sem}(\gMul,a,a') = a a'$ and $\mathrm{sem}(\gPow,a,a') = a^3$. A \emph{layer} $L = (w,G,K)$ consists of a number $w = w(L)$ of index bits (the layer has $2^{w}$ wires), a finite list $G$ of gates and a finite list $K$ of constant pairs $(u,c)\in\mathbb{N}\times\F$. Its constant on wire $u$ is $\mathrm{const}_L(u) = \sum_{(u',c)\in K,\ u'=u} c$, and its evaluation on a table $b$ of the layer below is the table $\mathrm{eval}(L,b)\colon[2^{w}]\to\F$ given by
\[
\mathrm{eval}(L,b)(u) \;=\; \mathrm{const}_L(u) \;+\; \sum_{g\in G,\ \mathrm{out}(g) = u} c(g)\cdot\mathrm{sem}\bigl(\kappa(g),\, b(\mathrm{in}_1(g)),\, b(\mathrm{in}_2(g))\bigr).
\]
Both sums run over lists: a constant pair or a gate listed twice contributes twice. A \emph{circuit} consists of layers $L_0,\dots,L_{n-1}$ and a number $w_{\mathrm{in}}$ of input index bits. $L_0$ is the output layer, the gates of $L_i$ read the values of layer $i+1$, and the gates of $L_{n-1}$ read the input. For an input table $V_n$ the layer values are $V_i = \mathrm{eval}(L_i,V_{i+1})$ for $i = n-1,\dots,0$, and the output is $V_0$. A layer $L$ is \emph{well formed over} $B\in\mathbb{N}$ if every gate has $\mathrm{out}(g)<2^{w(L)}$, $\mathrm{in}_1(g),\mathrm{in}_2(g)<B$, and $\mathrm{in}_2(g) = \mathrm{in}_1(g)$ when $\kappa(g)\in\{\gLin,\gPow\}$, and every constant pair has $u<2^{w(L)}$. The circuit is well formed if $L_i$ is well formed over $2^{w(L_{i+1})}$ for $i<n-1$ and $L_{n-1}$ over $2^{w_{\mathrm{in}}}$.
\end{definition}

\noindent Mechanized as \isaf{gate\_sem}, \isaf{const\_at}, \isaf{layer\_eval}, \isaf{circuit\_values}, \isaf{wf\_layer} and \isaf{wf\_circuit} in \texttt{Layered\_Circuit.thy}. The formalization lists the values output first, $V_0,\dots,V_n$, and its recurrence is \isaf{circuit\_values\_step} in \texttt{GKR\_Assembly.thy}. The degree of a layer's output in the values below is at most three, which the cube gate attains.

The soundness theorem does not take a circuit as input. It takes the list of layers together with the value table that each layer reads.

\begin{definition}[Layer chain]
\label{def:chain}
A \emph{layer chain} is a nonempty list $\mathcal{C} = ((s_i,L_i,b_i))_{i\in[n]}$ of numbers $s_i$, layers $L_i$ and tables $b_i\colon[2^{s_i}]\to\F$ such that every gate of $L_i$ has $\mathrm{in}_1(g),\mathrm{in}_2(g)<2^{s_i}$, and, for $i<n-1$, $s_i = w(L_{i+1})$ and $b_i = \mathrm{eval}(L_{i+1},b_{i+1})$.
\end{definition}

\noindent Mechanized as \isaf{gkr\_chain\_wf} in \texttt{GKR\_Assembly.thy}, which also admits the empty chain and does not bound the length of the input table; the difference matters only for lifting (see before \cref{prop:lift-chain}). Thus $s_i$ is the number of index bits of the table that $L_i$ reads, namely layer $i+1$ for $i<n-1$ and the input for $i = n-1$, and the sumcheck of layer $i$ runs over $2s_i$ variables; we write $\sout = w(L_0)$ for the number of output index bits. For a well-formed circuit with input $V_n$, the triples $(s_i,L_i,V_{i+1})$ with $s_i = w(L_{i+1})$ for $i<n-1$ and $s_{n-1} = w_{\mathrm{in}}$ form a layer chain, directly from the definitions.

\subsection{The layer identity}
\label{sec:gkr-identity}

Fix a layer $L$ with $w = w(L)$ and a number $s$ such that every gate of $L$ reads below $2^s$. Its \emph{wiring tables} record, for $u\in[2^{w}]$ and $x,y\in[2^s]$, the total coefficient of each kind of gate:
\[
\mathrm{lin}_L(u,x) = \sum_{\substack{g\in G:\ \kappa(g)=\gLin,\\ \mathrm{out}(g)=u,\ \mathrm{in}_1(g)=x}} c(g),\quad
\mathrm{mul}_L(u,x,y) = \sum_{\substack{g\in G:\ \kappa(g)=\gMul,\ \mathrm{out}(g)=u,\\ \mathrm{in}_1(g)=x,\ \mathrm{in}_2(g)=y}} c(g),\quad
\mathrm{pow}_L(u,x) = \sum_{\substack{g\in G:\ \kappa(g)=\gPow,\\ \mathrm{out}(g)=u,\ \mathrm{in}_1(g)=x}} c(g),
\]
with sums over the list $G$. Their extensions in the output coordinate are, for $z\in\F^{w}$,
\[
\widetilde{\mathrm{lin}}_L(z,x) = \sum_{u\in[2^{w}]}\eqk_{w}(z,u)\,\mathrm{lin}_L(u,x),
\]
and likewise $\widetilde{\mathrm{mul}}_L(z,x,y)$, $\widetilde{\mathrm{pow}}_L(z,x)$ and $\widetilde{\mathrm{const}}_L(z) = \sum_{u}\eqk_{w}(z,u)\,\mathrm{const}_L(u)$. A \emph{claim list} is a list $C = ((z_1,c_1),\dots,(z_m,c_m))$ of points $z_k\in\F^{w}$ with coefficients $c_k\in\F$. For a table $b\colon[2^s]\to\F$ we define the \emph{cube integrand} and the \emph{true claim value}
\begin{align*}
I_{L,s,C,b}(x,y) &= \sum_{k=1}^{m} c_k\Bigl(\widetilde{\mathrm{lin}}_L(z_k,x)\,b(x)\,\beta_s(\bpt{y}{s}) + \widetilde{\mathrm{mul}}_L(z_k,x,y)\,b(x)\,b(y) + \widetilde{\mathrm{pow}}_L(z_k,x)\,b(x)^3\,\beta_s(\bpt{y}{s})\Bigr),\\
\ccs_L(C,b) &= \sum_{k=1}^{m} c_k\bigl(\widetilde{V}(z_k) - \widetilde{\mathrm{const}}_L(z_k)\bigr),\qquad V = \mathrm{eval}(L,b),
\end{align*}
for $x,y\in[2^s]$, where $\widetilde V$ is the $w$-variate extension. By Lemma~\ref{lem:mle-facts}(d), $\beta_s(\bpt{y}{s})$ is $1$ for $y = 0$ and $0$ otherwise. These objects are mechanized as \isaf{lin\_wt}, \isaf{mul\_wt}, \isaf{pow3\_wt}, \isaf{lin\_mle}, \isaf{mul\_mle}, \isaf{pow3\_mle}, \isaf{const\_mle}, \isaf{layer\_integrand} and \isaf{layer\_claim\_sum} in \texttt{Sumcheck\_Instance.thy}. If every claimed evaluation at $z_k$ were the true value $\widetilde V(z_k)$, the verifier's combined claim with the constants removed would be $\ccs_L(C,b)$.

\begin{lemma}[Layer identity; mechanized]
\label{lem:layer-identity}
If every gate of $L$ has $\mathrm{in}_1(g),\mathrm{in}_2(g)<2^s$, then for every claim list $C$ and table $b\colon[2^s]\to\F$,
\[
\sum_{x\in[2^s]}\sum_{y\in[2^s]} I_{L,s,C,b}(x,y) \;=\; \ccs_L(C,b).
\]
\end{lemma}

\begin{proof}
Both sides are linear combinations, with the coefficients $c_k$, of the same expressions for single claims $((z_k,1))$, so it suffices to treat $C = ((z,1))$.

\emph{Regrouping by kind.} For $u\in[2^{w}]$,
\[
\mathrm{eval}(L,b)(u) = \mathrm{const}_L(u) + \sum_{x}\mathrm{lin}_L(u,x)\,b(x) + \sum_{x,y}\mathrm{mul}_L(u,x,y)\,b(x)\,b(y) + \sum_{x}\mathrm{pow}_L(u,x)\,b(x)^3,
\]
with $x,y$ ranging over $[2^s]$. A $\gLin$ gate $g$ with $\mathrm{out}(g) = u$ contributes $c(g)\,b(\mathrm{in}_1(g)) = \sum_{x}[\mathrm{in}_1(g) = x]\,c(g)\,b(x)$, which uses $\mathrm{in}_1(g)<2^s$; summing over the gate list and exchanging the two finite sums collects the coefficients into $\mathrm{lin}_L(u,x)$. A $\gMul$ gate contributes $c(g)\,b(\mathrm{in}_1(g))\,b(\mathrm{in}_2(g))$ and is collected by the pair $(\mathrm{in}_1,\mathrm{in}_2)$ into $\mathrm{mul}_L$, and a $\gPow$ gate contributes $c(g)\,b(\mathrm{in}_1(g))^3$ and is collected into $\mathrm{pow}_L$.

\emph{Extension in the output coordinate.} Multiplying by $\eqk_{w}(z,u)$, summing over $u\in[2^{w}]$ and exchanging sums gives
\[
\widetilde V(z) = \widetilde{\mathrm{const}}_L(z) + \sum_{x}\widetilde{\mathrm{lin}}_L(z,x)\,b(x) + \sum_{x,y}\widetilde{\mathrm{mul}}_L(z,x,y)\,b(x)\,b(y) + \sum_{x}\widetilde{\mathrm{pow}}_L(z,x)\,b(x)^3.
\]

\emph{Lifting the unary terms.} Since $\sum_{y}\beta_s(\bpt{y}{s}) = 1$ (Lemma~\ref{lem:mle-facts}(d)), $\sum_x A(x) = \sum_{x}\sum_{y}A(x)\,\beta_s(\bpt{y}{s})$ for the $\gLin$ and $\gPow$ terms. Summing $I_{L,s,((z,1)),b}$ over the cube therefore gives $\widetilde V(z) - \widetilde{\mathrm{const}}_L(z) = \ccs_L(((z,1)),b)$.
\end{proof}

\noindent The identity is the completeness half of a layer reduction: when the incoming claim value equals $\ccs_L(C,b)$, the honest sumcheck prover is accepted for every challenge tuple (\isaf{gkr\_layer\_reduction\_complete} in \texttt{GKR\_Assembly.thy}).

\paragraph{The integrand as a polynomial.} The formalization defines the integrand only on the cube. The verifier's final check evaluates its natural extension to all of $\F^{2s}$, which we define next.

\begin{definition}[Integrand polynomial]
\label{def:integrand-poly}
With $L$, $s$, $C$ and $b$ as above, let $X = (X_0,\dots,X_{s-1})$ and $Y = (Y_0,\dots,Y_{s-1})$ be indeterminates, and for each claim $k$ let
\begin{gather*}
\Lambda_k(X) = \sum_{x}\eqk_s(X,x)\,\widetilde{\mathrm{lin}}_L(z_k,x),\qquad
\Pi_k(X) = \sum_{x}\eqk_s(X,x)\,\widetilde{\mathrm{pow}}_L(z_k,x),\\
M_k(X,Y) = \sum_{x,y}\eqk_s(X,x)\,\eqk_s(Y,y)\,\widetilde{\mathrm{mul}}_L(z_k,x,y).
\end{gather*}
The integrand polynomial is
\[
\widehat{P}_{L,s,C,b}(X,Y) = \sum_{k=1}^{m} c_k\Bigl(\Lambda_k(X)\,\widetilde b(X)\,\beta_s(Y) + M_k(X,Y)\,\widetilde b(X)\,\widetilde b(Y) + \Pi_k(X)\,\widetilde b(X)^3\,\beta_s(Y)\Bigr),
\]
where $\widetilde b$ is the $s$-variate extension of $b$. For $\xi,\eta\in\F^s$ and $t_x,t_y\in\F$ let
\[
\Phi_{L,s,C}(\xi,\eta,t_x,t_y) = \sum_{k=1}^{m} c_k\Bigl(\Lambda_k(\xi)\,t_x\,\beta_s(\eta) + M_k(\xi,\eta)\,t_x\,t_y + \Pi_k(\xi)\,t_x^3\,\beta_s(\eta)\Bigr).
\]
\end{definition}

\noindent The value $\Lambda_k(\xi) = \sum_{u,x}\eqk_{w}(z_k,u)\,\eqk_s(\xi,x)\,\mathrm{lin}_L(u,x)$ is the multilinear extension of the table $\mathrm{lin}_L$ in all its index bits, evaluated at $(z_k,\xi)$, and likewise for $M_k$ and $\Pi_k$. The verifier computes these wiring extensions from the circuit description, so $\Phi_{L,s,C}$ requires no knowledge of $b$, and by construction
\begin{equation}
\label{eq:phi-hatp}
\Phi_{L,s,C}\bigl(\xi,\eta,\widetilde b(\xi),\widetilde b(\eta)\bigr) = \widehat{P}_{L,s,C,b}(\xi,\eta).
\end{equation}

\begin{lemma}[Integrand polynomial; paper proof]
\label{lem:integrand-poly}
\begin{enumerate}[label=(\alph*),leftmargin=*,nosep]
\item $\widehat{P}_{L,s,C,b}(\bpt{x}{s},\bpt{y}{s}) = I_{L,s,C,b}(x,y)$ for all $x,y\in[2^s]$.
\item For every $j\in[s]$, $\deg_{X_j}\widehat{P}_{L,s,C,b}\le 4$ and $\deg_{Y_j}\widehat{P}_{L,s,C,b}\le 2$, and $\deg\widehat{P}_{L,s,C,b}\le 5s$.
\end{enumerate}
\end{lemma}

\begin{proof}
(a) By Lemma~\ref{lem:mle-facts}(a,b), $\Lambda_k(\bpt{x}{s}) = \widetilde{\mathrm{lin}}_L(z_k,x)$, $M_k(\bpt{x}{s},\bpt{y}{s}) = \widetilde{\mathrm{mul}}_L(z_k,x,y)$, $\Pi_k(\bpt{x}{s}) = \widetilde{\mathrm{pow}}_L(z_k,x)$ and $\widetilde b(\bpt{x}{s}) = b(x)$. (b) The polynomials $\Lambda_k$, $\Pi_k$ and $\widetilde b(X)$ have degree at most one in each $X_j$ and do not contain $Y$; $M_k$ has degree at most one in each $X_j$ and each $Y_j$; $\widetilde b(Y)$ and $\beta_s(Y)$ have degree at most one in each $Y_j$ and do not contain $X$. The degree in a fixed variable is subadditive under products and does not increase under sums. In $X_j$ the three terms have degree at most $1+1$, $1+1$ and $1+3$; in $Y_j$ at most $1$, $1+1$ and $1$. Each of $\Lambda_k$, $\Pi_k$, $\widetilde b(X)$, $\widetilde b(Y)$ and $\beta_s(Y)$ has total degree at most $s$ and $M_k$ at most $2s$, so the three terms have total degree at most $3s$, $4s$ and $5s$.
\end{proof}

\subsection{The verifier}
\label{sec:gkr-verifier}

In the sumcheck of a layer with $s$ index bits, variables $0,\dots,s-1$ stand for $X_0,\dots,X_{s-1}$ and variables $s,\dots,2s-1$ for $Y_0,\dots,Y_{s-1}$. A challenge tuple $\mathbf{r}\in\F^{2s}$ thus determines the points $\xi = \mathbf{r}_{[0,s)}$ and $\eta = \mathbf{r}_{[s,2s)}$, and the Boolean assignment $\sigma_{x,y}$ of index pair $(x,y)\in[2^s]^2$ gives variable $j$ the value $\mathrm{bit}_j(x)$ and variable $s+j$ the value $\mathrm{bit}_j(y)$ (\isaf{idx\_subst} in \texttt{Sumcheck\_Instance.thy}).

\begin{definition}[Verifier]
\label{def:verifier}
The verifier $\mathcal{V}$ is given the layers $L_0,\dots,L_{n-1}$ with the numbers $s_0,\dots,s_{n-1}$, the input table $b_{n-1}\colon[2^{s_{n-1}}]\to\F$ in the clear, and a claimed output table $o\colon[2^{\sout}]\to\F$ with $\sout = w(L_0)$. It interacts with a prover as follows.
\begin{enumerate}[label=\arabic*.,leftmargin=*]
\item \emph{Output seed.} Draw $z\leftarrow\F^{\sout}$, set $C_0 = ((z,1))$ and $v_0 = \widetilde o(z) - \widetilde{\mathrm{const}}_{L_0}(z)$, and run the layer reduction at layer $0$.
\item \emph{Layer reduction at layer $i$ with claim list $C_i$ and value $v_i$.} Write $s = s_i$ and $L = L_i$.
\begin{enumerate}[label=(\alph*),leftmargin=*]
\item For $t = 0,\dots,2s-1$: receive a univariate polynomial $g_t$, reject unless $\deg g_t\le 4$ and $g_t(0)+g_t(1) = u_t$, where $u_0 = v_i$; then draw $\mathbf{r}_t\leftarrow\F$ and set $u_{t+1} = g_t(\mathbf{r}_t)$. Let $\xi = \mathbf{r}_{[0,s)}$ and $\eta = \mathbf{r}_{[s,2s)}$.
\item Receive $t_x,t_y\in\F$, the claimed values of $\widetilde{b_i}$ at $\xi$ and $\eta$, and reject unless $u_{2s} = \Phi_{L,s,C_i}(\xi,\eta,t_x,t_y)$.
\item Draw $\gamma\leftarrow\F$. If $i<n-1$, set $C_{i+1} = ((\xi,1),(\eta,\gamma))$ and
\[
v_{i+1} = \bigl(t_x - \widetilde{\mathrm{const}}_{L_{i+1}}(\xi)\bigr) + \gamma\bigl(t_y - \widetilde{\mathrm{const}}_{L_{i+1}}(\eta)\bigr),
\]
and run the layer reduction at layer $i+1$.
\item If $i = n-1$, accept iff $t_x = \widetilde{b_{n-1}}(\xi)$ and $t_y = \widetilde{b_{n-1}}(\eta)$.
\end{enumerate}
\end{enumerate}
\end{definition}

\noindent The verifier draws $N = \sout + \sum_{i\in[n]}(2s_i+1)$ challenges, in the order $z$, then for each layer the $2s_i$ sumcheck challenges and $\gamma$; at layer $n-1$ the challenge $\gamma$ is drawn and not used, which matches the randomness layout of the formalization (\isaf{chain\_rlen}). The prover sees every challenge as soon as it is drawn. The verifier holds the input table and evaluates its extension itself. \Cref{fig:chain} shows the steps of $\mathcal{V}$ and the term of the bound that each contributes.

\begin{figure}[b]
\centering
\begin{tikzpicture}[font=\small,
  st/.style={draw, rounded corners=2pt, align=center, inner sep=4pt, minimum height=10mm, fill=black!3},
  arr/.style={-{Stealth[length=2mm]}, semithick},
  err/.style={font=\footnotesize, text=black!75}]
\node[st] (o) {claimed output\\ $o$ on $2^{\sout}$ wires};
\node[st, right=15mm of o] (c) {claim $v_i$\\ about layer $i$};
\node[st, right=25mm of c] (xy) {values of $\widetilde{b_i}$\\ at $\xi$ and $\eta$};
\node[st, right=15mm of xy] (in) {input table\\ $b_{n-1}$};
\draw[arr] (o) -- node[above] {seed $z$} node[below, err] {$\sout/|\F|$} (c);
\draw[arr] (c) -- node[above, align=center] {sumcheck,\\ $2s_i$ rounds} node[below, err] {$8s_i/|\F|$} (xy);
\draw[arr] (xy.south) .. controls +(0,-1.0) and +(0,-1.0) .. node[below, align=center] {combine with $\gamma$, then $i\leftarrow i+1$\\ \textcolor{black!75}{\footnotesize $1/|\F|$}} (c.south);
\draw[arr] (xy) -- node[above] {$i = n-1$} node[below, err] {no error} (in);
\end{tikzpicture}
\caption{The steps of the verifier and the term of the bound that each contributes.}
\label{fig:chain}
\end{figure}
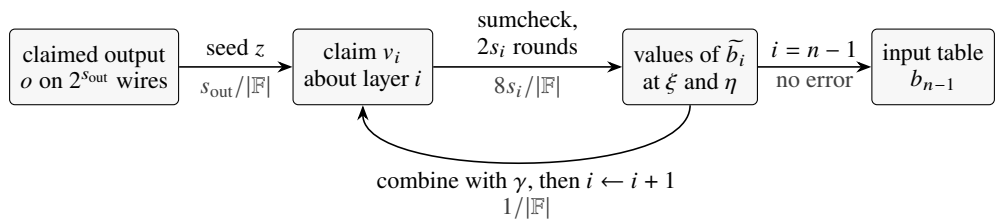

\paragraph{Relation to the formalization.} The development does not define $\mathcal{V}$ as a program. \texttt{GKR\_Assembly.thy} defines the event of Definition~\ref{def:chainbad} directly over an adversary, with the randomness layout and the claim update of $\mathcal{V}$; the final input check of $\mathcal{V}$ appears in the event as the absence of a continuation at the last layer. The difference lies in each layer's acceptance. The event uses the acceptance predicate of Definition~\ref{def:sumcheck} for a \emph{representative polynomial}: its round check is $\deg q\le\deg p$ for the residual representative under the abstract degree, and its final check evaluates the representative at $(\xi,\eta)$. $\mathcal{V}$ checks $\deg g_t\le 4$ and evaluates $\Phi$. Section~\ref{sec:gkr-faithful} treats the consequences.

\subsection{The chain event}
\label{sec:gkr-chain}

\begin{definition}[Representative]
\label{def:repr}
Let $L$, $s$, $C$, $b$ be as in Section~\ref{sec:gkr-identity} and $d\in\mathbb{N}$. An abstract polynomial $p\in\mathcal{P}$ \emph{represents} the layer instance with degree bound $d$, written $\cRepr(p,L,s,C,b,d)$, if $\mathrm{vars}(p)\subseteq[2s]$, $\deg p\le d$, and $p(\sigma_{x,y}) = I_{L,s,C,b}(x,y)$ for all $x,y\in[2^s]$. A \emph{representative family} $P$ assigns an abstract polynomial $P(L,s,C,b)$ to every layer, number, claim list and table.
\end{definition}

\noindent Mechanized as \isaf{layer\_poly\_repr} in \texttt{Sumcheck\_Instance.thy}. Only agreement on the cube is required; the values of a representative elsewhere are unconstrained.

\begin{lemma}[Layer sumcheck soundness; mechanized]
\label{lem:layer-sumcheck}
Assume $\cRepr(p,L,s,C,b,d)$, that every gate of $L$ reads below $2^s$, and $v\neq\ccs_L(C,b)$. For every sumcheck prover $\mathrm{pr}$, state $\mathit{st}$ and previous challenge $r'$,
\[
\Pr_{\mathbf{r}\leftarrow\F^{2s}}\Bigl[\cSC\bigl(\mathrm{pr},\mathit{st},(\{0,1\},p,v),r',((j,\mathbf{r}_j))_{j\in[2s]}\bigr)\Bigr] \le \frac{2s\,d}{|\F|}.
\]
\end{lemma}

\begin{proof}
The assignments in $\{0,1\}^{[2s]}$ are exactly the $\sigma_{x,y}$ with $x,y\in[2^s]$, each obtained once (\isaf{sum\_substs\_index}). Hence $\sum_{\sigma\in\{0,1\}^{[2s]}}p(\sigma) = \sum_{x,y}p(\sigma_{x,y}) = \sum_{x,y}I_{L,s,C,b}(x,y) = \ccs_L(C,b)$ by $\cRepr$ and Lemma~\ref{lem:layer-identity}. So $v$ differs from the sum of $p$ over the cube, the acceptance event coincides with the event of Theorem~\ref{thm:sumcheck-afp} for $H = \{0,1\}$, the variables $0,\dots,2s-1$ and the degree bound $d$, and that theorem gives the bound.
\end{proof}

\begin{definition}[Chain event]
\label{def:chainbad}
An \emph{adversary} $A$ maps a layer index $j\in\mathbb{N}$ and the tuple $\pi$ of challenges drawn so far to a quintuple $A(j,\pi) = (\mathrm{pr},\mathit{st},r',t_x,t_y)$: a sumcheck prover with its initial state and previous challenge, and two carried values. Fix a representative family $P$. For a layer chain $\mathcal{C}$ with first element $(s,L,b)$ and remaining list $\mathcal{C}'$, a claim list $C$, a value $v$ and a challenge tuple $\mathbf{r}$, the event $\cBad_P(A,\mathcal{C},C,v,\mathbf{r})$ holds iff (i) or (ii) holds:
\begin{enumerate}[label=(\roman*),leftmargin=*]
\item $\cSC\bigl(\mathrm{pr},\mathit{st},(\{0,1\},P(L,s,C,b),v),r',((j,\mathbf{r}_j))_{j\in[2s]}\bigr)$, where $(\mathrm{pr},\mathit{st},r',\cdot,\cdot) = A(0,())$;
\item $\mathcal{C}'$ is nonempty with first layer $L'$, and, with $(\cdot,\cdot,\cdot,t_x,t_y) = A(0,\mathbf{r}_{[0,2s)})$, $\xi = \mathbf{r}_{[0,s)}$, $\eta = \mathbf{r}_{[s,2s)}$ and $\gamma = \mathbf{r}_{2s}$,
\[
\bigl(t_x\neq\widetilde b(\xi)\ \text{or}\ t_y\neq\widetilde b(\eta)\bigr)\quad\text{and}\quad \cBad_P\bigl(A',\mathcal{C}',C',v',\mathbf{r}_{\geq 2s+1}\bigr),
\]
where $\widetilde b$ is the $s$-variate extension of $b$, $C' = ((\xi,1),(\eta,\gamma))$, $v' = (t_x - \widetilde{\mathrm{const}}_{L'}(\xi)) + \gamma\,(t_y - \widetilde{\mathrm{const}}_{L'}(\eta))$ and $A'(j,\pi) = A(j+1,\mathbf{r}_{[0,2s+1)}\,\|\,\pi)$.
\end{enumerate}
On the empty list, $\cBad_P$ is false. For a claimed output $o\colon[2^{\sout}]\to\F$, $\sout = w(L_0)$, the \emph{seeded chain event} is
\[
\cBad^{\mathrm{out}}_P(A,\mathcal{C},o,\mathbf{r}) \;=\; \cBad_P\bigl(A_z,\mathcal{C},((z,1)),\,\widetilde o(z) - \widetilde{\mathrm{const}}_{L_0}(z),\,\mathbf{r}_{\geq\sout}\bigr),\qquad z = \mathbf{r}_{[0,\sout)},\ A_z(j,\pi) = A(j,z\,\|\,\pi).
\]
\end{definition}

\noindent Mechanized as \isaf{gkr\_chain\_bad} in \texttt{GKR\_Assembly.thy}; the seeded event is the event of \isaf{gkr\_assembly\_soundness}. Disjunct (i) is acceptance, by the sumcheck verifier of Definition~\ref{def:sumcheck}, of the incoming claim against the representative $P(L_i,s_i,C_i,b_i)$, so its final check compares the last round value with the value of the representative at $(\xi,\eta)$. Disjunct (ii) says that the carried values are false and that the resulting claim fails further down. The claim list $C'$ and value $v'$ are those of step 2(c) of Definition~\ref{def:verifier}, and at the last layer only (i) remains because $\mathcal{V}$ checks the carried values against the input directly. The adversary supplies the sumcheck prover at the entry of a layer and the carried values after the layer's $2s$ sumcheck challenges, and the shift $A'$ passes all earlier challenges to later layers. Since a sumcheck prover may store every challenge it receives in its state, every adaptive deterministic prover strategy is of this form.

\begin{lemma}[Two-point carry; mechanized]
\label{lem:carry}
Let $L'$ be a layer, $b'$ a table, $V = \mathrm{eval}(L',b')$ with $w(L')$-variate extension $\widetilde V$, $\xi,\eta\in\F^{w(L')}$, and $t_x,t_y\in\F$ with $t_x\neq\widetilde V(\xi)$ or $t_y\neq\widetilde V(\eta)$. Then
\[
\Pr_{\gamma\leftarrow\F}\Bigl[\bigl(t_x - \widetilde{\mathrm{const}}_{L'}(\xi)\bigr) + \gamma\bigl(t_y - \widetilde{\mathrm{const}}_{L'}(\eta)\bigr) = \ccs_{L'}\bigl(((\xi,1),(\eta,\gamma)),b'\bigr)\Bigr] \le \frac{1}{|\F|}.
\]
\end{lemma}

\begin{proof}
By definition $\ccs_{L'}(((\xi,1),(\eta,\gamma)),b') = (\widetilde V(\xi) - \widetilde{\mathrm{const}}_{L'}(\xi)) + \gamma(\widetilde V(\eta) - \widetilde{\mathrm{const}}_{L'}(\eta))$, so the event is $\alpha + \gamma\delta = 0$ with $\alpha = t_x - \widetilde V(\xi)$ and $\delta = t_y - \widetilde V(\eta)$, not both zero. If $\delta\neq 0$ exactly one $\gamma$ satisfies it; if $\delta = 0$ then $\alpha\neq 0$ and none does.
\end{proof}

\begin{lemma}[Output seed; mechanized]
\label{lem:seed}
Let $f,g\colon[2^n]\to\F$ with $f(a)\neq g(a)$ for some $a\in[2^n]$. Then $\Pr_{z\leftarrow\F^n}\bigl[\widetilde f(z) = \widetilde g(z)\bigr]\le n/|\F|$.
\end{lemma}

\begin{proof}
Let $h = f-g$. By Lemma~\ref{lem:mle-facts}(c), $\widetilde f - \widetilde g = \widetilde h$, a polynomial of total degree at most $n$, and it is nonzero because $\widetilde h(\bpt{a}{n}) = h(a)\neq 0$ by Lemma~\ref{lem:mle-facts}(b). Lemma~\ref{lem:schwartz-zippel} gives the bound. The mechanized proof avoids the general lemma: by Lemma~\ref{lem:mle-facts}(e), fixing the first coordinate $p$ turns $\widetilde h$ into the extension of the table $(1-p)h_0 + p h_1$, which is identically zero for at most one value of $p$, and induction on $n$ bounds the remaining event.
\end{proof}

\subsection{The chain bound}
\label{sec:gkr-thmb}

The chain event is bounded by induction along the chain, starting from a false incoming claim.

\begin{lemma}[Chain bound; mechanized]
\label{lem:chain-bound}
Let $\mathcal{C} = ((s_i,L_i,b_i))_{i\in[n]}$ be a layer chain, $P$ and $\dbnd\colon\mathbb{N}\to\mathbb{N}$ such that $\cRepr(P(L,s,C',b),L,s,C',b,\dbnd(s))$ for all $L,s,C',b$, and $C$ a claim list with $v\neq\ccs_{L_0}(C,b_0)$. For every adversary $A$, with $M = \sum_{i\in[n]}(2s_i+1)$,
\[
\Pr_{\mathbf{r}\leftarrow\F^{M}}\bigl[\cBad_P(A,\mathcal{C},C,v,\mathbf{r})\bigr] \le \frac{1}{|\F|}\sum_{i\in[n]}\bigl(2s_i\,\dbnd(s_i) + 1\bigr).
\]
\end{lemma}

\begin{proof}[Proof sketch]
Induction on $n$; the complete proof is in Appendix~\ref{app:gkr-proofs}. Write $(s,L,b)$ for the first element. Disjunct (i) depends only on $\mathbf{r}_{[0,2s)}$, so by Lemma~\ref{lem:split}(b) and Lemma~\ref{lem:layer-sumcheck} (with the gate bound and $b$ from the chain conditions) it has probability at most $2s\,\dbnd(s)/|\F|$. For $n = 1$ this is the whole event. For $n\geq 2$, an outcome in disjunct (ii) either makes the next claim value true, $v' = \ccs_{L'}(C',b')$, or keeps it false. The first case has probability at most $1/|\F|$: the carried values and the points $\xi,\eta$ are fixed by $\mathbf{r}_{[0,2s)}$, the chain conditions give $b = \mathrm{eval}(L',b')$ and $s = w(L')$, and the fresh coordinate $\gamma = \mathbf{r}_{2s}$ is then uniform, so Lemma~\ref{lem:carry} applies for each fixed prefix and Lemma~\ref{lem:split}(a) averages. In the second case, for each fixed prefix $\mathbf{r}_{[0,2s+1)}$ the remaining event is $\cBad_P$ for the shorter chain, a false incoming claim and the fixed adversary $A'$, over the fresh tuple $\mathbf{r}_{\geq 2s+1}$; the induction hypothesis bounds it, and Lemma~\ref{lem:split}(a) averages. The union bound adds $2s\,\dbnd(s)/|\F|$, $1/|\F|$ and the bound for the rest of the chain.
\end{proof}

\begin{theorem}[Chain soundness; mechanized]
\label{thm:b}
Fix a structure of abstract polynomials over $\F$. Let $\mathcal{C} = ((s_i,L_i,b_i))_{i\in[n]}$ be a layer chain, $\sout = w(L_0)$ and $N = \sout + \sum_{i\in[n]}(2s_i+1)$. Let $P$ be a representative family and $\dbnd\colon\mathbb{N}\to\mathbb{N}$ a function such that $\cRepr(P(L,s,C,b),L,s,C,b,\dbnd(s))$ holds for all $L$, $s$, $C$ and $b$. Let $o\colon[2^{\sout}]\to\F$ be a claimed output with $o(u)\neq\mathrm{eval}(L_0,b_0)(u)$ for some $u\in[2^{\sout}]$. Then for every adversary $A$,
\[
\Pr_{\mathbf{r}\leftarrow\F^{N}}\bigl[\cBad^{\mathrm{out}}_P(A,\mathcal{C},o,\mathbf{r})\bigr] \;\le\; \frac{\sout + \sum_{i\in[n]}\bigl(2s_i\,\dbnd(s_i)+1\bigr)}{|\F|}.
\]
\end{theorem}

\noindent As in the formal statement, the premise on $P$ quantifies over all layer instances; the proof uses it only for the instances that occur in the event. For a circuit chain, $\mathrm{eval}(L_0,b_0) = V_0$ is the true output.

\begin{proof}
Let $z = \mathbf{r}_{[0,\sout)}$ and $V_0 = \mathrm{eval}(L_0,b_0)$. Since $\ccs_{L_0}(((z,1)),b_0) = \widetilde{V_0}(z) - \widetilde{\mathrm{const}}_{L_0}(z)$, the seeded claim value $\widetilde o(z) - \widetilde{\mathrm{const}}_{L_0}(z)$ is false unless $\widetilde o(z) = \widetilde{V_0}(z)$. Hence the seeded event is contained in $\mathrm{SEED}\cup\mathrm{REST}$, where $\mathrm{SEED}$ is the event $\widetilde o(z) = \widetilde{V_0}(z)$ and $\mathrm{REST}$ is the event that the seeded claim value is false and the seeded chain event holds. $\mathrm{SEED}$ depends only on $z$, so by Lemma~\ref{lem:split}(b) and Lemma~\ref{lem:seed} its probability is at most $\sout/|\F|$. For each fixed $z$, $\mathrm{REST}$ is the event of Lemma~\ref{lem:chain-bound} for the claim list $((z,1))$, a false value and the fixed adversary $A_z$, over the fresh tuple $\mathbf{r}_{\geq\sout}$; by Lemma~\ref{lem:split}(a) and Lemma~\ref{lem:chain-bound} its probability is at most $\sum_i(2s_i\,\dbnd(s_i)+1)/|\F|$. The union bound gives the claim.
\end{proof}

\noindent Each term of the bound has one source. The output seed contributes $\sout/|\F|$ through Lemma~\ref{lem:seed}. Layer $i$ contributes $2s_i\,\dbnd(s_i)/|\F|$, the sumcheck bound for its representative (Lemma~\ref{lem:layer-sumcheck}), and $1/|\F|$, the carry collision of Lemma~\ref{lem:carry}; at the last layer there is no carry, and the formal statement keeps the $1/|\F|$ as padding.

\begin{remark}[A degree-four specialization]
\label{cor:b-deg4}\label{rem:deg4-vacuous}
The development also states Theorem~\ref{thm:b} with $\dbnd$ the constant function $4$ (\isaf{gkr\_assembly\_soundness\_deg4} in \texttt{GKR\_Assembly.thy}), which displays the bound $\bigl(\sout + \sum_{i\in[n]}(8s_i+1)\bigr)/|\F|$. Its premise asks for a representative of abstract degree at most $4$ for every layer instance, and no representative family satisfies it when the degree is the total degree, as it is in every instance of the development. Take $s = 3$, a layer $L$ with $w(L) = 0$ whose only gate is $(\gLin,0,0,0,1)$ and which has no constants, the claim list $C = (((),1))$ and a table $b$ with $b(0) = 1$. Then $\widetilde{\mathrm{lin}}_L((),x) = [x = 0]$, so $I_{L,3,C,b}(x,y) = 1$ if $x = y = 0$ and $0$ otherwise. If $p$ is a polynomial in six variables with $p(\sigma_{x,y}) = I_{L,3,C,b}(x,y)$ on the cube, replacing every power $X^e$, $e\geq 1$, by $X$ changes no value on $\cube{6}$ and does not increase the total degree, and it yields the unique multilinear polynomial with these values, $\prod_{j\in[6]}(1-X_j)$, of total degree $6$. Hence the specialization holds only vacuously in the development, and we do not use it. The bound it displays is established for the verifier of Definition~\ref{def:verifier} by Proposition~\ref{prop:b-faithful}, on paper.
\end{remark}

\subsection{The representative and the verifier}
\label{sec:gkr-faithful}

Theorem~\ref{thm:b} bounds the chain event for any representative family that satisfies its premise. Whether that bounds the probability that $\mathcal{V}$ accepts depends on the representative. Disjunct (i) evaluates the representative at the random point $(\xi,\eta)$, whereas $\mathcal{V}$ evaluates $\Phi$, which by \eqref{eq:phi-hatp} equals $\widehat{P}(\xi,\eta)$ when the carried values are true.

In the total-degree instance the development discharges the premise of Theorem~\ref{thm:b} with the multilinear interpolation $\mathrm{ML}(I_{L,s,C,b})$ of the cube integrand, the unique polynomial of degree at most one in each of the $2s$ variables that agrees with $I_{L,s,C,b}$ on the cube, whose total degree is at most $2s$ (\isaf{table\_repr} and \isaf{layer\_repr\_mpoly} in \texttt{Layer\_Representative.thy}). Every instance of the layer locale in the development uses the total degree. Since the premise is discharged by a mechanized lemma, Theorem~\ref{thm:b} applies with $\dbnd(s) = 2s$ and the per-layer term $4s_i^2+1$. The event it then bounds concerns $\mathrm{ML}(I)$. The polynomials $\mathrm{ML}(I)$ and $\widehat{P}$ agree on the cube, but $\widehat{P}$ has degree up to four in each $X_j$. When they differ, their difference has total degree at most $5s$ (Lemma~\ref{lem:integrand-poly}(b)), so by Lemma~\ref{lem:schwartz-zippel} they agree at a uniformly random $(\xi,\eta)$ with probability at most $5s/|\F|$. An accepting run of $\mathcal{V}$, honest runs included, therefore need not satisfy disjunct (i) for $\mathrm{ML}(I)$, and Theorem~\ref{thm:b} with this representative does not bound the acceptance probability of $\mathcal{V}$. The mechanized theorem \isaf{verifier\_\allowbreak acceptance\_\allowbreak implies\_\allowbreak gkr\_chain\_bad} (\texttt{Verifier\_\allowbreak Acceptance\_\allowbreak Refinement.thy}) derives the chain event from an accepting run under a per-layer premise: when the carried values are true, the run's round checks and final reconstruction constitute acceptance in the sense of Definition~\ref{def:sumcheck} for the representative. With $\mathrm{ML}(I)$ and the checks of $\mathcal{V}$, that premise fails on every run that passes the checks of $\mathcal{V}$ with true carried values at a point where $\mathrm{ML}(I)(\xi,\eta)\neq\widehat{P}(\xi,\eta)$, honest runs included.

On paper, the family $\widehat{P}$ itself satisfies the premise of Theorem~\ref{thm:b} in the total-degree instance with $\dbnd(s) = 5s$, by Lemma~\ref{lem:integrand-poly}, with $X_j$ and $Y_j$ read as the variables $j$ and $s+j$. Disjunct (i) then compares the final round value with $\widehat{P}(\xi,\eta)$, as $\mathcal{V}$ does, but it also requires in every round that the degree of $g_t$ be at most the total degree of the residual of $\widehat{P}$, which $\mathcal{V}$ does not check; an accepting run of $\mathcal{V}$ need not satisfy it. We therefore close the gap on paper with the representative $\widehat{P}$ and the round check of $\mathcal{V}$. The matching layer acceptance is $\cSC^{4}$ of Definition~\ref{def:sumcheck-D}, with the variable order of Section~\ref{sec:gkr-verifier}.

\begin{definition}[Faithful chain event]
\label{def:chainbad-star}
A \emph{faithful adversary} $A$ maps $(j,\pi)$ to a triple $(\varpi,t_x,t_y)$ of a strategy in the sense of Definition~\ref{def:sumcheck-D} and two carried values. The event $\cBad^{\star}(A,\mathcal{C},C,v,\mathbf{r})$ is defined by the recursion of Definition~\ref{def:chainbad}, with disjunct (i) replaced by
\begin{enumerate}[label=(\roman*$^{\star}$),leftmargin=*]
\item $\cSC^{4}\bigl(\widehat{P}_{L,s,C,b},\,v,\,\varpi,\,\mathbf{r}_{[0,2s)}\bigr)$, where $(\varpi,\cdot,\cdot) = A(0,())$,
\end{enumerate}
and with the carried values of (ii) read from $A(0,\mathbf{r}_{[0,2s)})$. The seeded event $\cBad^{\star,\mathrm{out}}$ is obtained from $\cBad^{\star}$ as $\cBad^{\mathrm{out}}_P$ from $\cBad_P$.
\end{definition}

A deterministic prover strategy $\mathcal{P}^{*}$ for $\mathcal{V}$ maps the challenges received so far to its next message. It \emph{induces} the faithful adversary $A^{*}$ with $A^{*}(j,\pi) = (\varpi,t_x,t_y)$, where $\varpi(\mathbf{u})$ is the round polynomial that $\mathcal{P}^{*}$ sends after the challenges $\pi\,\|\,\mathbf{u}$ and $(t_x,t_y)$ are the carried values it sends after the challenges $\pi$.

\begin{lemma}[Acceptance implies the faithful chain event; paper proof]
\label{lem:accept-bad}
Let $\mathcal{C} = ((s_i,L_i,b_i))_{i\in[n]}$ be a layer chain, $o\colon[2^{\sout}]\to\F$, and $\mathcal{P}^{*}$ a deterministic prover strategy for $\mathcal{V}$ run on $(L_i,s_i)_{i\in[n]}$, $b_{n-1}$ and $o$, with induced adversary $A^{*}$. For every $\mathbf{r}\in\F^{N}$, if $\mathcal{V}$ accepts with the challenges $\mathbf{r}$, then $\cBad^{\star,\mathrm{out}}(A^{*},\mathcal{C},o,\mathbf{r})$ holds. In particular this holds when $o$ is a false output claim.
\end{lemma}

\noindent Its induction has the structure of \isaf{verifier\_\allowbreak acceptance\_\allowbreak implies\_\allowbreak gkr\_chain\_bad}, whose per-layer premise is replaced here by the proof of Case~1 below for $\widehat{P}$ and $\cSC^{4}$.

\begin{proof}
We show by induction on the length of $\mathcal{C}$: for every claim list $C$, value $v$ and strategy, if the layer reduction of $\mathcal{V}$ started at the first layer of $\mathcal{C}$ with $(C,v)$ accepts with the challenges $\mathbf{r}$, then $\cBad^{\star}(A^{*},\mathcal{C},C,v,\mathbf{r})$. The lemma is the case $C = ((z,1))$ and $v = \widetilde o(z) - \widetilde{\mathrm{const}}_{L_0}(z)$ of step 1, for the strategy that $\mathcal{P}^{*}$ follows after $z$, whose induced adversary is $A^{*}_z$.

Let $(s,L,b)$ be the first element. Acceptance implies that the $2s$ rounds of step 2(a) pass with the round polynomials $\varpi(\mathbf{r}_{[0,t)})$ of $A^{*}(0,())$, and that $u_{2s} = \Phi_{L,s,C}(\xi,\eta,t_x,t_y)$ for the carried values $(t_x,t_y)$ of $A^{*}(0,\mathbf{r}_{[0,2s)})$.

\emph{Case 1: $t_x = \widetilde b(\xi)$ and $t_y = \widetilde b(\eta)$.} By \eqref{eq:phi-hatp}, $u_{2s} = \widehat{P}_{L,s,C,b}(\xi,\eta)$. Together with the round checks, this is $\cSC^{4}(\widehat{P}_{L,s,C,b},v,\varpi,\mathbf{r}_{[0,2s)})$, so (i$^{\star}$) holds.

\emph{Case 2: $t_x\neq\widetilde b(\xi)$ or $t_y\neq\widetilde b(\eta)$.} If $\mathcal{C}$ has one element, then $b = b_{n-1}$ and step 2(d) rejects, a contradiction. Otherwise $\mathcal{V}$ continues at the next layer with the claim list and value of step 2(c), which are $C'$ and $v'$ of disjunct (ii). That run is the layer reduction on $\mathcal{C}'$ with the challenges $\mathbf{r}_{\geq 2s+1}$ for the strategy that $\mathcal{P}^{*}$ follows after $\mathbf{r}_{[0,2s+1)}$, and its induced adversary is the shift $A'$ of Definition~\ref{def:chainbad}. Since it accepts, the induction hypothesis gives $\cBad^{\star}(A',\mathcal{C}',C',v',\mathbf{r}_{\geq 2s+1})$, and (ii) holds.
\end{proof}

\begin{proposition}[Soundness of the verifier; paper proof]
\label{prop:b-faithful}
Let $\mathcal{C} = ((s_i,L_i,b_i))_{i\in[n]}$ be a layer chain and $o\colon[2^{\sout}]\to\F$ with $o(u)\neq\mathrm{eval}(L_0,b_0)(u)$ for some $u\in[2^{\sout}]$. For every faithful adversary $A$,
\[
\Pr_{\mathbf{r}\leftarrow\F^{N}}\bigl[\cBad^{\star,\mathrm{out}}(A,\mathcal{C},o,\mathbf{r})\bigr] \;\le\; \frac{\sout + \sum_{i\in[n]}(8s_i+1)}{|\F|}.
\]
Consequently, for every prover strategy whose coins are independent of the challenges, $\mathcal{V}$ accepts $o$ with probability at most this bound.
\end{proposition}

\noindent Its proof uses the mechanized Lemmas~\ref{lem:layer-identity}, \ref{lem:carry}, \ref{lem:seed} and~\ref{lem:split}, and Theorem~\ref{thm:sumcheck-std} in place of the imported Theorem~\ref{thm:sumcheck-afp}. The bound equals the one displayed in Remark~\ref{cor:b-deg4}.

\begin{proof}
We repeat the proofs of Lemma~\ref{lem:chain-bound} and Theorem~\ref{thm:b} with disjunct (i$^{\star}$) in place of (i). The only step that involves the layer acceptance is the bound on disjunct (i), where Lemma~\ref{lem:layer-sumcheck} is replaced by the following claim: if every gate of $L$ reads below $2^s$ and $v\neq\ccs_L(C,b)$, then for every strategy $\varpi$,
\[
\Pr_{\mathbf{r}\leftarrow\F^{2s}}\bigl[\cSC^{4}(\widehat{P}_{L,s,C,b},v,\varpi,\mathbf{r})\bigr]\le 8s/|\F|.
\]
In the variable order of Section~\ref{sec:gkr-verifier}, the Boolean points of $\cube{2s}$ are the pairs $(\bpt{x}{s},\bpt{y}{s})$, so by Lemma~\ref{lem:integrand-poly}(a) and Lemma~\ref{lem:layer-identity} the sum of $\widehat{P}_{L,s,C,b}$ over $\cube{2s}$ is $\sum_{x,y}I_{L,s,C,b}(x,y) = \ccs_L(C,b)\neq v$. By Lemma~\ref{lem:integrand-poly}(b) every variable has degree at most $4$, and Theorem~\ref{thm:sumcheck-std} with $m = 2s$ and $D_t = 4$ gives $8s/|\F|$. The bounds for the carry (Lemma~\ref{lem:carry}), for the continuation (induction and Lemma~\ref{lem:split}) and for the seed (Lemma~\ref{lem:seed}) do not involve the layer polynomial and carry over unchanged. Each layer then contributes $(8s_i+1)/|\F|$ and the seed $\sout/|\F|$.

For the consequence, let $\mathcal{P}^{*}$ be deterministic with induced adversary $A^{*}$. By Lemma~\ref{lem:accept-bad} the event that $\mathcal{V}$ accepts is contained in $\cBad^{\star,\mathrm{out}}(A^{*},\mathcal{C},o,\cdot)$, whose probability is bounded above. A strategy with independent coins is handled by averaging over the coins (Section~\ref{sec:probspace}).
\end{proof}


\section{Instantiation over the Quartic Extension}
\label{sec:smallfield}

The bound of \cref{thm:b} is an output term plus one term per layer, divided by $|\F|$. Over the KoalaBear prime field, $p = 2^{31}-2^{24}+1$, it is too weak for circuits of realistic size, so small-field provers evaluate circuits over $\F_p$ and draw the verifier's challenges from an extension field. Applying \cref{thm:b} in this setting requires two facts. The extension, here the binomial extension $\F_p[X]/(X^4-3)$ with basis $1,X,X^2,X^3$, must be a field with $p^4$ elements, and a circuit over $\F_p$ together with a false output claim must carry over to it. This section proves both; \cref{cor:perinput-ext} in \cref{sec:composition} combines them with the compiler soundness theorem of \cref{sec:compiler} into a per-input statement across the two fields. Two routine verifications, of ring axioms and of an equivalent form of chain well-formedness, are in \cref{app:ext-proofs}.

\subsection{The Base Field}

Let $p = 2^{31}-2^{24}+1 = 127\cdot 2^{24}+1 = 2{,}130{,}706{,}433$, and let $\F_p$ be the set $\{0,1,\dots,p-1\}$ with addition, negation and multiplication modulo $p$. It is a commutative ring with $p$ elements in which $0\neq 1$.

\begin{lemma}[The base field; mechanized]
\label{lem:ext-base}
The integer $p$ is prime, and every nonzero element of $\F_p$ has a multiplicative inverse. Hence $\F_p$ is a field with $p$ elements.
\end{lemma}

\begin{proof}
If $p = ab$ with $1<a\le b$, then $a^2\le p$, so it suffices to show that no integer $m\ge 2$ with $m^2\le p$ divides $p$. Since $46159^2 = 2{,}130{,}653{,}281\le p<46160^2 = 2{,}130{,}745{,}600$, these are the integers $2\le m\le 46159$, and computing the $46{,}158$ remainders $p\bmod m$ shows that none of them is zero.

For the inverse, let $x\in\{1,\dots,p-1\}$ and $\mu_x(y) = xy\bmod p$ for $y\in\{0,\dots,p-1\}$. If $\mu_x(y_1) = \mu_x(y_2)$, then $p$ divides $x(y_1-y_2)$. As $p$ is prime and $0<x<p$, $p$ divides $y_1-y_2$, and $|y_1-y_2|<p$ gives $y_1 = y_2$. An injective map of a finite set into itself is surjective, so $\mu_x(y) = 1$ for some $y$.
\end{proof}

\subsection{The Quartic Extension}

The extension is built as a tower of two quadratic extensions, each adjoining a square root of a nonsquare.

\begin{definition}[The tower]
\label{def:ext-tower}
Let $\extE$ be the set $\F_p\times\F_p$, whose element $(a,b)$ we write $a+b\theta$, with componentwise addition and
\[
(a+b\theta)(c+e\theta) = (ac+3be)+(ae+bc)\,\theta .
\]
Let $\extK$ be the set $\extE\times\extE$, whose element $(u,v)$ we write $u+v\alpha$, with componentwise addition and
\[
(u+v\alpha)(u'+v'\alpha) = (uu'+\theta\,vv')+(uv'+vu')\,\alpha .
\]
We regard $\F_p\subseteq\extE$ via $a\mapsto a+0\theta$ and $\extE\subseteq\extK$ via $u\mapsto u+0\alpha$; the multiplication rules then give $\theta^2 = 3$ and $\alpha^2 = \theta$. The \emph{base embedding} is $\iota\colon\F_p\to\extK$, $\iota(a) = (a+0\theta)+0\alpha$, and the \emph{coordinate map} is
\[
\chi\colon\F_p^4\to\extK,\qquad \chi(a_0,a_1,a_2,a_3) = (a_0+a_2\theta)+(a_1+a_3\theta)\,\alpha .
\]
\end{definition}

\noindent In the mechanization, $\extE$, $\extK$, $\iota$, $\alpha$ and $\chi$ are \isaf{kb\_quad}, \isaf{koala\_bear\_ext4}, \isaf{kb4\_embed}, \isaf{kb\_alpha} and \isaf{coeff4}.

\begin{lemma}[Nonsquares; mechanized]
\label{lem:nonsquare}
\leavevmode
\begin{enumerate}[label=(\alph*)]
\item $3$ is not a square in $\F_p$.
\item $-3$ is not a square in $\F_p$.
\item $\theta$ is not a square in $\extE$.
\end{enumerate}
\end{lemma}

\begin{proof}
(a) For an odd prime $q$ and an integer $a$, the Legendre symbol $\bigl(\frac{a}{q}\bigr)$ is $0$ if $q$ divides $a$, $1$ if $a$ is a nonzero square modulo $q$, and $-1$ otherwise. By the law of quadratic reciprocity~\cite{irelandrosen1990}, $\bigl(\frac{q}{q'}\bigr)\bigl(\frac{q'}{q}\bigr) = (-1)^{\frac{q-1}{2}\cdot\frac{q'-1}{2}}$ for distinct odd primes $q$ and $q'$. For $q = 3$ and $q' = p$, which is prime by \cref{lem:ext-base}, the exponent is $\frac{p-1}{2} = 2^{23}\cdot 127$, an even number, so $\bigl(\frac{3}{p}\bigr) = \bigl(\frac{p}{3}\bigr)$. From $2^{31}\equiv 2$ and $2^{24}\equiv 1\pmod 3$ we get $p\equiv 2\pmod 3$, and $2$ is not a square modulo $3$. Hence $\bigl(\frac{3}{p}\bigr) = \bigl(\frac{p}{3}\bigr) = -1$.

(b) Let $i = 2{,}113{,}994{,}754$; one multiplication shows $i^2\equiv -1\pmod p$. If $x^2 = -3$ in $\F_p$, then $(ix)^2 = i^2x^2 = (-1)(-3) = 3$, contradicting (a).

(c) Suppose $(a+b\theta)^2 = \theta$ with $a,b\in\F_p$. By the multiplication rule of $\extE$, $(a+b\theta)^2 = (a^2+3b^2)+2ab\,\theta$, so $a^2+3b^2 = 0$ and $2ab = 1$. The second equation gives $b\neq 0$, and dividing the first by $b^2$ gives $(a/b)^2 = -3$, contradicting (b).
\end{proof}

\begin{proposition}[The quartic extension field; partly mechanized]
\label{prop:ext-field}
\leavevmode
\begin{enumerate}[label=(\alph*)]
\item $\extE$ is a field with $p^2$ elements.
\item $\extK$ is a field with $p^4$ elements.
\item $\iota$ is an injective ring homomorphism, and $\iota(x^{-1}) = \iota(x)^{-1}$ for $x\neq 0$.
\item $\alpha^2 = \theta$, $\alpha^4 = \iota(3)$, and $\alpha\notin\iota(\F_p)$.
\item $\chi$ is a bijection, and $\chi(a_0,a_1,a_2,a_3) = \iota(a_0)+\iota(a_1)\,\alpha+\iota(a_2)\,\alpha^2+\iota(a_3)\,\alpha^3$.
\item The ring homomorphism $\F_p[X]\to\extK$ that agrees with $\iota$ on constants and maps $X$ to $\alpha$ induces an isomorphism $\F_p[X]/(X^4-3)\cong\extK$. In particular, $X^4-3$ is irreducible over $\F_p$.
\end{enumerate}
\end{proposition}

\noindent Parts (a)--(d) and the bijectivity of $\chi$ are mechanized; the formula for $\chi$ in (e) and part (f) are proved on paper only (\cref{tab:mechanization}).

\begin{proof}
Both levels of the tower are instances of one construction: for a commutative ring $R$ and $\delta\in R$, the set $R\times R$ with componentwise addition and $(a,b)(c,e) = (ac+\delta\,be,\ ae+bc)$ is a commutative ring with zero $(0,0)$ and unit $(1,0)$ (\cref{app:ext-proofs}). $\extE$ is the case $R = \F_p$, $\delta = 3$, and $\extK$ the case $R = \extE$, $\delta = \theta$.

(a), (b). Let $R$ be a field and $\delta\in R$ a nonsquare. For $(a,b)\neq(0,0)$ let $N = a^2-\delta\,b^2$. If $N = 0$, then $b\neq 0$ would give $(a/b)^2 = \delta$, and $b = 0$ would give $a^2 = 0$, hence $a = 0$; both are impossible, so $N\neq 0$. Then
\[
(a,b)\,(a/N,\,-b/N) = \bigl((a^2-\delta\,b^2)/N,\ (-ab+ba)/N\bigr) = (1,0),
\]
so $R\times R$ is a field with $|R|^2$ elements. By \cref{lem:ext-base} and \cref{lem:nonsquare}(a) this applies to $R = \F_p$ and $\delta = 3$, which gives (a). By (a) and \cref{lem:nonsquare}(c) it applies to $R = \extE$ and $\delta = \theta$, which gives (b), with $|\extK| = (p^2)^2 = p^4$.

(c) The inclusions $a\mapsto a+0\theta$ and $u\mapsto u+0\alpha$ preserve $0$, $1$, sums and negatives, and they preserve products because the multiplication rules reduce to $ac$ and $uu'$ when the second coordinates vanish. Hence their composite $\iota$ is a ring homomorphism. It is injective because $a$ is the first coordinate of the first coordinate of $\iota(a)$. From $\iota(x)\,\iota(x^{-1}) = \iota(1) = 1$ we get $\iota(x^{-1}) = \iota(x)^{-1}$.

(d) $\alpha^2 = (0+1\alpha)(0+1\alpha) = (0+\theta\cdot 1\cdot 1)+0\alpha = \theta$, so $\alpha^4 = \theta^2 = 3 = \iota(3)$. Every element of $\iota(\F_p)$ has second coordinate $0$, and $\alpha$ has second coordinate $1$.

(e) $\chi$ sends $(a_0,a_1,a_2,a_3)$ to the pair of pairs $\bigl((a_0,a_2),(a_1,a_3)\bigr)$, so it is a bijection from $\F_p^4$ onto $(\F_p\times\F_p)\times(\F_p\times\F_p) = \extK$. For the formula, $\alpha^2 = \theta$ and $\alpha^3 = \theta\alpha$ by (d), and the multiplication rules give $\iota(a)\,(u+v\alpha) = au+(av)\,\alpha$ with $a\,(c+e\theta) = ac+(ae)\,\theta$. Hence $\iota(a_0)+\iota(a_1)\alpha+\iota(a_2)\alpha^2+\iota(a_3)\alpha^3 = (a_0+a_2\theta)+(a_1+a_3\theta)\,\alpha$.

(f) Let $\phi\colon\F_p[X]\to\extK$ be the homomorphism of the statement. By (d), $X^4-3$ lies in its kernel, and by (e), $\phi(a_0+a_1X+a_2X^2+a_3X^3) = \chi(a_0,a_1,a_2,a_3)$, so $\phi$ is surjective. Every class of $\F_p[X]/(X^4-3)$ has a unique representative of degree below $4$, so the quotient has $p^4$ elements, and the induced surjection onto $\extK$ is a bijection. As $\extK$ is a field, $(X^4-3)$ is a maximal ideal, so $X^4-3$ is irreducible.
\end{proof}

\noindent By (e) and (f), $\chi$ reads a quadruple of base-field elements as the coefficients of $1,X,X^2,X^3$ in $\F_p[X]/(X^4-3)$.

\begin{lemma}[Uniform coordinates; mechanized]
\label{lem:pushforward}
If $(a_0,a_1,a_2,a_3)$ is uniformly distributed on $\F_p^4$, that is, if $a_0,\dots,a_3$ are independent and uniform on $\F_p$, then $\chi(a_0,a_1,a_2,a_3)$ is uniformly distributed on $\extK$.
\end{lemma}

\begin{proof}
By \cref{prop:ext-field}(e), $\chi$ is a bijection between sets with $p^4$ elements. For $z\in\extK$, $\Pr[\chi(a) = z] = \Pr[a = \chi^{-1}(z)] = p^{-4} = 1/|\extK|$.
\end{proof}

\noindent For $N$ independent uniform quadruples, the same argument applied to the bijection $\chi^N\colon(\F_p^4)^N\to\extK^N$ shows that the $N$ resulting elements are independent and uniform on $\extK$, which is the challenge distribution of \cref{thm:b-ext4} below; this extension is proved on paper and is not mechanized. Whether the values that an implementation derives from its transcript are independent and uniform is a property of its Fiat--Shamir construction and is outside this section.

\subsection{The Chain Bound over the Extension}

\begin{theorem}[Chain soundness over the quartic extension; mechanized]
\label{thm:b-ext4}
\Cref{thm:b} holds with $\F = \extK$, the representatives being multivariate polynomials over $\extK$ whose degree is their total degree. Thus, for a layer chain over $\extK$ whose layers read tables with $s_i$ index bits, representatives admissible for a degree-bound function $\mathrm{db}$, any prover strategy, and a claimed output table that differs from the true output at some point of the output cube, the chain event of \cref{def:chainbad} has probability at most
\[
\frac{s_{\mathrm{out}}+\sum_i\bigl(2s_i\,\mathrm{db}(s_i)+1\bigr)}{p^4}
\]
over challenges drawn independently and uniformly from $\extK$.
\end{theorem}

\begin{proof}
The proof of \cref{thm:b} uses no property of $\F$ beyond its being a finite field, and of the polynomials only the interface of the imported sumcheck development~\cite{sumcheckafp2024}: variable sets are finite; the degree of a sum is at most the larger of the two degrees, and partial substitution does not raise the degree; evaluation is additive and compatible with partial substitution; and two distinct univariate polynomials of degree at most $d$ agree on at most $d$ points. That development proves these interface facts for multivariate polynomials over a finite integral domain, with total degree as the degree. Since $\extK$ is a finite field by \cref{prop:ext-field}(b), both requirements hold, \cref{thm:b} applies with $\F = \extK$, and $|\extK| = p^4$.
\end{proof}

\noindent Admissible representatives exist over $\extK$ with $\mathrm{db}(s) = 2s$: the multilinear interpolation of each layer integrand is constructed over an arbitrary finite field (\isaf{layer\_repr\_mpoly} in \texttt{Layer\_Representative.thy}), and these representatives discharge the premise of \cref{thm:b-ext4}. The verifier's acceptance over $\extK$ is bounded separately, on paper, by \cref{prop:b-faithful}, which uses the integrand polynomial of \cref{lem:integrand-poly} in place of a representative.

\begin{remark}[Size of the bound]
\label{rem:ext-scale}
Since $\log_2 p^4>123.95$, the bound is small at practical circuit sizes. For the verifier of \cref{def:verifier}, \cref{prop:b-faithful} over $\extK$ gives the per-layer term $8s_i+1$. For at most $2^{10}$ layers with $s_i\le 32$ and $s_{\mathrm{out}}\le 2^{10}$, as for circuits whose layers have at most $2^{32}$ wires, the numerator is at most $2^{10}(1+257)<2^{18.02}$, and the bound is below $2^{-105}$. The circuits of the constraint-level model of \cref{sec:compiler} give every entry its own width bit, so their widths are larger: at the default parameters $d = 24$ and $m = 31$ of \cref{sec:semantics}, with $W = 8$, the input and the output of the hash stack have fewer than $2^{10}$ width bits. For at most $2^{10}$ layers with $s_i < 2^{10}$ and $s_{\mathrm{out}} < 2^{10}$, the numerator is below $2^{10}(1+8\cdot 2^{10}+1)<2^{23.01}$, and the bound is below $2^{-100}$. \Cref{fig:bound} plots the bound for $s_i = s_{\mathrm{out}} = 32$ over $\F_p$ and over $\extK$, and over $\extK$ for $s_{\mathrm{out}} = 448$, the update circuit at the default parameters, with every $s_i$ at its upper bound $1023$. This is a statement about the interactive protocol over $\extK$ (\cref{sec:limitations}).
\end{remark}

\begin{figure}[t]
\centering
\begin{tikzpicture}
\begin{axis}[
  width=0.66\textwidth, height=5.4cm,
  xmode=log, log basis x=2,
  xmin=1, xmax=4096,
  ymin=-128, ymax=0,
  xtick={1,4,16,64,256,1024,4096},
  xticklabels={$2^{0}$,$2^{2}$,$2^{4}$,$2^{6}$,$2^{8}$,$2^{10}$,$2^{12}$},
  ytick={-120,-100,-80,-60,-40,-20,0},
  xlabel={number of layers $n$},
  ylabel={$\log_2$ of the bound},
  label style={font=\small}, tick label style={font=\footnotesize},
  grid=major, grid style={black!10},
  legend style={font=\footnotesize, draw=none, fill=white, fill opacity=0.85, text opacity=1, at={(0.03,0.55)}, anchor=west},
  legend cell align=left]
\addplot[domain=1:4096, samples=80, thick, black!45] {ln(32+257*x)/ln(2) - 30.98869};
\addlegendentry{$\F_p$, $s_i = 32$}
\addplot[domain=1:4096, samples=80, thick, blue!65!black] {ln(32+257*x)/ln(2) - 123.95475};
\addlegendentry{$\extK$, $s_i = 32$}
\addplot[domain=1:4096, samples=80, thick, dashed, blue!65!black] {ln(448+8185*x)/ln(2) - 123.95475};
\addlegendentry{$\extK$, $s_i = 1023$}
\addplot[only marks, mark=*, mark size=1.3pt, black!45, forget plot] coordinates {(1024,-12.98)};
\addplot[only marks, mark=*, mark size=1.3pt, blue!65!black, forget plot] coordinates {(1024,-105.95) (1024,-100.96)};
\node[font=\footnotesize, anchor=north, black!60] at (axis cs:1024,-14.5) {$2^{-13.0}$};
\node[font=\footnotesize, anchor=north, blue!65!black] at (axis cs:1024,-106.5) {$2^{-105.9}$};
\node[font=\footnotesize, anchor=south east, blue!65!black] at (axis cs:900,-100.4) {$2^{-101.0}$};
\end{axis}
\end{tikzpicture}
\caption{The bound of \cref{prop:b-faithful} against the number of layers, for $s_i = s_{\mathrm{out}} = 32$ and, over $\extK$, for the compiled update circuit ($s_{\mathrm{out}} = 448$) with every $s_i$ at its upper bound $1023$. Dots mark $n = 2^{10}$ (\cref{rem:ext-scale}).}
\label{fig:bound}
\end{figure}
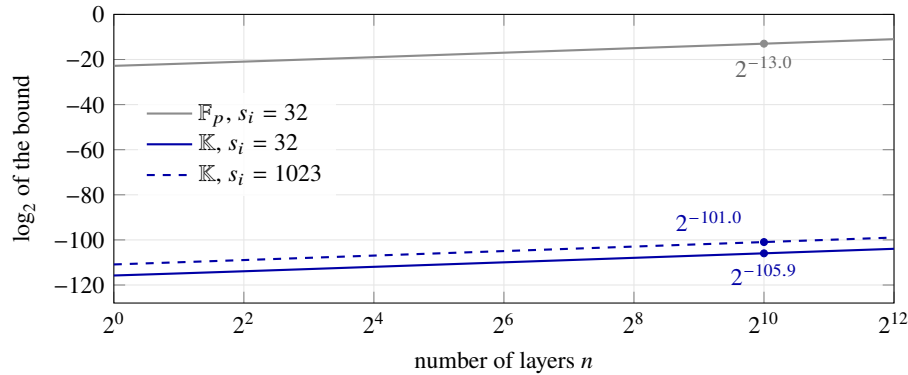

\subsection{Carrying Base-Field Circuits into the Extension}

We use the circuits of \cref{def:circuit} and the multilinear extensions of \cref{def:mle}. In this subsection a table $b$ is also read as a vector with entries $b_j = b(j)$ and length $|b|$, and we write $\mathrm{eval}(L,b)_z$ for $\mathrm{eval}(L,b)(z)$. By \cref{def:circuit},
\begin{equation}
\label{eq:ext-eval}
\mathrm{eval}(L,b)_z = \mathrm{const}_L(z) \;+\sum_{g\in G,\;\mathrm{out}(g) = z} c(g)\,\mathrm{sem}\bigl(\kappa(g),\,b_{\mathrm{in}_1(g)},\,b_{\mathrm{in}_2(g)}\bigr),
\end{equation}
where both sums run over lists, so that repeated entries count repeatedly; $\mathrm{const}_L(z)$ is the \emph{constant sum} of $L$ at $z$. We say that $L$ \emph{reads inside} $b$ if $\mathrm{in}_1(g),\mathrm{in}_2(g)<|b|$ for every gate $g$ of $L$. A layer that is well formed over $|b|$ reads inside $b$.

\begin{definition}[Lifting]
\label{def:ext-lift}
The lift of a gate over $\F_p$ replaces its coefficient $c$ by $\iota(c)$ and keeps its kind and wires. The lift $\operatorname{lift}(L)$ of a layer keeps its width, lifts every gate, and replaces every constant $(o,c)$ by $(o,\iota(c))$. A vector is lifted entrywise, $\iota(b) = (\iota(b_0),\dots,\iota(b_{m-1}))$, and a chain is lifted by lifting each of its layers and vectors and keeping its widths.
\end{definition}

\noindent In the mechanization, the lifts of gates and layers are \isaf{lift\_gate} and \isaf{lift\_layer} in \texttt{KoalaBear\_Ext4\_Lift.thy}.

\begin{lemma}[Lifting one layer; partly mechanized]
\label{lem:lift-layer}
Let $L$ be a layer over $\F_p$ and $b$ a vector over $\F_p$.
\begin{enumerate}[label=(\alph*)]
\item $\iota$ commutes with finite sums, finite products and powers.
\item For $r\in\F_p^n$ and $j<2^n$, $\mathrm{eq}(\iota(r),j) = \iota(\mathrm{eq}(r,j))$. Consequently, for every table $f$ on $\{0,\dots,2^n-1\}$ over $\F_p$, $\widetilde{\iota\circ f}(\iota(r)) = \iota(\widetilde f(r))$; in particular, $\widetilde{\iota\circ f}$ takes the value $\iota(f(j))$ at the point of $\cube{n}$ given by the bits of $j$.
\item $w(\operatorname{lift}(L)) = w(L)$, and for every wire $z$ the constant sum of $\operatorname{lift}(L)$ at $z$ is $\iota$ of that of $L$.
\item If $L$ reads inside $b$, then $\mathrm{eval}(\operatorname{lift}(L),\iota(b)) = \iota(\mathrm{eval}(L,b))$.
\item For every $m$, $\operatorname{lift}(L)$ is well formed over $m$ if and only if $L$ is.
\item If $L$ reads inside $b$, $z<2^{w(L)}$, and $f(z)\neq\mathrm{eval}(L,b)_z$ for a function $f$ from wires to $\F_p$, then $\iota(f(z))\neq\mathrm{eval}(\operatorname{lift}(L),\iota(b))_z$.
\end{enumerate}
\end{lemma}

\noindent Parts (a)--(c) and (e) are mechanized. Parts (d) and (f) are mechanized under the stronger hypothesis that $L$ is well formed over $|b|$; the versions stated here follow from mechanized lemmas by the argument below (\cref{tab:mechanization}).

\begin{proof}
(a) holds for every ring homomorphism, and $\iota$ is one by \cref{prop:ext-field}(c). (b) $\mathrm{eq}(r,j)$ is a product of factors $r_\ell$ and $1-r_\ell$, and $\widetilde f(r)$ is a sum of products $\mathrm{eq}(r,j)\,f(j)$, so (a) gives both identities; at the point given by the bits of $j$, $\mathrm{eq}(\cdot,j')$ is $1$ for $j' = j$ and $0$ otherwise. (c) The width is unchanged, and the constant sum of $\operatorname{lift}(L)$ at $z$ is the sum of the values $\iota(c)$ over the same constants, which is $\iota$ of the constant sum of $L$ at $z$ by (a). (e) Lifting changes coefficients and constant values only, while well-formedness refers only to the width, the wires and the kinds.

(d) Both sides have length $2^{w(L)}$ by (c). Fix $z<2^{w(L)}$ and evaluate \eqref{eq:ext-eval} for $\operatorname{lift}(L)$ on $\iota(b)$. The constant sum is $\iota$ of that of $L$ by (c). A gate $g$ of $L$ with $\mathrm{out}(g) = z$ contributes, through its lift,
\[
\iota(c(g))\,\mathrm{sem}\bigl(\kappa(g),\,\iota(b)_{\mathrm{in}_1(g)},\,\iota(b)_{\mathrm{in}_2(g)}\bigr) = \iota\bigl(c(g)\,\mathrm{sem}(\kappa(g),\,b_{\mathrm{in}_1(g)},\,b_{\mathrm{in}_2(g)})\bigr),
\]
because $\mathrm{in}_1(g),\mathrm{in}_2(g)<|b|$ gives $\iota(b)_{\mathrm{in}_k(g)} = \iota(b_{\mathrm{in}_k(g)})$ for $k = 1,2$, and $\mathrm{sem}(\kappa(g),\cdot,\cdot)$ is a monomial in its last two arguments. The lifted gates writing to $z$ are the lifts of the gates of $L$ writing to $z$, in the same order, so (a) gives $\mathrm{eval}(\operatorname{lift}(L),\iota(b))_z = \iota(\mathrm{eval}(L,b)_z)$.

(f) By (d), $\mathrm{eval}(\operatorname{lift}(L),\iota(b))_z = \iota(\mathrm{eval}(L,b)_z)$, and $\iota$ is injective by \cref{prop:ext-field}(c).
\end{proof}

\noindent Part (b) concerns points with coordinates in $\iota(\F_p)$. At other points of $\extK^n$ the extension of an embedded table in general leaves $\iota(\F_p)$: for the table $(1,2)$ on $\{0,1\}$ it takes the value $(1-\alpha)+2\alpha = 1+\alpha$ at $\alpha$ (mechanized as \isaf{lifted\_mle\_positive\_witness}). The chain over $\extK$ therefore works with the extensions over $\extK$ of embedded tables, and the transfer below concerns the tables themselves, through part (d).

We write a layer chain (\cref{def:chain}) over a field as a triple $(s,L,b)$ of lists of a common length $n$: widths $s_i$, layers $L_i$, and the tables $b_i$ that the layers read, $b_{n-1}$ being the input. A triple of lists of a common length $n\ge 1$ satisfies the mechanized chain condition \isaf{gkr\_chain\_wf} exactly when
\begin{enumerate}[label=(Ch\arabic*)]
\item every gate $g$ of $L_i$ has $\mathrm{in}_1(g),\mathrm{in}_2(g)<2^{s_i}$, for all $i<n$, and
\item $s_i = w(L_{i+1})$ and $b_i = \mathrm{eval}(L_{i+1},b_{i+1})$, for all $i<n-1$.
\end{enumerate}
The mechanized definition \isaf{gkr\_chain\_wf} states these conditions recursively along the chain, and an induction on $n$ shows that the two forms agree (\cref{app:ext-proofs}). A layer chain in the sense of \cref{def:chain} also has $|b_{n-1}| = 2^{s_{n-1}}$; conditions (Ch1) and (Ch2) alone do not bound the length of $b_{n-1}$, so the following proposition is stated for triples that satisfy them and whose last layer reads inside its input.

\begin{proposition}[Lifting a chain; paper proof]
\label{prop:lift-chain}
Let $(s,L,b)$ be a triple of lists of length $n\ge 1$ over $\F_p$ that satisfies (Ch1) and (Ch2) and whose last layer $L_{n-1}$ reads inside $b_{n-1}$; both hold for a layer chain (\cref{def:chain}). Then:
\begin{enumerate}[label=(\alph*)]
\item $\mathrm{eval}(\operatorname{lift}(L_i),\iota(b_i)) = \iota(\mathrm{eval}(L_i,b_i))$ for every $i<n$;
\item the lifted triple $(s,\operatorname{lift}(L),\iota(b))$ satisfies (Ch1) and (Ch2) over $\extK$, and it is a layer chain over $\extK$ when $(s,L,b)$ is a layer chain over $\F_p$;
\item $w(\operatorname{lift}(L_0)) = w(L_0)$, and for every function $f$ from wires to $\F_p$ and every $z<2^{w(L_0)}$ with $f(z)\neq\mathrm{eval}(L_0,b_0)_z$, we have $\iota(f(z))\neq\mathrm{eval}(\operatorname{lift}(L_0),\iota(b_0))_z$.
\end{enumerate}
In particular, if $f$ is a false output claim for $(s,L,b)$, that is, $f(z)\neq\mathrm{eval}(L_0,b_0)_z$ for some $z<2^{w(L_0)}$, then $\iota\circ f$ is a false output claim for the lifted chain.
\end{proposition}

\noindent It rests on \cref{lem:lift-layer}.

\begin{proof}
We first show that $L_i$ reads inside $b_i$ for every $i<n$. For $i = n-1$ this is the hypothesis. For $i<n-1$, (Ch2) gives $b_i = \mathrm{eval}(L_{i+1},b_{i+1})$, a vector of length $2^{w(L_{i+1})} = 2^{s_i}$, and (Ch1) bounds the input wires of every gate of $L_i$ by $2^{s_i}$.

(a) follows from \cref{lem:lift-layer}(d).

(b) We check (Ch1) and (Ch2) for the lifted triple, whose three lists again have length $n$; lifting preserves the length of $b_{n-1}$, so a layer chain lifts to a layer chain. Lifting a gate keeps its input wires, so every gate of $\operatorname{lift}(L_i)$ has input wires below $2^{s_i}$. For $i<n-1$, \cref{lem:lift-layer}(c) gives $w(\operatorname{lift}(L_{i+1})) = w(L_{i+1}) = s_i$, and (Ch2) for the original chain together with (a) at index $i+1$ gives
\[
\iota(b_i) = \iota(\mathrm{eval}(L_{i+1},b_{i+1})) = \mathrm{eval}(\operatorname{lift}(L_{i+1}),\iota(b_{i+1})).
\]

(c) This is \cref{lem:lift-layer}(c) and (f) for $L_0$ and $b_0$, which applies because $L_0$ reads inside $b_0$.
\end{proof}

\section{Sparse-Merkle Semantics and Leaf Binding}
\label{sec:semantics}

This section fixes the meaning against which the compiled circuit is verified: the tree, the leaf encoding, the two hash functions and the validity predicate of the three operations (Definition~\ref{def:smt}). It then proves two facts about how a root binds the leaf stored at a key: the pre-image layout of the leaf hash is injective (Lemma~\ref{lem:leaffold}), and two different leaves with valid paths to one root at one key yield a collision of the leaf hash or of the node hash (Proposition~\ref{prop:binding}). Because the node hash of the implementation is not collision resistant, the section ends with the property that is assumed instead, position binding of the tree (Definition~\ref{def:pb}).

\begin{definition}[Sparse-Merkle semantics]
\label{def:smt}
Fix a depth $d \in \mathbb{N}$, a bound $m \geq 1$ on the length of leaf encodings, a set $\mathcal{U}$ of payloads, a set $\mathcal{D}$ of digests and the field $\F$.
\begin{enumerate}[label=(\alph*),leftmargin=*]
\item \emph{Keys.} A key is a natural number $k$; it is \emph{in range} if $k < 2^d$. Its bits are $\smBit_j(k) = \lfloor k/2^j \rfloor \bmod 2$ for $0 \leq j < d$. Bit $0$ belongs to the level of the leaf.
\item \emph{Leaf states.} The set $\mathcal{L}$ of leaf states consists of $\smEmpty$ (never populated), $\smOcc(v)$ for $v \in \mathcal{U}$ (populated with payload $v$) and $\smTomb$ (populated, then retired).
\item \emph{Leaf encoding.} Given a payload encoder $\smPenc \colon \mathcal{U} \to \F^{\ast}$, the encoding $\smEncode \colon \mathcal{L} \to \F^{\ast}$ is
\[
\smEncode(\smEmpty) = (0), \qquad \smEncode(\smOcc(v)) = (1) \,\|\, \smPenc(v), \qquad \smEncode(\smTomb) = (2),
\]
where $\|$ denotes concatenation. Its first entry is the \emph{tag} $\smTag(\ell) \in \{0, 1, 2\}$ of the state $\ell$.
\item \emph{Fold.} Let $\nu(n) = n \cdot 1_{\F}$ and $w_m = 8 \lceil (m+1)/8 \rceil$. For $e = (e_0, \ldots, e_{n-1}) \in \F^{n}$ with $1 \leq n \leq m$,
\[
\smFold_m(e) = \bigl(e_0,\ \nu(n),\ e_1, \ldots, e_{n-1},\ 0, \ldots, 0\bigr) \in \F^{w_m},
\]
that is, the tag lane, the length lane, the remaining entries verbatim, and zeros up to the width $w_m$. The width satisfies $w_m \geq m + 1$ and is a multiple of the sponge rate $8$.
\item \emph{Hashes.} A leaf hash $\hleaf \colon \mathcal{L} \to \mathcal{D}$ and a node hash $\hnode \colon \mathcal{D} \times \mathcal{D} \to \mathcal{D}$.
\item \emph{Paths.} For $a \in \mathcal{D}$, a key $k$ and a sibling list $\vec{s} = (s_0, \ldots, s_{n-1}) \in \mathcal{D}^{n}$, the \emph{accumulator chain} is $a_0 = a$ and, for $0 \leq j < n$,
\[
a_{j+1} = \begin{cases} \hnode(a_j, s_j) & \text{if } \smBit_j(k) = 0,\\ \hnode(s_j, a_j) & \text{if } \smBit_j(k) = 1; \end{cases}
\]
the path root is $\smPathRoot(a, k, \vec{s}) = a_n$. For a leaf state $\ell$ and a digest $r$, $\smPathOK(\ell, k, \vec{s}, r)$ holds iff $n = d$, $k < 2^d$ and $\smPathRoot(\hleaf(\ell), k, \vec{s}) = r$.
\item \emph{Operations and validity.} An operation is $\smMem(k, v)$, $\smNon(k)$ or $\smUpd(k, \ell_{\mathrm{o}}, \ell_{\mathrm{n}})$ for a key $k$ (the key of the operation), a payload $v$ and leaf states $\ell_{\mathrm{o}}, \ell_{\mathrm{n}}$. A witness is a pair $w = (\ell, \vec{s})$ of a leaf state and a sibling list. For digests $r$ (the claimed root) and $r'$ (the claimed post-root),
\[
\begin{aligned}
\smtvalid(\smMem(k,v), r, r', (\ell,\vec{s})) &\iff r' = r \,\wedge\, \ell = \smOcc(v) \,\wedge\, \smPathOK(\ell, k, \vec{s}, r),\\
\smtvalid(\smNon(k), r, r', (\ell,\vec{s})) &\iff r' = r \,\wedge\, \ell \in \{\smEmpty, \smTomb\} \,\wedge\, \smPathOK(\ell, k, \vec{s}, r),\\
\smtvalid(\smUpd(k,\ell_{\mathrm{o}},\ell_{\mathrm{n}}), r, r', (\ell,\vec{s})) &\iff \ell = \ell_{\mathrm{o}} \,\wedge\, \smPathOK(\ell_{\mathrm{o}}, k, \vec{s}, r) \,\wedge\, \smPathOK(\ell_{\mathrm{n}}, k, \vec{s}, r').
\end{aligned}
\]
\end{enumerate}
\end{definition}

The definition transcribes \isaf{encode}, \isaf{path\_root}, \isaf{path\_ok} and \isaf{smt\_valid} of \texttt{SMT\_Semantics.thy} and \isaf{leaf\_fold} and \isaf{leaf\_pre\_width} of \texttt{Leaf\_Fold.thy}. The formal \isaf{path\_root} consumes the low bit of $k$ at each level and recurses on $\lfloor k/2 \rfloor$, which gives the chain of item~(f) (Lemma~\ref{lem:smt-acc-chain}). An update uses one sibling list for both roots, so it changes the leaf at $k$ and nothing else. In the formal semantics $\hleaf$ and $\hnode$ are parameters with no assumption on them. In the implementation, digests are vectors of $8$ field elements, $\hnode$ is a truncated Poseidon2 permutation~\cite{poseidon2} with no feed-forward, and $\hleaf(\ell) = \smSponge(\smFold_m(\smEncode(\ell)))$ for an overwrite-mode sponge $\smSponge$ of rate $8$ over the same permutation; the default parameters are $d = 24$ and $m = 31$. This factorization of $\hleaf$ is used only in Lemma~\ref{lem:leaffold} and Proposition~\ref{prop:binding}. The results of Section~\ref{sec:compiler} use no property of either hash.

\begin{lemma}[Leaf-fold injectivity and the leaf-side reduction; mechanized]
\label{lem:leaffold}
Assume that $\nu$ is injective on $\{0, 1, \ldots, m\}$.
\begin{enumerate}[label=(\roman*),leftmargin=*]
\item If $e, e' \in \F^{\ast}$ satisfy $1 \leq |e| \leq m$, $1 \leq |e'| \leq m$ and $\smFold_m(e) = \smFold_m(e')$, then $e = e'$.
\item Let $H$ be a function on $\F^{\ast}$ and $\smSponge$ a function on $\F^{w_m}$ such that $H(e) = \smSponge(\smFold_m(e))$ whenever $1 \leq |e| \leq m$. If $e \neq e'$ are two such encodings with $H(e) = H(e')$, then $p = \smFold_m(e)$ and $p' = \smFold_m(e')$ satisfy $p \neq p'$ and $\smSponge(p) = \smSponge(p')$.
\end{enumerate}
\end{lemma}

\noindent The formal statement of (ii) asserts that some $p \neq p'$ with $\smSponge(p) = \smSponge(p')$ exist; that the pair is $(\smFold_m(e), \smFold_m(e'))$ is read off its proof and follows from (i).

\begin{proof}
(i) Write $n = |e|$ and $n' = |e'|$, and compare the two vectors lane by lane. Lane $0$ gives $e_0 = e'_0$. Lane $1$ gives $\nu(n) = \nu(n')$; since $n, n' \leq m$, the hypothesis on $\nu$ gives $n = n'$. The remaining $w_m - 2$ lanes are $(e_1, \ldots, e_{n-1})$ followed by $w_m - n - 1$ zeros, and the same for $e'$ with the same $n$, so $e_j = e'_j$ for $1 \leq j < n$. Hence $e = e'$.

(ii) By (i), $e \neq e'$ implies $p \neq p'$, and $\smSponge(p) = H(e) = H(e') = \smSponge(p')$.
\end{proof}

In a field of positive characteristic $c$, $\nu(a) = \nu(b)$ holds iff $c$ divides $a - b$, so the hypothesis on $\nu$ is equivalent to $m < c$; in characteristic zero it always holds. For KoalaBear ($c = p = 2^{31} - 2^{24} + 1$) and $m = 31$ it is mechanized as \isaf{kb\_of\_nat\_inj\_on\_bound} in \texttt{KoalaBear\_Field.thy}. The condition $m < |\F|$ would not suffice in an extension field, whose characteristic is smaller than its size. The length lane is needed: the encodings $(1, 7)$ and $(1, 7, 0)$ agree on every lane of their folds except the length lane (\isaf{zero\_extension\_class\_separated}). The lemma assumes nothing about $\smSponge$; it converts a collision of $H$ into an explicit collision of $\smSponge$.

The next proposition is the standard Merkle-tree argument~\cite{merkle1987} for one key, followed by Lemma~\ref{lem:leaffold} on the leaf side.

\begin{proposition}[Binding of a root to the leaf at a key; paper proof]
\label{prop:binding}
Let $k$ be a key, $r$ a digest, and $(\ell, \vec{s})$, $(\ell', \vec{s}\,')$ two witnesses with $\ell \neq \ell'$, $\smPathOK(\ell, k, \vec{s}, r)$ and $\smPathOK(\ell', k, \vec{s}\,', r)$. Then one of the following holds, and which one, together with the colliding inputs, is computed from $(k, \ell, \ell', \vec{s}, \vec{s}\,')$ by $2$ evaluations of $\hleaf$ and $2d$ evaluations of $\hnode$:
\begin{enumerate}[label=(\alph*),leftmargin=*]
\item $\hleaf(\ell) = \hleaf(\ell')$;
\item there are pairs $(u, t) \neq (u', t')$ in $\mathcal{D}^2$ with $\hnode(u, t) = \hnode(u', t')$.
\end{enumerate}
Assume moreover that $\hleaf(\lambda) = \smSponge(\smFold_m(\smEncode(\lambda)))$ for every $\lambda \in \mathcal{L}$, that $\smPenc$ is injective with $1 + |\smPenc(v)| \leq m$ for every $v$, that $2 \neq 0$ in $\F$, and that $\nu$ is injective on $\{0, \ldots, m\}$. Then in case~(a), $p = \smFold_m(\smEncode(\ell))$ and $p' = \smFold_m(\smEncode(\ell'))$ satisfy $p \neq p'$ and $\smSponge(p) = \smSponge(p')$.
\end{proposition}

\noindent Its last step, from equal leaf hashes of two leaf states with different in-bound encodings to a collision of the sponge, is mechanized, in the existential form of Lemma~\ref{lem:leaffold}~(ii), as \isaf{encoded\_state\_collision\_to\_sponge\_collision} in \texttt{Leaf\_Fold.thy}. The factorization of $\hleaf$ through $\smSponge$, $\smFold_m$ and $\smEncode$ is a hypothesis of the proposition and is not part of the formal semantics.

\begin{proof}
Both sibling lists have length $d$. Let $(a_0, \ldots, a_d)$ and $(a'_0, \ldots, a'_d)$ be the accumulator chains of $(\hleaf(\ell), k, \vec{s})$ and $(\hleaf(\ell'), k, \vec{s}\,')$; computing them takes the stated hash evaluations, and $a_d = r = a'_d$. If $a_0 = a'_0$, then $\hleaf(\ell) = \hleaf(\ell')$, which is case~(a). Otherwise let $j < d$ be the largest index with $a_j \neq a'_j$, so that $a_{j+1} = a'_{j+1}$. Both chains use the same bit $\beta = \smBit_j(k)$ at level $j$ because the key is the same. If $\beta = 0$, then $\hnode(a_j, s_j) = a_{j+1} = a'_{j+1} = \hnode(a'_j, s'_j)$ and the pairs $(a_j, s_j)$ and $(a'_j, s'_j)$ differ in the first component. If $\beta = 1$, then $\hnode(s_j, a_j) = \hnode(s'_j, a'_j)$ and the pairs differ in the second component. This is case~(b).

For the second statement, the encodings $\smEncode(\ell)$ and $\smEncode(\ell')$ are different. The tags $0, 1, 2$ are pairwise distinct in $\F$: $1 \neq 0$ in a field, $2 \neq 0$ by hypothesis, and $2 = 1$ would give $1 = 0$. Hence states of different kinds have different first entries, and $\smOcc(v) \neq \smOcc(v')$ gives $\smPenc(v) \neq \smPenc(v')$ by injectivity. Both encodings have length between $1$ and $m$: the empty and retired states have length $1 \leq m$, and $|\smEncode(\smOcc(v))| = 1 + |\smPenc(v)| \leq m$. In case~(a), Lemma~\ref{lem:leaffold}~(ii) with $H(e) = \smSponge(\smFold_m(e))$ gives $p \neq p'$ and $\smSponge(p) = \smSponge(p')$.
\end{proof}

Proposition~\ref{prop:binding} is a security argument only for a node hash that is collision resistant, and the node hash of the implementation is not. With $\mathsf{Perm}$ the public Poseidon2 permutation of $\F^{16}$, $\hnode(a, b)$ is the first half of $\mathsf{Perm}(a \| b)$, with no feed-forward, so for any $z \in \F^{8}$ and distinct $u, u' \in \F^{8}$ the vectors $\mathsf{Perm}^{-1}(z \| u)$ and $\mathsf{Perm}^{-1}(z \| u')$ are two different inputs with node hash $z$, found with two calls to the inverse permutation. Such a collision does not give two authentication paths, because a path must start at a leaf hash and its accumulators must connect to that start. The property a system needs is the following one, which case~(b) of the proposition does not reduce to anything simpler for this hash.

\begin{definition}[Position binding]
\label{def:pb}
An algorithm $\mathcal{A}$ \emph{breaks position binding} of the tree of depth $d$ if it outputs $(k, r, \ell, \vec{s}, \ell', \vec{s}\,')$ with $\ell \neq \ell'$, $\smPathOK(\ell, k, \vec{s}, r)$ and $\smPathOK(\ell', k, \vec{s}\,', r)$. Its advantage $\mathsf{Adv}^{\mathrm{pb}}(\mathcal{A})$ is the probability of this event.
\end{definition}

\noindent Two leaf states with equal leaf hashes break position binding with any sibling list, so the definition includes the collision resistance of $\hleaf$ on leaf states. The hash is fixed and unkeyed, so an algorithm that breaks position binding exists, as for any fixed compressing hash; the assumption is read through the explicit reduction of Corollary~\ref{cor:one-leaf}. We do not prove position binding. It is the one cryptographic assumption of this paper and it enters only Corollary~\ref{cor:one-leaf}; the results of Section~\ref{sec:compiler} and the single-run bounds of Section~\ref{sec:composition} use no cryptographic property of the hashes.

Support for the assumption comes from Coratger, Khovratovich, Mennink and Wagner~\cite{coratger2026merkle}, who analyze Merkle trees with a truncated permutation as compression function and, in one of their results, an overwrite sponge over the same permutation as leaf hash. In the ideal-permutation model they prove strong position binding for a tree of fixed depth, with advantage at most $(4q^2 + 2q)/(|\mathcal{D}| - 1)$ for $q \leq |\mathcal{D}|$ queries to the permutation and its inverse (their Theorem~3 with Lemma~1). Under the hypotheses of Proposition~\ref{prop:binding}, different leaf states have different inputs to the sponge, so their notion implies Definition~\ref{def:pb} up to the order of the key bits. For $|\mathcal{D}| = p^{8}$ with the KoalaBear prime $p$, that bound reaches $1$ at about $2^{123}$ queries, the birthday limit of a $248$-bit digest. Their leaf hash starts with a public key block and the implementation's does not, and their model idealizes Poseidon2, so we cite their theorem as support and do not apply it. Section~\ref{sec:composition} uses position binding to read an accepted run.

\section{Exact Arithmetization}
\label{sec:compiler}

This section defines the input vector that encodes a witness (Definition~\ref{def:encoding}), the constraint system and the circuit compiled from it (Definition~\ref{def:residuals}), and the verifier model (Definition~\ref{def:verifier-model}). It then proves that the verifier model accepts exactly the valid transitions (Theorems~\ref{thm:a-sound} and~\ref{thm:a-complete}, Corollary~\ref{cor:a-iff}). Circuits are layered circuits as in Definition~\ref{def:circuit}. A circuit $C$ \emph{accepts} an input vector $x$, written $\smAcc(C, x)$, when every entry of its output vector $V_0$ is zero (\isaf{circuit\_accept} in \texttt{Layered\_Circuit.thy}). The results of this section hold over any integral domain $\F$; from Section~\ref{sec:composition} on, $\F$ is a finite field.

The compiled circuit depends only on the kind of the operation: $\smKind(\smMem(k,v)) = \smKM$, $\smKind(\smNon(k)) = \smKN$ and $\smKind(\smUpd(k,\ell_{\mathrm o},\ell_{\mathrm n})) = \smKU$. We call $\smKM$ and $\smKN$ the read-only kinds. Digests and leaf states reach the circuit through two representation maps: an injective map $\smRep \colon \mathcal{D} \to \F^{W}$ that writes a digest as $W$ field elements, and a map $\smPre \colon \mathcal{L} \to \F^{\smWl}$ with $\smWl \geq 1$, the leaf pre-image, whose first entry is the tag, $\smPre(\lambda)_0 = \smTag(\lambda)$. In the implementation $W = 8$, $\smWl = w_m$ and $\smPre(\lambda) = \smFold_m(\smEncode(\lambda))$. The results below use only the injectivity of $\smRep$ and the tag entry of $\smPre$.

\begin{definition}[Witness encoder]
\label{def:encoding}
Let $\ell, \ell'$ be leaf states, $k$ a key, $\vec{s} = (s_0, \ldots, s_{d-1})$ a sibling list of length $d$, and $r, r', \mathit{vd}$ digests. Write $(a_0, \ldots, a_d)$ and $(a'_0, \ldots, a'_d)$ for the accumulator chains of $(\hleaf(\ell), k, \vec{s})$ and $(\hleaf(\ell'), k, \vec{s})$, $\vec{b} = (\smBit_0(k), \ldots, \smBit_{d-1}(k))$, and $0^{t}$ for $t$ zeros.
\begin{enumerate}[label=(\alph*),leftmargin=*]
\item For a read-only kind $\tau$, let $n_{\mathrm{ro}} = \smWl + W(2d+3) + d$ and
\[
x^{\circ} = \smPre(\ell) \,\|\, \smRep(a_0) \| \cdots \| \smRep(a_d) \,\|\, \smRep(s_0) \| \cdots \| \smRep(s_{d-1}) \,\|\, \vec{b} \,\|\, \smRep(r) \,\|\, \smRep(\mathit{vd}),
\]
a vector of length $n_{\mathrm{ro}}$. Then $\smEnc_{\tau}(\ell, k, \vec{s}, r, \mathit{vd}) = x^{\circ} \,\|\, 0^{2^{n_{\mathrm{ro}}} - n_{\mathrm{ro}}}$.
\item For updates, let $n_{\mathrm{up}} = 2\smWl + W(3d+5) + d$ and
\[
x^{\circ} = \smPre(\ell) \| \smPre(\ell') \,\|\, \smRep(a_0) \| \cdots \| \smRep(a_d) \,\|\, \smRep(a'_0) \| \cdots \| \smRep(a'_d) \,\|\, \smRep(s_0) \| \cdots \| \smRep(s_{d-1}) \,\|\, \vec{b} \,\|\, \smRep(r) \| \smRep(r') \| \smRep(\mathit{vd}).
\]
Then $\smEnc_{\smKU}(\ell, \ell', k, \vec{s}, r, r', \mathit{vd}) = x^{\circ} \,\|\, 0^{2^{n_{\mathrm{up}}} - n_{\mathrm{up}}}$.
\item For an operation $\mathit{op}$ and a witness $w = (\ell, \vec{s})$ with $|\vec{s}| = d$, the input vector $\smEnc(\mathit{op}, r, r', \mathit{vd}, w)$ is $\smEnc_{\smKM}(\ell, k, \vec{s}, r, \mathit{vd})$ for $\mathit{op} = \smMem(k, v)$, $\smEnc_{\smKN}(\ell, k, \vec{s}, r, \mathit{vd})$ for $\mathit{op} = \smNon(k)$, and $\smEnc_{\smKU}(\ell_{\mathrm o}, \ell_{\mathrm n}, k, \vec{s}, r, r', \mathit{vd})$ for $\mathit{op} = \smUpd(k, \ell_{\mathrm o}, \ell_{\mathrm n})$.
\end{enumerate}
We write $n_\tau$ for $n_{\mathrm{ro}}$ or $n_{\mathrm{up}}$; every encoded vector has length $2^{n_\tau}$.
\end{definition}

\noindent The encoder computes the accumulator chains itself from the leaf states and the siblings, writes the key bits of $k$ and the claimed digests $r, r', \mathit{vd}$, and, for an update, encodes the old leaf $\ell_{\mathrm o}$ of the operation in place of the witness leaf. It is \isaf{encode\_witness} of \texttt{Compiler\_Model.thy}, built from \isaf{enc\_ro} and \isaf{enc\_up}; the chain is \isaf{acc\_list}.

\begin{definition}[Residual constraint system and compiled circuit]
\label{def:residuals}
\emph{Hash sub-circuits.} Let $\hwl \colon \F^{\smWl} \to \F^{W}$ and $\hwn \colon \F^{W} \times \F^{W} \to \F^{W}$ be the functions computed by the hash sub-circuits. They are related to the semantics only by
\begin{equation}
\hwl(\smPre(\lambda)) = \smRep(\hleaf(\lambda)), \qquad \hwn(\smRep(a), \smRep(b)) = \smRep(\hnode(a, b)) \qquad (\lambda \in \mathcal{L},\ a, b \in \mathcal{D}). \tag{H}
\end{equation}

\emph{Segments.} For a vector $x$ of length at least $n_\tau$, let $P$, $A_0, \ldots, A_d$, $S_0, \ldots, S_{d-1}$, $\beta_0, \ldots, \beta_{d-1}$, $R$ and $V$ be its segments at the positions of $\smPre(\ell)$, $\smRep(a_0), \ldots, \smRep(a_d)$, $\smRep(s_0), \ldots, \smRep(s_{d-1})$, $\smBit_0(k), \ldots, \smBit_{d-1}(k)$, $\smRep(r)$ and $\smRep(\mathit{vd})$ in Definition~\ref{def:encoding}~(a) or~(b); each $\beta_j$ is a single entry. For updates let also $P'$, $A'_0, \ldots, A'_d$ and $R'$ be the segments at the positions of $\smPre(\ell')$, $\smRep(a'_0), \ldots, \smRep(a'_d)$ and $\smRep(r')$. The \emph{tag entry} is $x_0$, the first entry of $P$.

\emph{Multiplexer.} For $\beta \in \F$ and $u, t \in \F^{W}$, lane by lane,
\[
\smMuxL(\beta, u, t) = u + \beta (t - u), \qquad \smMuxR(\beta, u, t) = t + \beta (u - t),
\]
so that $\beta = 0$ selects $(u, t)$ and $\beta = 1$ selects $(t, u)$. Let $M_j = \hwn\bigl(\smMuxL(\beta_j, A_j, S_j), \smMuxR(\beta_j, A_j, S_j)\bigr)$ for $j < d$, and let $M'_j$ be the same expression with $A'_j$ in place of $A_j$.

\emph{Constraints.} For a read-only kind $\tau$, $\smCons_\tau(x)$ is the conjunction of
\begin{center}
\begin{tabular}{@{}l@{\qquad}l@{\qquad}l@{}}
(C1) & $\hwl(P) = A_0$ & leaf hash\\
(C2) & $M_j = A_{j+1}$ for all $j < d$ & compression chain\\
(C3) & $A_d = R$ & root binding\\
(C4) & $A_0 = V$ & value-digest binding\\
(C5) & $\beta_j (\beta_j - 1) = 0$ for all $j < d$ & key-bit booleanity\\
(C6) & $x_0 (x_0 - 2) = 0$, only for $\tau = \smKN$ & non-membership tag
\end{tabular}
\end{center}
For updates, $\smCons_{\smKU}(x)$ is the conjunction of (C1) for $(P, A_0)$ and for $(P', A'_0)$; (C2) for both chains, $M_j = A_{j+1}$ and $M'_j = A'_{j+1}$, with the same $S_j$ and $\beta_j$; (C3) for $A_d = R$ and for $A'_d = R'$; (C4$'$) $A'_0 = V$; and (C5).

\emph{Residuals.} Each constraint is checked by residuals, field elements that vanish exactly when it holds: $u_i - t_i$ for each lane $i < W$ of a vector equation $u = t$ in (C1) to (C4$'$), $\beta_j^2 - \beta_j$ for (C5), and $x_0^2 - 2x_0$ for (C6). Let $\Delta_\tau$ be the list of these residuals and $\smRout = |\Delta_\tau|$.

\emph{Compiled circuit.} The circuit $C_\tau$ has $n_\tau$ input width bits and layers $L_0, L_1, \ldots, L_{n-1}$ with $L_0 = \smRes(\Delta_\tau)$ and $(L_1, \ldots, L_{n-1}) = \smStack_\tau$.
\begin{itemize}[leftmargin=*]
\item The \emph{evaluation stack} $\smStack_\tau$ is a list of layers that computes the vector $\smEv_\tau(x)$ listing $\hwl(P)$, $M_0, \ldots, M_{d-1}$, $A_0, \ldots, A_d$, $\beta_0, \ldots, \beta_{d-1}$, $R$, $V$ and $x_0$ (for updates: $\hwl(P)$, $\hwl(P')$, the $M_j$, the $M'_j$, the $A_j$, the $A'_j$, the $\beta_j$, $R$, $R'$ and $V$), padded with zeros to length $2^{e_\tau}$, where $e_\tau$ is the length of the unpadded list:
\[
V_n = x \ \text{ and } \ V_i = \mathrm{eval}(L_i, V_{i+1}) \ (1 \leq i < n) \quad \Longrightarrow \quad V_1 = \smEv_\tau(x), \qquad \text{for every } x \in \F^{2^{n_\tau}}. \tag{S}
\]
\item The \emph{residual layer} $\smRes(\Delta)$ has $|\Delta|$ width bits and no constants. For the $z$-th residual it has two gates that write wire $z$ and read entries of $\smEv_\tau(x)$: for $u_i - t_i$, two linear gates with coefficients $1$ and $-1$ reading $u_i$ and $t_i$; for $\beta^2 - \beta$, a product gate with coefficient $1$ reading $\beta$ twice and a linear gate with coefficient $-1$ reading $\beta$; for $x_0^2 - 2x_0$, the same with coefficient $-2$ on the linear gate.
\end{itemize}
\end{definition}

\noindent This is the locale \isaf{compiler\_model} of \texttt{Compiler\_Model.thy}: (H) is \isaf{hl\_correct} and \isaf{hn\_correct}, (S) is \isaf{stack\_spec\_ro} and \isaf{stack\_spec\_up}, the residual lists are \isaf{descs\_ro} and \isaf{descs\_up} with gates \isaf{res\_gates}, the conjunctions are \isaf{cons\_ro} and \isaf{cons\_up}, and $C_\tau$ is \isaf{compile}.

\Cref{fig:arith} shows how the verifier model of Definition~\ref{def:verifier-model} uses the compiled circuit.

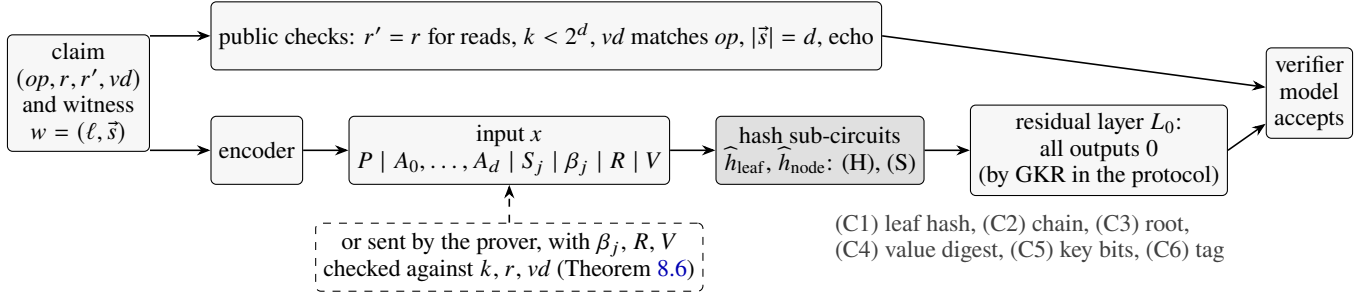
\begin{figure}[t]
\centering
\begin{tikzpicture}[font=\footnotesize,
  bx/.style={draw, rounded corners=2pt, align=center, inner sep=3pt, minimum height=9mm, fill=black!3},
  hs/.style={draw, rounded corners=2pt, align=center, inner sep=3pt, minimum height=9mm, fill=black!12},
  arr/.style={-{Stealth[length=1.8mm]}, semithick},
  note/.style={align=left, inner sep=1pt, text=black!75}]
\node[bx] (enc) {encoder};
\node[bx, right=6mm of enc] (inp) {input $x$\\ $P \mid A_0,\dots,A_d \mid S_j \mid \beta_j \mid R \mid V$};
\node[hs, right=6mm of inp] (hs) {hash sub-circuits\\ $\hwl$, $\hwn$: (H), (S)};
\node[bx, right=6mm of hs] (res) {residual layer $L_0$:\\ all outputs $0$\\ (by GKR in the protocol)};
\node[bx, above=6mm of enc.north west, anchor=south west] (pub) {public checks: $r' = r$ for reads, $k < 2^d$, $\mathit{vd}$ matches $\mathit{op}$, $|\vec{s}| = d$, echo};
\node[bx, anchor=east] (in) at ($(enc.west)!0.5!(enc.west |- pub.west) + (-8mm,0)$) {claim\\ $(\mathit{op}, r, r', \mathit{vd})$\\ and witness\\ $w = (\ell, \vec{s})$};
\node[bx, anchor=west] (acc) at ($(res.east)!0.5!(res.east |- pub.east) + (5mm,0)$) {verifier\\ model\\ accepts};
\node[note, below=2mm of res.south east, anchor=north east] (cons) {(C1) leaf hash, (C2) chain, (C3) root,\\ (C4) value digest, (C5) key bits, (C6) tag};
\node[bx, dashed, fill=white, below=5mm of inp] (prv) {or sent by the prover, with $\beta_j$, $R$, $V$\\ checked against $k$, $r$, $\mathit{vd}$ (\cref{thm:a-vector})};
\draw[arr] (in.east |- enc) -- (enc.west);
\draw[arr] (in.east |- pub) -- (pub.west);
\draw[arr] (enc) -- (inp);
\draw[arr] (inp) -- (hs);
\draw[arr] (hs) -- (res);
\draw[arr, dashed] (prv) -- (inp);
\draw[arr] (pub.east) -- (acc);
\draw[arr] (res.east) -- (acc);
\end{tikzpicture}
\caption{The verifier model with the compiled circuit, for read-only kinds. Updates add $P'$, $A'_j$ and $R'$ and use (C4$'$) in place of (C4); (C6) applies to non-membership only. The dashed input is the setting of \cref{thm:a-vector}.}
\label{fig:arith}
\end{figure}

\paragraph{Abstract hash sub-circuits.} The model fixes the residual layer, the multiplexer arithmetic and the layout, and treats the layers that compute $\hwl$ and $\hwn$ as an abstract list of layers specified by (S); the hash functions enter only through (H). The gate-level arithmetization of Poseidon2 inside the stack is outside the formal development, and the results of this section hold for every stack that satisfies (S). (H) is a correctness specification of the sub-circuits: no collision resistance or injectivity of the hashes is assumed in this section. The assumptions are satisfiable. \texttt{Compiler\_Instance.thy} builds a one-layer stack over KoalaBear for depth $1$ with a linear hash pair, proves (H) and (S) for it (\isaf{compiler\_model\_kb}), and applies Theorems~\ref{thm:a-complete} and~\ref{thm:a-sound} below to a membership claim (\isaf{instance\_theorem\_A\_activation}, \isaf{instance\_theorem\_A\_roundtrip}).

\paragraph{Widths in the model.} The residual layer has $\smRout = |\Delta_\tau|$ width bits, hence $2^{\smRout}$ wires. Wire $z < \smRout$ carries the $z$-th residual; the other wires have no gates, evaluate to $0$ and never prevent acceptance. Taking $\smRout$ bits for $\smRout$ residuals uses $\smRout < 2^{\smRout}$ and avoids index arithmetic. For the residual layer, this convention affects only the first term of the bound in Theorem~\ref{thm:c}, which equals $\smRout$; a layer of $\lceil \log_2 \smRout \rceil$ bits would make that term $\lceil \log_2 \smRout \rceil$. Counting the residuals of Definition~\ref{def:residuals},
\[
|\Delta_{\smKM}| = W(d+3) + d, \qquad |\Delta_{\smKN}| = W(d+3) + d + 1, \qquad |\Delta_{\smKU}| = W(2d+5) + d,
\]
so that for $W = 8$ these are $240$, $241$ and $448$ at $d = 24$, and $312$, $313$ and $584$ at $d = 32$. The other widths follow the same convention. $C_\tau$ has $n_\tau$ input width bits, so an input vector has $2^{n_\tau}$ entries, and the encoder writes zeros after position $n_\tau$. This length is a convention of the model and not a message size: the constraints read only the first $n_\tau$ entries, so the verifier of \cref{sec:smt-vector} can receive those entries and take the rest to be zero. The target $\smEv_\tau(x)$ of (S) is padded to $2^{e_\tau}$ entries because layer evaluation produces vectors of power-of-two length, so an unpadded target could not be met. These widths enter the bound of Theorem~\ref{thm:c} through the widths of the layers below $L_0$.

\begin{lemma}[Acceptance is the constraint conjunction; mechanized]
\label{lem:circuit-cons}
Assume (S). For a read-only kind $\tau$ and every $x \in \F^{2^{n_{\mathrm{ro}}}}$, $\smAcc(C_\tau, x) \iff \smCons_\tau(x)$. For every $x \in \F^{2^{n_{\mathrm{up}}}}$, $\smAcc(C_{\smKU}, x) \iff \smCons_{\smKU}(x)$.
\end{lemma}

\noindent In the development, \isaf{A\_i\_ro} and \isaf{A\_i\_up} are proved inside the locale \isaf{compiler\_model}, whose assumptions include (H) as well as (S). Their proofs use only (S), and the constraints $\smCons$ refer to the functions that the sub-circuits compute, not to $\hleaf$ or $\hnode$. When $\mathcal{D} = \F^{W}$ and $\smRep$ is the identity, as in \cref{sec:smt-vector}, interpreting the locale with these functions as the hashes makes (H) hold by definition, so the lemma applies there as stated.

\begin{proof}
Let $x$ have length $2^{n_\tau}$. By (S), the vector below the residual layer is $\smEv_\tau(x)$, so the output of $C_\tau$ on $x$ is $\mathrm{eval}(\smRes(\Delta_\tau), \smEv_\tau(x))$. For $z < \smRout$, wire $z$ is written by exactly the two gates of the $z$-th residual and has no constant. A linear gate with coefficient $c$ reading $u$ contributes $c u$, and a product gate with coefficient $c$ reading $\beta$ twice contributes $c \beta^2$, so wire $z$ takes the value of the $z$-th residual at $\smEv_\tau(x)$. For $\smRout \leq z < 2^{\smRout}$ no gate writes wire $z$, so its value is $0$. Hence $C_\tau$ accepts $x$ iff every residual vanishes at $\smEv_\tau(x)$. The residuals read the entries of $\smEv_\tau(x)$ that hold $\hwl(P)$, $M_j$, $A_j$, $\beta_j$, $R$, $V$ and $x_0$, and their update counterparts. The lane residuals of a vector equation all vanish iff the equation holds; moreover $\beta^2 - \beta = \beta(\beta - 1)$ and $x_0^2 - 2x_0 = x_0(x_0 - 2)$. So the residuals vanish iff $\smCons_\tau(x)$. Appendix~\ref{app:smt-proofs} gives the positions of these entries and the gate count per wire.
\end{proof}

\begin{lemma}[The constraints on encoded inputs; mechanized]
\label{lem:smt-honest-enc}
Assume (H), let $\vec{s}$ have length $d$, and let $(a_j)$, $(a'_j)$ be as in Definition~\ref{def:encoding}.
\begin{enumerate}[label=(\roman*),leftmargin=*]
\item For a read-only kind $\tau$ and $x = \smEnc_\tau(\ell, k, \vec{s}, r, \mathit{vd})$, the constraints (C1), (C2) and (C5) hold, and
\[
\smCons_\tau(x) \iff \smPathRoot(\hleaf(\ell), k, \vec{s}) = r \;\wedge\; \mathit{vd} = \hleaf(\ell) \;\wedge\; \bigl(\tau = \smKN \Rightarrow \ell \in \{\smEmpty, \smTomb\}\bigr).
\]
\item For $x = \smEnc_{\smKU}(\ell, \ell', k, \vec{s}, r, r', \mathit{vd})$, the constraints (C1) and (C2) for both chains and (C5) hold, and
\[
\smCons_{\smKU}(x) \iff \smPathRoot(\hleaf(\ell), k, \vec{s}) = r \;\wedge\; \smPathRoot(\hleaf(\ell'), k, \vec{s}) = r' \;\wedge\; \mathit{vd} = \hleaf(\ell').
\]
\end{enumerate}
\end{lemma}

\noindent The constraints that hold unconditionally are intermediate steps of these two proofs; for update chains the compression step is the separate lemma \isaf{up\_chain\_holds}.

\begin{proof}
(i) The segments of $x$ are the blocks written by the encoder: $P = \smPre(\ell)$, $A_j = \smRep(a_j)$, $S_j = \smRep(s_j)$, $\beta_j = \smBit_j(k)$, $R = \smRep(r)$, $V = \smRep(\mathit{vd})$ and $x_0 = \smTag(\ell)$. The zero padding lies after position $n_{\mathrm{ro}}$ and does not affect them.
(C1) holds by (H): $\hwl(\smPre(\ell)) = \smRep(\hleaf(\ell)) = \smRep(a_0)$.
For (C2), if $\smBit_j(k) = 0$ the multiplexer selects $(\smRep(a_j), \smRep(s_j))$, and by (H) and the definition of the chain, $M_j = \smRep(\hnode(a_j, s_j)) = \smRep(a_{j+1})$; if $\smBit_j(k) = 1$ it selects $(\smRep(s_j), \smRep(a_j))$ and $M_j = \smRep(\hnode(s_j, a_j)) = \smRep(a_{j+1})$.
(C5) holds because each $\beta_j \in \{0, 1\}$.
Since $\smRep$ is injective and $a_d = \smPathRoot(\hleaf(\ell), k, \vec{s})$, (C3) holds iff $\smPathRoot(\hleaf(\ell), k, \vec{s}) = r$, and (C4) holds iff $\hleaf(\ell) = \mathit{vd}$.
For (C6), $\F$ has no zero divisors, so $\smTag(\ell)(\smTag(\ell) - 2) = 0$ iff $\smTag(\ell) \in \{0, 2\}$. The tag of an occupied state is $1$, and $1 \notin \{0, 2\}$ since $1 = 2$ would give $1 = 0$. Hence (C6) holds iff $\ell \in \{\smEmpty, \smTomb\}$.

(ii) The same argument applies to both chains, which share $\vec{s}$ and $k$: (C1), (C2) and (C5) hold, (C3) holds iff both path roots equal $r$ and $r'$ respectively, and (C4$'$) holds iff $\hleaf(\ell') = \mathit{vd}$.
\end{proof}

\begin{definition}[Verifier model and canonical witness]
\label{def:verifier-model}
Let $\mathit{op}$ be an operation with key $k$, let $r, r', \mathit{vd}$ be digests and $w = (\ell, \vec{s})$ a witness.
\begin{enumerate}[label=(\alph*),leftmargin=*]
\item The value-digest check and the witness echo are
\[
\smVd(\mathit{op}, \mathit{vd}) \iff \begin{cases} \mathit{vd} = \hleaf(\smOcc(v)) & \text{if } \mathit{op} = \smMem(k, v),\\ \mathit{vd} \in \{\hleaf(\smEmpty), \hleaf(\smTomb)\} & \text{if } \mathit{op} = \smNon(k),\\ \mathit{vd} = \hleaf(\ell_{\mathrm n}) & \text{if } \mathit{op} = \smUpd(k, \ell_{\mathrm o}, \ell_{\mathrm n}), \end{cases}
\]
\[
\smEcho(\mathit{op}, w) \iff \begin{cases} \text{true} & \text{if } \mathit{op} = \smMem(k, v),\\ \ell \in \{\smEmpty, \smTomb\} & \text{if } \mathit{op} = \smNon(k),\\ \ell = \ell_{\mathrm o} & \text{if } \mathit{op} = \smUpd(k, \ell_{\mathrm o}, \ell_{\mathrm n}). \end{cases}
\]
\item The public-side checks $\smPub(\mathit{op}, r, r', \mathit{vd})$ hold iff $r' = r$ whenever $\smKind(\mathit{op}) \neq \smKU$, $k < 2^d$, and $\smVd(\mathit{op}, \mathit{vd})$.
\item The verifier model accepts, $\smVAcc(\mathit{op}, r, r', \mathit{vd}, w)$, iff $\smPub(\mathit{op}, r, r', \mathit{vd})$, $|\vec{s}| = d$, $\smEcho(\mathit{op}, w)$ and $\smAcc\bigl(C_{\smKind(\mathit{op})}, \smEnc(\mathit{op}, r, r', \mathit{vd}, w)\bigr)$; the last condition is evaluated only when $|\vec{s}| = d$.
\item The canonical witness and the canonical value digest are
\[
\smCanon(\mathit{op}, w) = \begin{cases} (\smOcc(v), \vec{s}) & \text{if } \mathit{op} = \smMem(k, v),\\ w & \text{otherwise}, \end{cases}
\]
\[
\smVdc(\mathit{op}, w) = \begin{cases} \hleaf(\smOcc(v)) & \text{if } \mathit{op} = \smMem(k, v),\\ \hleaf(\ell) & \text{if } \mathit{op} = \smNon(k),\\ \hleaf(\ell_{\mathrm n}) & \text{if } \mathit{op} = \smUpd(k, \ell_{\mathrm o}, \ell_{\mathrm n}). \end{cases}
\]
\end{enumerate}
\end{definition}

\noindent These are \isaf{vd\_ok}, \isaf{pub\_ok}, \isaf{echo\_ok}, \isaf{verifier\_accept}, \isaf{canon\_witness} and \isaf{wit\_vd} of \texttt{Compiler\_Correctness.thy}. The public-side checks read only the operation, the claimed roots and the claimed value digest; the witness enters through its length, the echo and the encoder. The verifier model evaluates $C_{\smKind(\mathit{op})}$ on the encoded input. In the protocol of Section~\ref{sec:gkr} the verifier performs the same checks and replaces this evaluation by an interactive proof that the output of the circuit on the encoded input is the zero vector; Section~\ref{sec:composition} connects the two.

\begin{theorem}[Compiler soundness; mechanized]
\label{thm:a-sound}
Assume (H) and (S). If $\smVAcc(\mathit{op}, r, r', \mathit{vd}, w)$, then $\smtvalid(\mathit{op}, r, r', \smCanon(\mathit{op}, w))$.
\end{theorem}

\begin{proof}
Let $w = (\ell, \vec{s})$, let $k$ be the key of $\mathit{op}$, $\tau = \smKind(\mathit{op})$ and $x = \smEnc(\mathit{op}, r, r', \mathit{vd}, w)$. Acceptance gives $\smPub(\mathit{op}, r, r', \mathit{vd})$, $|\vec{s}| = d$, $\smEcho(\mathit{op}, w)$ and $\smAcc(C_\tau, x)$; in particular $k < 2^d$. The vector $x$ has length $2^{n_\tau}$, so Lemma~\ref{lem:circuit-cons} gives $\smCons_\tau(x)$, and Lemma~\ref{lem:smt-honest-enc} applies because $|\vec{s}| = d$.
\begin{itemize}[leftmargin=*]
\item $\mathit{op} = \smMem(k, v)$. Lemma~\ref{lem:smt-honest-enc} gives $\smPathRoot(\hleaf(\ell), k, \vec{s}) = r$ and $\mathit{vd} = \hleaf(\ell)$, and the public-side checks give $r' = r$ and $\mathit{vd} = \hleaf(\smOcc(v))$. Hence $\smPathRoot(\hleaf(\smOcc(v)), k, \vec{s}) = r$, so $\smPathOK(\smOcc(v), k, \vec{s}, r)$, which together with $r' = r$ is $\smtvalid(\mathit{op}, r, r', (\smOcc(v), \vec{s}))$.
\item $\mathit{op} = \smNon(k)$. Lemma~\ref{lem:smt-honest-enc} gives $\smPathRoot(\hleaf(\ell), k, \vec{s}) = r$ and $\ell \in \{\smEmpty, \smTomb\}$, and the public-side checks give $r' = r$. Hence $\smtvalid(\mathit{op}, r, r', w)$, and $\smCanon(\mathit{op}, w) = w$.
\item $\mathit{op} = \smUpd(k, \ell_{\mathrm o}, \ell_{\mathrm n})$. The echo gives $\ell = \ell_{\mathrm o}$, and Lemma~\ref{lem:smt-honest-enc} gives $\smPathRoot(\hleaf(\ell_{\mathrm o}), k, \vec{s}) = r$ and $\smPathRoot(\hleaf(\ell_{\mathrm n}), k, \vec{s}) = r'$. Hence $\smPathOK(\ell_{\mathrm o}, k, \vec{s}, r)$ and $\smPathOK(\ell_{\mathrm n}, k, \vec{s}, r')$, which with $\ell = \ell_{\mathrm o}$ is $\smtvalid(\mathit{op}, r, r', w)$. \qedhere
\end{itemize}
\end{proof}

For membership the conclusion is about $(\smOcc(v), \vec{s})$. Acceptance of the encoded input depends on the witness leaf only through $\hleaf(\ell)$, which the value-digest constraint and the public-side checks equate with $\mathit{vd} = \hleaf(\smOcc(v))$; so the verifier model cannot separate $\ell$ from $\smOcc(v)$ when the two have the same leaf hash. Whether $\ell = \smOcc(v)$ is a question about collisions of $\hleaf$, which Lemma~\ref{lem:leaffold} reduces to the sponge; leaving it out of Theorem~\ref{thm:a-sound} keeps the theorem free of hash assumptions.

\begin{theorem}[Compiler completeness; mechanized]
\label{thm:a-complete}
Assume (H) and (S). If $\smtvalid(\mathit{op}, r, r', w)$, then $\smVAcc(\mathit{op}, r, r', \smVdc(\mathit{op}, w), w)$.
\end{theorem}

\begin{proof}
Let $w = (\ell, \vec{s})$, $k$ the key of $\mathit{op}$ and $\mathit{vd} = \smVdc(\mathit{op}, w)$. Validity contains $\smPathOK$ for the leaf state on the root $r$, so $|\vec{s}| = d$ and $k < 2^d$.
\begin{itemize}[leftmargin=*]
\item $\mathit{op} = \smMem(k, v)$. Validity gives $r' = r$, $\ell = \smOcc(v)$ and $\smPathRoot(\hleaf(\ell), k, \vec{s}) = r$, and $\mathit{vd} = \hleaf(\smOcc(v)) = \hleaf(\ell)$. Lemma~\ref{lem:smt-honest-enc} gives $\smCons_{\smKM}$ on the encoded input, and Lemma~\ref{lem:circuit-cons} gives acceptance by $C_{\smKM}$. The public-side checks hold ($r' = r$, $k < 2^d$, $\mathit{vd} = \hleaf(\smOcc(v))$) and the echo is trivial.
\item $\mathit{op} = \smNon(k)$. Validity gives $r' = r$, $\ell \in \{\smEmpty, \smTomb\}$ and $\smPathRoot(\hleaf(\ell), k, \vec{s}) = r$, and $\mathit{vd} = \hleaf(\ell) \in \{\hleaf(\smEmpty), \hleaf(\smTomb)\}$. As before, the circuit accepts, the public-side checks hold, and the echo is $\ell \in \{\smEmpty, \smTomb\}$.
\item $\mathit{op} = \smUpd(k, \ell_{\mathrm o}, \ell_{\mathrm n})$. Validity gives $\ell = \ell_{\mathrm o}$, $\smPathRoot(\hleaf(\ell_{\mathrm o}), k, \vec{s}) = r$ and $\smPathRoot(\hleaf(\ell_{\mathrm n}), k, \vec{s}) = r'$, and $\mathit{vd} = \hleaf(\ell_{\mathrm n})$. Lemmas~\ref{lem:smt-honest-enc} and~\ref{lem:circuit-cons} give acceptance by $C_{\smKU}$, the public-side checks hold ($k < 2^d$, $\mathit{vd} = \hleaf(\ell_{\mathrm n})$, no root equality for updates), and the echo is $\ell = \ell_{\mathrm o}$. \qedhere
\end{itemize}
\end{proof}

\begin{corollary}[Acceptance equals validity; mechanized]
\label{cor:a-iff}
Assume (H) and (S). Then $\bigl(\exists\, \mathit{vd}, w.\ \smVAcc(\mathit{op}, r, r', \mathit{vd}, w)\bigr) \iff \bigl(\exists\, w.\ \smtvalid(\mathit{op}, r, r', w)\bigr)$.
\end{corollary}

\begin{proof}
From left to right, Theorem~\ref{thm:a-sound} makes $\smCanon(\mathit{op}, w)$ valid. From right to left, Theorem~\ref{thm:a-complete} gives acceptance with $\mathit{vd} = \smVdc(\mathit{op}, w)$.
\end{proof}

\paragraph{Inputs produced by the encoder.} Theorems~\ref{thm:a-sound} and~\ref{thm:a-complete} concern input vectors that the encoder of Definition~\ref{def:encoding} produces from a witness and the public values. On such inputs the leaf-hash, compression-chain and booleanity constraints hold by construction (Lemma~\ref{lem:smt-honest-enc}), so soundness rests on the constraints that compare the computed chains with the public digests ((C3), (C4) and (C4$'$)), on the tag constraint (C6) and on the public-side checks. The verifier of the implementation also computes the input vector it uses from the witness and the public inputs before it runs the protocol. A verifier that receives the input vector in place of the witness needs soundness for arbitrary input vectors, which the theorems of this section do not give; \cref{sec:smt-vector} proves it on paper. Lemma~\ref{lem:circuit-cons} holds for every input vector, but the step from the constraints to validity uses facts that the encoder supplies and no residual checks: the key-bit entries carry the bits of the public key, the root entries carry the public roots, the value-digest entries carry the public value digest, and, for an update, the first leaf pre-image is $\smPre(\ell_{\mathrm o})$ for the public old leaf. \Cref{sec:smt-vector} therefore adds checks on the public segments of the input. It takes a digest to be the vector that represents it, assumes that the node sub-circuit computes the compression function on every pair of vectors, and uses the leaf part of (H) only at the old leaf of an update.

\section{Composition and Per-Input Soundness}
\label{sec:composition}

This section joins the deterministic equivalence of Section~\ref{sec:compiler} with the probabilistic bound of Theorem~\ref{thm:b}. From here on $\F$ is a finite field, and the circuit and the challenges live in the same field: the compiler model of Section~\ref{sec:compiler} and the layer-sumcheck setting of Section~\ref{sec:gkr} are fixed over one $\F$; a representative family and a degree-bound function $\mathrm{db}$ are hypotheses of the theorems below, not fixed parameters. Theorem~\ref{thm:c} is the mechanized composition for claims that no witness makes valid, and Corollary~\ref{cor:perinput} is its per-input form.

Fix an operation $\mathit{op}$ of kind $\tau = \smKind(\mathit{op})$, digests $r, r', \mathit{vd}$ and a witness $w = (\ell, \vec{s})$, and let $x = \smEnc(\mathit{op}, r, r', \mathit{vd}, w)$. Let $L_0 = \smRes(\Delta_\tau), L_1, \ldots, L_{n-1}$ be the layers of $C_\tau$ (Definition~\ref{def:residuals}), and let $V_n = x$ and $V_i = \mathrm{eval}(L_i, V_{i+1})$ for $1 \leq i < n$ be the true value vectors below the output layer. The claimed output is the all-zero table on $\cube{\smRout}$, where $\smRout = |\Delta_\tau|$ is the number of width bits of $L_0$.

\begin{theorem}[Composition soundness; mechanized]
\label{thm:c}
Assume (H) and (S) of Definition~\ref{def:residuals}, and
\begin{enumerate}[label=(\roman*),leftmargin=*]
\item the chain hypothesis of Theorem~\ref{thm:b} for the layers $L_0, \ldots, L_{n-1}$, widths $s_0, \ldots, s_{n-1}$ and below-vectors $V_1, \ldots, V_n$;
\item the hypothesis of Theorem~\ref{thm:b} on the representative family and the degree-bound function $\mathrm{db}$;
\item $\smPub(\mathit{op}, r, r', \mathit{vd})$, $|\vec{s}| = d$ and $\smEcho(\mathit{op}, w)$;
\item no witness $w'$ satisfies $\smtvalid(\mathit{op}, r, r', w')$.
\end{enumerate}
Then for every prover strategy, the chain event of Definition~\ref{def:chainbad} for these layers and below-vectors and the all-zero claimed output has probability at most
\[
\frac{\smRout + \sum_{i=0}^{n-1} \bigl( 2 s_i \cdot \mathrm{db}(s_i) + 1 \bigr)}{|\F|}
\]
over challenges drawn uniformly and independently from $\F$, $\smRout + \sum_{i=0}^{n-1}(2 s_i + 1)$ of them.
\end{theorem}

\begin{proof}
(1) $C_\tau$ does not accept $x$. Otherwise, with (iii), $\smVAcc(\mathit{op}, r, r', \mathit{vd}, w)$ would hold, and Theorem~\ref{thm:a-sound} would make $\smCanon(\mathit{op}, w)$ a valid witness, contradicting (iv).

(2) The output of $C_\tau$ on $x$ is $V_0 = \mathrm{eval}(L_0, V_1)$, a vector of length $2^{\smRout}$. By (1) some entry of $V_0$ is nonzero, so the all-zero claimed output differs from the true output at some point of $\cube{\smRout}$.

(3) This is the false-output hypothesis of Theorem~\ref{thm:b} for the output layer $L_0$, whose width is $\smRout$ bits. Together with (i) and (ii), Theorem~\ref{thm:b} gives the bound; its output-width term is $\smRout$.
\end{proof}

Hypotheses (i) and (ii) are structural. (i) states fan-in ranges and width agreement of the compiled layers and depends on the concrete stack; (ii) holds with $\mathrm{db}(s) = 2s$ for the representative used in the formal instances (Section~\ref{sec:gkr}). All hypotheses are discharged together for a concrete update over KoalaBear at depth $1$, with the linear hash pair of Section~\ref{sec:compiler}, whose two path equations are incompatible (\isaf{instance\_theorem\_C\_activation} in \texttt{Composition\_Instance.thy}); so they are jointly satisfiable. The specialization $\mathrm{db} = 4$ is mechanized as \isaf{theorem\_C\_composition\_deg4} and has the status of Remark~\ref{cor:b-deg4}.

\paragraph{The per-input form.} Hypothesis (iv) quantifies over all witnesses, and for a compressing node hash it is strong. A sibling list has $dW$ free field elements and a root only $W$, so heuristically, if $\hnode$ behaves like a random function, most roots are reached from a given leaf by some sibling list, and a false claim then has a valid witness. The proof of Theorem~\ref{thm:c} uses (iv) only in step (1), to rule out validity of the single witness $\smCanon(\mathit{op}, w)$. The same proof therefore gives a statement about the witness at hand.

\begin{corollary}[Per-input soundness; paper proof]
\label{cor:perinput}
Assume (H), (S) and hypotheses (i) to (iii) of Theorem~\ref{thm:c}. If $\smtvalid(\mathit{op}, r, r', \smCanon(\mathit{op}, w))$ does not hold, then for every prover strategy the chain event of Theorem~\ref{thm:c} has probability at most the bound of Theorem~\ref{thm:c}.
\end{corollary}

\begin{proof}
If $C_\tau$ accepted $x$, then (iii) would give $\smVAcc(\mathit{op}, r, r', \mathit{vd}, w)$, and Theorem~\ref{thm:a-sound} would make $\smCanon(\mathit{op}, w)$ valid, against the assumption. So $C_\tau$ does not accept $x$, and steps (2) and (3) of the proof of Theorem~\ref{thm:c} apply unchanged.
\end{proof}

\begin{corollary}[Per-input soundness of the verifier; paper proof]
\label{cor:perinput-verifier}
Assume (H), (S) and hypotheses (i) and (iii) of Theorem~\ref{thm:c}, with the claim, the value digest and the witness fixed before the challenges. Let the verifier of Definition~\ref{def:verifier} run on the layers of $C_\tau$ with the input table $x$ and the all-zero claimed output. If $\smtvalid(\mathit{op}, r, r', \smCanon(\mathit{op}, w))$ does not hold, then for every prover strategy whose coins are independent of the challenges, the verifier accepts with probability at most
\[
\frac{\smRout + \sum_{i=0}^{n-1}(8 s_i + 1)}{|\F|}.
\]
\end{corollary}

\begin{proof}
As in the proof of Corollary~\ref{cor:perinput}, $C_\tau$ does not accept $x$, so the all-zero claimed output differs from the true output $V_0$ at some point of $\cube{\smRout}$. By hypothesis (i), the triples $(s_i, L_i, V_{i+1})$ form a layer chain whose last table is the input $x$ that the verifier holds, and $w(L_0) = \smRout$. Proposition~\ref{prop:b-faithful} bounds the probability that the verifier accepts this false output claim by the stated quantity.
\end{proof}

\paragraph{What an accepted run establishes.} Corollary~\ref{cor:perinput-verifier} is the statement of Section~\ref{sec:goal}. Fix the claim, the value digest and the witness before the challenges, with the public-side checks satisfied, and let $\varepsilon$ be its bound. Except with probability at most $\varepsilon$, an accepted run has a valid canonical witness: the root $r$ authenticates $\smOcc(v)$ at the key for a membership claim, an empty or retired leaf for a non-membership claim, and, for an update, the old leaf under $r$ and the new leaf under $r'$ with one sibling list. A second leaf state with a valid path at the same key and root breaks position binding (Definition~\ref{def:pb}). In particular, valid canonical witnesses for $\smMem(k, v)$ and $\smMem(k, v')$ with $v \neq v'$ under one root, or for $\smMem(k, v)$ and $\smNon(k)$ under one root, break it; Corollary~\ref{cor:one-leaf} turns this into a bound on two accepted runs. Corollary~\ref{cor:perinput-ext} states the same bound for a circuit over KoalaBear with challenges drawn from its quartic extension.

\subsection{Per-Input Soundness across the Two Fields}
\label{sec:perinput-ext}

The composition above uses one field for circuit values and challenges, while the implementation evaluates the compiled circuit over $\F_p$ and draws challenges from $\extK$. The following corollary joins \cref{thm:a-sound} over $\F_p$ to the bounds over $\extK$ through \cref{prop:lift-chain}.

\begin{corollary}[Per-input soundness over the quartic extension; paper proof]
\label{cor:perinput-ext}
Instantiate the compiler model of \cref{sec:compiler} with $\F = \F_p$ and assume (H) and (S). Let $\mathit{op}$, $r$, $r'$, $\mathit{vd}$ and a witness $w$ with $d$ siblings satisfy the public-side checks and the echo (\cref{def:verifier-model}), all fixed before the challenges. Let $x = \smEnc(\mathit{op},r,r',\mathit{vd},w)$, let $L_0,\dots,L_{n-1}$ be the layers of $C_\tau$, $\tau = \smKind(\mathit{op})$, with $L_0 = \smRes(\Delta_\tau)$, and let $b_{n-1} = x$ and $b_i = \mathrm{eval}(L_{i+1},b_{i+1})$ for $i<n-1$. Assume that $(s,L,b)$ is a layer chain over $\F_p$ (\cref{def:chain}) for some widths $s_0,\dots,s_{n-1}$; then $s_{n-1} = n_\tau$ and $L_{n-1}$ reads inside $x$. If $\smtvalid(\mathit{op},r,r',\smCanon(\mathit{op},w))$ does not hold, then:
\begin{enumerate}[label=(\alph*),leftmargin=*]
\item the verifier of \cref{def:verifier}, run over $\extK$ on the lifted layers $\operatorname{lift}(L_i)$ with the input table $\iota(x)$ and the all-zero claimed output, accepts with probability at most
\[
\frac{\smRout+\sum_{i\in[n]}(8s_i+1)}{p^4}
\]
over challenges drawn independently and uniformly from $\extK$, for every prover strategy whose coins are independent of the challenges;
\item for every degree-bound function $\dbnd$ and every representative family over $\extK$ admissible for it, the chain event of \cref{def:chainbad} for the lifted chain $(s,\operatorname{lift}(L),\iota(b))$ and the all-zero claimed output has probability at most $\bigl(\smRout+\sum_{i\in[n]}(2s_i\,\dbnd(s_i)+1)\bigr)/p^4$.
\end{enumerate}
\end{corollary}

\noindent It composes the mechanized \cref{thm:a-sound,thm:b-ext4} with \cref{prop:lift-chain,prop:b-faithful}; no single formal statement covers both fields.

\begin{proof}
Over $\F_p$: the output of $C_\tau$ on $x$ is $\mathrm{eval}(L_0,b_0)$. If all its entries were zero, the circuit would accept $x$, the verifier model would accept because its other checks pass, and \cref{thm:a-sound} would make $\smCanon(\mathit{op},w)$ valid. Hence $\mathrm{eval}(L_0,b_0)_z\neq 0$ for some $z<2^{\smRout}$, and the all-zero claimed output is false for $(s,L,b)$.

Lifting: \cref{prop:lift-chain} with $f = 0$ shows that $(s,\operatorname{lift}(L),\iota(b))$ is a layer chain over $\extK$, that $w(\operatorname{lift}(L_0)) = \smRout$, and that $0 = \iota(0)\neq\mathrm{eval}(\operatorname{lift}(L_0),\iota(b_0))_z$. The injectivity of $\iota$ keeps the nonzero output wire nonzero, so the all-zero claimed output is false for the lifted chain at the same wire. Its last table is $\iota(x)$, the input that the verifier holds.

Over $\extK$: \cref{prop:b-faithful}, which holds over every finite field, applied to the lifted chain gives (a), and \cref{thm:b-ext4} gives (b); in both the output term is $\smRout$.
\end{proof}

\noindent By \cref{prop:lift-chain}(a), the layer values of the lifted chain are the embedded base-field values of the compiled circuit on $x$, and the extension of $\iota(x)$ at a point of $\extK^{s_{n-1}}$ is the value that a verifier holding $x$ computes from its entries.

\subsection{Input Vectors Supplied by the Prover}
\label{sec:smt-vector}

Theorems~\ref{thm:a-sound} and~\ref{thm:a-complete} concern input vectors produced by the encoder, which computes the accumulator chains itself. This subsection treats a verifier that computes no node hash: it receives the input vector from the prover, checks the entries that carry public values, and leaves the hash chains to the circuit. It still compares the claimed value digest with leaf hashes of public leaf states (Definition~\ref{def:verifier-model}). The results are proved on paper from Lemma~\ref{lem:circuit-cons}, which is mechanized for every input vector of the right length.

Throughout this subsection a digest is the vector that represents it: $\mathcal{D} = \F^{W}$ and $\smRep$ is the identity. The node part of (H) then reads $\hwn(a, b) = \hnode(a, b)$ for all $a, b \in \F^{W}$, so the node sub-circuit is assumed to compute $\hnode$ on every pair of vectors; the leaf part of (H) is used only at the old leaf of an update. This convention matches the implementation with $W = 8$: a digest there is an array of exactly $8$ base-field elements, the encoder writes its entries in order, and the compression function is defined on every pair of such arrays. That the node sub-circuit computes the compression function on every input is a correctness property of the sub-circuit, assumed here as in (H) and outside the formal development.

\begin{definition}[Input-vector verifier]
\label{def:vector-verifier}
Let $\mathit{op}$ be an operation of kind $\tau$ with key $k$, let $r, r', \mathit{vd} \in \F^{W}$, and let $x \in \F^{2^{n_\tau}}$ be arbitrary, with the segments $P, A_j, S_j, \beta_j, R, V$ (and $A'_j, R'$ for updates) of Definition~\ref{def:residuals}.
\begin{enumerate}[label=(\alph*),leftmargin=*]
\item The \emph{public segment checks} hold iff $\beta_j = \smBit_j(k)$ for all $j < d$, $R = r$ and $V = \mathit{vd}$, and, if $\mathit{op} = \smUpd(k, \ell_{\mathrm o}, \ell_{\mathrm n})$, also $R' = r'$ and $P = \smPre(\ell_{\mathrm o})$.
\item The input-vector verifier accepts, $\smVecAcc(\mathit{op}, r, r', \mathit{vd}, x)$, iff $\smPub(\mathit{op}, r, r', \mathit{vd})$ and the public segment checks hold and $\smAcc(C_\tau, x)$.
\item The \emph{witness read from $x$} is $\smWit(\mathit{op}, \mathit{vd}, x) = (\ell^{\ast}, (S_0, \ldots, S_{d-1}))$, where $\ell^{\ast} = \smOcc(v)$ if $\mathit{op} = \smMem(k, v)$; $\ell^{\ast} = \smEmpty$ if $\mathit{op} = \smNon(k)$ and $\mathit{vd} = \hleaf(\smEmpty)$, and $\ell^{\ast} = \smTomb$ for other non-membership cases; and $\ell^{\ast} = \ell_{\mathrm o}$ if $\mathit{op} = \smUpd(k, \ell_{\mathrm o}, \ell_{\mathrm n})$.
\end{enumerate}
\end{definition}

\noindent Entries of $x$ outside these segments, such as the zero padding, are not checked; by Lemma~\ref{lem:circuit-cons}, acceptance by $C_\tau$ does not depend on them. For a non-membership claim that passes $\smPub$, the leaf $\ell^{\ast}$ is the empty or retired state whose leaf hash is $\mathit{vd}$.

\begin{theorem}[Soundness on input vectors; paper proof]
\label{thm:a-vector}
Assume (S), $\mathcal{D} = \F^{W}$ with $\smRep$ the identity, and two instances of (H): $\hwn(a, b) = \hnode(a, b)$ for all $a, b \in \F^{W}$, and, if $\mathit{op} = \smUpd(k, \ell_{\mathrm o}, \ell_{\mathrm n})$, $\hwl(\smPre(\ell_{\mathrm o})) = \hleaf(\ell_{\mathrm o})$. If $\smVecAcc(\mathit{op}, r, r', \mathit{vd}, x)$, then $\smtvalid(\mathit{op}, r, r', \smWit(\mathit{op}, \mathit{vd}, x))$.
\end{theorem}

\begin{proof}
Write $\vec{s} = (S_0, \ldots, S_{d-1})$; each $S_j$ is a digest since $\mathcal{D} = \F^{W}$, and $|\vec{s}| = d$. The vector $x$ has length $2^{n_\tau}$ and $C_\tau$ accepts it, so Lemma~\ref{lem:circuit-cons} gives $\smCons_\tau(x)$; this step is mechanized (\isaf{A\_i\_ro}, \isaf{A\_i\_up}, read as explained after \cref{lem:circuit-cons}) and holds for every input vector. The public checks give $k < 2^d$ and $\beta_j = \smBit_j(k) \in \{0, 1\}$ for all $j < d$.

\emph{Chain claim.} If $Y_0, \ldots, Y_d \in \F^{W}$ satisfy $Y_{j+1} = \hwn\bigl(\smMuxL(\beta_j, Y_j, S_j), \smMuxR(\beta_j, Y_j, S_j)\bigr)$ for all $j < d$, then $Y_d = \smPathRoot(Y_0, k, \vec{s})$. Let $(a_0, \ldots, a_d)$ be the accumulator chain of $(Y_0, k, \vec{s})$ (Definition~\ref{def:smt}~(f)); we show $Y_j = a_j$ by induction on $j$. For $j = 0$ both are $Y_0$. Assume $Y_j = a_j$. If $\smBit_j(k) = 0$, the multiplexer selects $(Y_j, S_j)$ (\isaf{mux\_zero}), and the node part of (H) gives $Y_{j+1} = \hwn(Y_j, S_j) = \hnode(a_j, S_j) = a_{j+1}$. If $\smBit_j(k) = 1$, it selects $(S_j, Y_j)$ (\isaf{mux\_one}), and $Y_{j+1} = \hnode(S_j, a_j) = a_{j+1}$. Hence $Y_d = a_d = \smPathRoot(Y_0, k, \vec{s})$.

\emph{Read-only kinds.} (C4) and the check $V = \mathit{vd}$ give $A_0 = \mathit{vd}$. By (C2) the segments $A_0, \ldots, A_d$ satisfy the hypothesis of the chain claim, so $A_d = \smPathRoot(\mathit{vd}, k, \vec{s})$, and (C3) with the check $R = r$ gives $\smPathRoot(\mathit{vd}, k, \vec{s}) = r$.
\begin{itemize}[leftmargin=*]
\item $\mathit{op} = \smMem(k, v)$. $\smPub$ gives $r' = r$ and $\mathit{vd} = \hleaf(\smOcc(v))$. Hence $\smPathOK(\smOcc(v), k, \vec{s}, r)$, and with $r' = r$ this is $\smtvalid(\mathit{op}, r, r', (\smOcc(v), \vec{s}))$.
\item $\mathit{op} = \smNon(k)$. $\smPub$ gives $r' = r$ and $\mathit{vd} \in \{\hleaf(\smEmpty), \hleaf(\smTomb)\}$, so $\ell^{\ast} \in \{\smEmpty, \smTomb\}$ and $\hleaf(\ell^{\ast}) = \mathit{vd}$. Hence $\smPathOK(\ell^{\ast}, k, \vec{s}, r)$ and $\smtvalid(\mathit{op}, r, r', (\ell^{\ast}, \vec{s}))$.
\end{itemize}

\emph{Updates}, $\mathit{op} = \smUpd(k, \ell_{\mathrm o}, \ell_{\mathrm n})$. The check $P = \smPre(\ell_{\mathrm o})$, (C1) for the first chain and the leaf part of (H) give $A_0 = \hwl(\smPre(\ell_{\mathrm o})) = \hleaf(\ell_{\mathrm o})$. The chain claim with (C2), then (C3) with $R = r$, give $\smPathRoot(\hleaf(\ell_{\mathrm o}), k, \vec{s}) = A_d = r$. For the second chain, (C4$'$) with $V = \mathit{vd}$ and $\smPub$ give $A'_0 = \mathit{vd} = \hleaf(\ell_{\mathrm n})$, and the chain claim with (C2), then (C3) with $R' = r'$, give $\smPathRoot(\hleaf(\ell_{\mathrm n}), k, \vec{s}) = A'_d = r'$. With $k < 2^d$ and $|\vec{s}| = d$ these are $\smPathOK(\ell_{\mathrm o}, k, \vec{s}, r)$ and $\smPathOK(\ell_{\mathrm n}, k, \vec{s}, r')$; since $\ell^{\ast} = \ell_{\mathrm o}$, $\smtvalid(\mathit{op}, r, r', (\ell_{\mathrm o}, \vec{s}))$ holds.
\end{proof}

\begin{remark}[Which constraints carry Theorem~\ref{thm:a-vector}]
\label{rem:vector-constraints}
The proof uses (C2) at every level of each chain, (C3) for each root, (C4) for read-only kinds and (C4$'$) for updates, and (C1) only for the first chain of an update, at $P = \smPre(\ell_{\mathrm o})$. The key-bit check replaces (C5); reading the non-membership leaf from $\mathit{vd}$, which $\smPub$ confines to the leaf hashes of the empty and retired states, replaces (C6); and (C1) is not needed for chains that start at $V$. Theorem~\ref{thm:a-sound} uses only (C3), (C4) or (C4$'$), and (C6) with the public-side checks: on encoder inputs (C1), (C2) and (C5) hold by construction (Lemma~\ref{lem:smt-honest-enc}), so that theorem would hold even for a circuit without them. Here the chain constraint is indispensable: without (C2), for $d \geq 1$ and any membership claim that passes $\smPub$, the vector with public segments $\beta_j = \smBit_j(k)$, $R = r$, $V = \mathit{vd}$ and with $P = \smPre(\smOcc(v))$, $A_0 = \mathit{vd}$, $A_1 = \cdots = A_d = r$ and all $S_j = 0$ satisfies every remaining constraint, (C1) by the leaf part of (H) at $\smOcc(v)$, while the read witness $(\smOcc(v), (0, \ldots, 0))$ is valid only if $r = \smPathRoot(\hleaf(\smOcc(v)), k, (0, \ldots, 0))$. The check $P = \smPre(\ell_{\mathrm o})$ is needed for the same reason: without it, the first chain of an update could start from the pre-image of another leaf. The hash assumption also differs: Theorem~\ref{thm:a-sound} allows any injective $\smRep$, whereas Theorem~\ref{thm:a-vector} needs $\smRep$ to be the identity on $\F^{W}$, so that the node part of (H) covers every pair of vectors.
\end{remark}

\begin{corollary}[Per-input soundness for prover-supplied vectors; paper proof]
\label{cor:perinput-vector}
Let $\F$ be a finite field, and assume the hypotheses of Theorem~\ref{thm:a-vector}. Let $L_0 = \smRes(\Delta_\tau), L_1, \ldots, L_{n-1}$ be the layers of $C_\tau$, and assume the chain hypothesis (i) of Theorem~\ref{thm:c} with widths $s_0, \ldots, s_{n-1}$, so that $s_{n-1} = n_\tau$; with the true value vectors below $L_0$, its conditions concern fan-in ranges and widths only and do not depend on the input. Consider the verifier that receives $x \in \F^{2^{n_\tau}}$ as the prover's first message, rejects unless $\smPub(\mathit{op}, r, r', \mathit{vd})$ and the public segment checks hold, and otherwise runs the verifier of Definition~\ref{def:verifier} on $C_\tau$ with input table $x$ and the all-zero claimed output, drawing its challenges uniformly and independently after $x$ is sent. Fix $x$. If $\smtvalid(\mathit{op}, r, r', \smWit(\mathit{op}, \mathit{vd}, x))$ does not hold, then for every prover strategy for the remaining messages this verifier accepts with probability at most
\[
\frac{\smRout + \sum_{i=0}^{n-1} (8 s_i + 1)}{|\F|}.
\]
Consequently, for every prover strategy, the probability that the verifier accepts while the witness read from the first message is not valid is at most the same bound.
\end{corollary}

\begin{proof}
If the public checks fail, the verifier rejects. Otherwise $C_\tau$ does not accept $x$: if it did, $\smVecAcc(\mathit{op}, r, r', \mathit{vd}, x)$ would hold, and Theorem~\ref{thm:a-vector} would make $\smWit(\mathit{op}, \mathit{vd}, x)$ valid. So the true output $V_0 = \mathrm{eval}(L_0, V_1)$ of $C_\tau$ on $x$ has a nonzero entry, and the all-zero claimed output is false. By Lemma~\ref{lem:accept-bad}, acceptance of this false claim by the verifier of Definition~\ref{def:verifier} implies the chain event for the integrand itself, and Proposition~\ref{prop:b-faithful} bounds its probability by the displayed quantity. There $\smRout$ is the output width of $C_\tau$ in bits, and $8 s_i = 2 s_i \cdot 4$ accounts for the $2 s_i$ variables of the sumcheck of layer $i$, in each of which the integrand has degree at most $4$. For the last statement, the challenges are independent of the first message; the bound holds for each fixed $x$ and each fixed choice of the prover's internal randomness, and averaging over them gives the bound for the joint event.
\end{proof}

\begin{corollary}[Per-input soundness for prover-supplied vectors over the quartic extension; paper proof]
\label{cor:perinput-vector-ext}
Instantiate the compiler model of \cref{sec:compiler} with $\F = \F_p$, and assume the hypotheses of Theorem~\ref{thm:a-vector} and the chain hypothesis of Corollary~\ref{cor:perinput-vector} with widths $s_0, \ldots, s_{n-1}$. Consider the verifier that receives $x \in \F_p^{2^{n_\tau}}$ as the prover's first message, rejects unless $\smPub(\mathit{op}, r, r', \mathit{vd})$ and the public segment checks hold, and otherwise runs the verifier of Definition~\ref{def:verifier} over $\extK$ on the lifted layers $\operatorname{lift}(L_i)$ with the input table $\iota(x)$ and the all-zero claimed output, drawing its challenges independently and uniformly from $\extK$ after $x$ is sent. Fix $x$. If $\smtvalid(\mathit{op}, r, r', \smWit(\mathit{op}, \mathit{vd}, x))$ does not hold, then for every prover strategy for the remaining messages this verifier accepts with probability at most
\[
\frac{\smRout + \sum_{i=0}^{n-1} (8 s_i + 1)}{p^4}.
\]
Consequently, for every prover strategy, the probability that the verifier accepts while the witness read from the first message is not valid is at most the same bound.
\end{corollary}

\begin{proof}
If the public checks fail, the verifier rejects. Otherwise, as in the proof of Corollary~\ref{cor:perinput-vector}, $C_\tau$ does not accept $x$ over $\F_p$, so its output on $x$ has a nonzero entry. Let $b_{n-1} = x$ and $b_i = \mathrm{eval}(L_{i+1}, b_{i+1})$ for $i < n-1$; the chain hypothesis makes $(s, L, b)$ a layer chain over $\F_p$, and the all-zero claimed output is false for it. From here the argument is that of part (a) of Corollary~\ref{cor:perinput-ext}: \cref{prop:lift-chain} with $f = 0$ carries the chain and the false output to $\extK$, the last table of the lifted chain is $\iota(x)$, and \cref{lem:accept-bad,prop:b-faithful} over $\extK$ give the bound. The last statement follows by averaging, as in Corollary~\ref{cor:perinput-vector}.
\end{proof}

\subsection{One Leaf per Key and Root}
\label{sec:one-leaf}

The corollaries above bound one run. A system also relies on accepted runs for a key and a root asserting one leaf, and this needs the assumption of \cref{sec:semantics}. Call the \emph{asserted witness} of a run the canonical witness $\smCanon(\mathit{op}, w)$ for the verifiers of \cref{cor:perinput-verifier,cor:perinput-ext} and the read witness $\smWit(\mathit{op}, \mathit{vd}, x)$ for the verifiers of \cref{cor:perinput-vector,cor:perinput-vector-ext}. For every kind of operation, validity of the asserted witness contains $\smPathOK$ for its leaf at the key of the operation under the claimed root (\cref{def:smt}).

\begin{corollary}[One leaf per key and root; paper proof]
\label{cor:one-leaf}
Fix one of the verifiers of \cref{cor:perinput-verifier,cor:perinput-ext,cor:perinput-vector,cor:perinput-vector-ext}, the first two of which are taken to reject a first message that fails the public-side checks, the length check or the echo. Assume the other hypotheses of that corollary for all three circuits $C_\tau$ and for every operation that $\mathcal{A}$ may claim, and let $\varepsilon$ be the largest of its bounds over the three kinds. Let $\mathcal{A}$ play the prover in two runs of this verifier with independent challenges, fixing the claim, the value digest and the witness or input vector of each run before the challenges of that run are drawn. Let $E$ be the event that both runs accept, that their claims have the same key and the same root, and that their asserted witnesses have different leaves. Then
\[
\Pr[E] \leq 2\varepsilon + \mathsf{Adv}^{\mathrm{pb}}(\mathcal{B}),
\]
where $\mathcal{B}$ runs $\mathcal{A}$, simulates the two verifiers and outputs the key, the root and the leaves and sibling lists of the two asserted witnesses.
\end{corollary}

\begin{proof}
Let $G_i$ be the event that the asserted witness of run $i$ makes its transition valid. On $E$, either $G_1$ and $G_2$ both hold or some run $i$ accepts without $G_i$. In the first case both asserted witnesses satisfy $\smPathOK$ at the common key and root, and their leaves differ, so $\mathcal{B}$ breaks position binding. For the second case fix $i$, the coins of $\mathcal{A}$ and the challenges of the other run; these are independent of the challenges of run $i$. They determine the claim, the value digest and the first message of run $i$ and leave $\mathcal{A}$ a prover strategy for the remaining messages. If the checks that precede the protocol fail on that first message, the run rejects. If they pass and the asserted witness is not valid, the corollary for the chosen verifier bounds acceptance by $\varepsilon$. Averaging gives $\Pr[\text{run } i \text{ accepts} \wedge \neg G_i] \leq \varepsilon$, and a union bound gives the claim.
\end{proof}

\noindent The verifier reads the asserted witness in the clear, so the reduction needs no extractor. The same proof covers the new leaf of an update under its post-root, since validity also contains $\smPathOK(\ell_{\mathrm n}, k, \vec{s}, r')$.

\section{What Is Machine-Checked}
\label{sec:mechanization}

The results of \cref{sec:gkr,sec:smallfield,sec:semantics,sec:compiler,sec:composition} are proved in this paper. This section states which of them are also checked by Isabelle/HOL, so that a reader can separate the two; \cref{tab:mechanization} in \cref{app:mechanization} lists every numbered result with its status and the corresponding Isabelle theorems.

\paragraph{The development.} The formalization is written in Isabelle/HOL~\cite{isabellehol2002} (Isabelle2025-2) on top of the sumcheck entry of the Archive of Formal Proofs~\cite{sumcheckafp2024}, release \texttt{afp-2026-08-05}. It has six sessions with twenty theories (11,112 lines, 418 named statements): one for layered circuits, multilinear extensions and the chain bound; one for the sparse-Merkle semantics, the compiler model, the arithmetization theorem and the composition; and four that build the quartic extension and instantiate the chain bound over it. No theory contains an unfinished proof (\isaf{sorry}, \isaf{oops}), an axiom (\isaf{axiomatization}) or an oracle.

\paragraph{Status marks.} Each theorem, lemma, proposition and corollary of ours carries a mark in its heading. \emph{Mechanized} means that an Isabelle theorem states the result up to notation, with the same or weaker hypotheses; where the formal statement differs in another way, the text after the result says how. \emph{Paper proof} marks a result proved only in this paper, from mechanized results and standard facts, and \emph{partly mechanized} results say which parts are which.

\paragraph{Nonvacuity.} A theorem proved inside an Isabelle locale is an implication from the locale's assumptions. Every locale used above is therefore interpreted by a concrete instance over the KoalaBear field, and the arithmetization and composition theorems are applied, for a tree of depth one with a linear hash pair, to inputs whose premises they satisfy.

\paragraph{What is not used.} The development also contains a degree-four specialization of the chain bound, which holds only vacuously (\cref{cor:b-deg4}), and conditional theorems about the verifier that ships with the artifact (\isaf{verifier\_acceptance\_implies\_gkr\_chain\_bad} and its corollary \isaf{production\_verifier\_acceptance\_refines\_gkr\_chain\_bad} in \texttt{Verifier\_Acceptance\_Refinement.thy}), whose premises, a per-layer bridge in the theorem and the relation \isaf{rust\_hol\_value\_trace\_relation} between the Rust values and the model in the corollary, are not discharged and, with the multilinear representative, fail on honest runs in general (\cref{sec:gkr-faithful}), so these theorems do not apply to the shipped verifier. Contracts checked with Creusot and Why3 relate parts of the Rust code to model clauses; this relation is partial. No result of this paper uses these.

\paragraph{Artifact.} The Isabelle sessions, the Rust workspace and the verification scripts are available at \url{https://github.com/Oraclizer/statesync-gkr} and archived as StateSync-GKR v1.1.1, \href{https://doi.org/10.5281/zenodo.23280552}{doi:10.5281/zenodo.23280552}. A clean build of the six sessions is the reproduction check:
\begin{center}
\isaf{isabelle build -c -d /path/to/afp/thys -d formal/isabelle SESSION}
\end{center}
The results were obtained on the signed release v1.1.0, commit \texttt{1b8d2f829792172b347dedbfde969016c2c05789}, recorded as \href{https://doi.org/10.5281/zenodo.23136385}{doi:10.5281/zenodo.23136385} with restricted file access; v1.1.1 changes the package license metadata and retains the theory files of v1.1.0.

\section{Limitations}
\label{sec:limitations}

This section collects the boundaries of the results in one place.

\paragraph{Interactive model.} All probability statements concern the interactive public-coin protocol with challenges drawn uniformly and independently from the field, matching the imported sumcheck development. The Fiat--Shamir transformation~\cite{fiatshamir1986} that makes the running protocol non-interactive is outside the analysis. Fiat--Shamir is delicate for GKR-based arguments: for every choice of the Fiat--Shamir hash function, Khovratovich, Rothblum and Soukhanov construct a circuit, depending on that hash, for which the compiled argument accepts a false statement~\cite{khovratovich2025fs}. A security claim about the non-interactive protocol would need a multi-round reduction with explicit assumptions and quantitative loss; we make no such claim. The uniform distribution of an extension challenge drawn as four uniform base-field coordinates is proved (\cref{lem:pushforward}); whether the implementation's transcript-derived samples are uniform and independent, conditioned on earlier messages and on the input, is a property of its Fiat--Shamir construction and is not established here.

\paragraph{Circuit inputs.} The mechanized arithmetization theorem (\cref{thm:a-sound,thm:a-complete}) concerns circuit inputs produced by the witness encoder of \cref{def:encoding}, which computes the hash accumulators from the witness; the verifier of \cref{sec:goal} computes that input itself, as the shipped verifier does. On such inputs the hash-chain and bit residuals hold by construction, so soundness rests on the boundary residuals and the public checks. For a verifier that receives the input vector from the prover, \cref{thm:a-vector,cor:perinput-vector,cor:perinput-vector-ext} are paper proofs. They assume in addition that a digest is represented by its own field elements, that the node sub-circuit computes the compression function on every pair of vectors and, for updates, that the leaf sub-circuit computes the leaf hash of the old leaf, and they bind the key bits, the roots, the value digest and, for updates, the pre-image of the old leaf by checks on public segments of the input (\cref{rem:vector-constraints}). A verifier that receives only a commitment to the input vector is not treated.

\paragraph{Paper steps.} Every result marked P in \cref{tab:mechanization} is proved on paper only, from mechanized results and standard facts. The main statements rely on the following ones: the degree of the layer integrand (\cref{lem:integrand-poly}) and the application of the chain bound to the integrand that the verifier evaluates (\cref{lem:accept-bad,prop:b-faithful}); the lifting of a whole chain into the extension (\cref{prop:lift-chain}), the extension of \cref{lem:pushforward} to several independent challenges, and the equivalence of the two forms of chain well-formedness (\cref{app:ext-proofs}); the arithmetization theorem for inputs supplied by the prover (\cref{thm:a-vector}); the per-input forms over one field and across the two fields (\cref{cor:perinput,cor:perinput-verifier,cor:perinput-ext,cor:perinput-vector,cor:perinput-vector-ext}); root binding and the bound on two accepted runs (\cref{prop:binding,cor:one-leaf}); and the parts of \cref{prop:ext-field,lem:lift-layer} marked as such. The representative with which the development instantiates the mechanized chain bound is the multilinear interpolation of the layer integrand, which agrees with the integrand on the Boolean hypercube but in general not at the random point of the verifier's final check; this is why the faithful application is a separate step.

\paragraph{Hash functions.} The leaf and node hashes are abstract functions in the compiler model. Their in-circuit realization by Poseidon2 sub-circuits is not verified. The node hash of the implementation is not collision resistant (\cref{sec:semantics}); position binding of the tree (\cref{def:pb}) is an assumption of this paper, enters only \cref{cor:one-leaf}, and is outside the formal development. \Cref{lem:leaffold} reduces leaf-hash collisions on in-bound encodings to sponge collisions; the reduction from sponge collisions to properties of the permutation is outside the development.

\paragraph{Implementation.} No result of this paper reaches the Rust verifier: the conditional theorems about it do not apply to the shipped verifier, and the contract layer covers the code only in part (\cref{sec:mechanization}). A use of these results for the Rust verifier would also trust the underlying library's field arithmetic, challenger and Poseidon2 implementation, the Rust compiler and the operating system.

\paragraph{Other properties.} The core protocol provides no zero knowledge, and none is claimed. The per-input bounds read the witness from the prover's messages (\cref{sec:goal}), so no extractor is constructed, and no knowledge soundness is claimed beyond them. Completeness is proved for the arithmetization (\cref{thm:a-complete}) and for each layer step (\cref{lem:layer-identity}); an assembled completeness theorem for a full protocol run is not part of the development. The semantics covers the three sparse-Merkle operations over three-state leaves at a fixed but arbitrary depth. This paper makes no performance claim.

\section{Related Work}
\label{sec:related}

\paragraph{The GKR lineage.} The protocol originates with Goldwasser, Kalai, and Rothblum~\cite{gkr2015}; Thaler~\cite{thaler2013} gives time-optimal provers for regular wiring, the regime targeted by our compilation. The sumcheck protocol descends from Lund, Fortnow, Karloff, and Nisan~\cite{lfkn1992}, with polynomial identity bounds from Schwartz~\cite{schwartz1980}. This paper proves, and machine-checks the core of, the soundness of the protocol with the gate set and the field of a concrete system, composed with a constraint-level model of its arithmetization.

\paragraph{Computer-aided cryptography for proof systems.} The AFP sumcheck entry~\cite{sumcheckafp2024} and its accompanying analysis~\cite{bosshard2024} provide our per-round foundation, and their conclusion names GKR among the sumcheck-based protocols left for analysis. In Lean, complexitylib~\cite{complexitylib2026} proves perfect completeness and soundness of the GKR protocol for layered circuits with addition and multiplication gates over an arbitrary finite field, composing round-by-round soundness of its stages, with error $\bigl(w_{\mathrm{out}} + \sum_i 5 w_i\bigr)/|\F|$ for layers of $2^{w_i}$ gates; round-by-round soundness is stronger than the soundness bounds of this paper. The Lean library of the GordieHowe proof system~\cite{sharpgrubbs2026} machine-checks the security theorems of a GKR-based argument, from round-by-round knowledge soundness to the non-interactive argument, and, for ML-DSA-44, proves the source program complete and round-by-round knowledge sound for the FIPS 204 verification relation, with ExpandA bounded as FIPS 204 Appendix C permits, and carries this to the non-interactive argument. The theorems of its applications do not take as hypotheses that each compiled circuit enforces the checks assigned to it, accepts the honest input and stays within its listed layer and wire counts: these facts are proved from certificates over the exported circuits, which the Lean kernel checks. leanerVM~\cite{leanervm2026} proves round-by-round knowledge soundness of the oracle protocol of leanVM, in which a grand-product GKR checks the table bus, at every instance meeting the side conditions of that protocol, and proves that, on admissible sizes, a stack satisfying the protocol's relation for its instruction set yields a witness of the constraint system. At the time of writing, open pull requests prove that every satisfying witness of that constraint system represents a valid execution and, under a size hypothesis, the converse, and add knowledge soundness of a Lean implementation of the non-interactive verifier in the random-oracle model. That end-to-end statement, once merged, is stronger than ours, for a virtual machine and with GKR in a different role. Longfellow~\cite{longfellowformal2026} proves in Lean, for a model of its version-7 verifier, conditional extraction and false-statement soundness with error $(1+N+t_A)/2^{116}$ for $N$ transcript-oracle and $t_A$ Merkle-hash queries, along a chain from a trusted circuit of quadratic gates through a padded layer-by-layer sumcheck, Ligero, Fiat--Shamir, a two-field MAC and a Merkle extraction reduction. That statement is also stronger than ours in reaching the non-interactive verifier and the commitment. Its repository also holds a separate Lean soundness proof for the constraints of an ECDSA circuit, outside that chain. Our development proves the checks of the compiled circuit exact, a step that the GKR development of complexitylib does not include and that the main chain of Longfellow replaces by a trusted circuit. The GordieHowe library machine-checks certificates for the fixed circuits of its applications; our proof, on paper from a mechanized characterization of the constraints, covers sparse-Merkle circuits of every depth in a constraint-level model with only the hash sub-circuits abstract, and stops at the interactive protocol, where the GordieHowe library reaches the non-interactive argument. It differs from complexitylib also in its gate set, which includes the cube gate, and in carrying base-field circuits into the exact extension field from which a small-field prover draws its challenges, a step proved in part on paper (\cref{lem:lift-layer,prop:lift-chain}). ArkLib~\cite{arklib} develops interactive oracle reductions in Lean~4, with sumcheck and further protocols under active formalization, and Lochbihler and Mari\'c~\cite{lochbihlermaric2020} formalize authenticated data structures as functors in Isabelle/HOL.

\paragraph{Correctness of arithmetizations.} Under-constrained circuits are the most prevalent vulnerability class within ZK circuits in the systematization of Chaliasos et al.~\cite{chaliasos2024sok}. Coda~\cite{coda2024} certifies reimplementations of Circom circuits with refinement types, and Picus~\cite{picus2023} searches for under-constrained circuits automatically; both consider arbitrary assignments to the circuit's witness. Our mechanized arithmetization theorem is narrower than a statement about arbitrary assignments: it concerns the input that the verifier encodes from the witness. For input vectors supplied by the prover we give a paper proof (\cref{thm:a-vector}) from a mechanized characterization of acceptance that holds for every input vector (\cref{lem:circuit-cons}). This paper composes such a theorem with a GKR bound, machine-checked in the existential form and proved on paper per input, and analyzes a verifier that leaves the node hashes to the circuit.

\paragraph{Finite fields in proof assistants.} The Archive of Formal Proofs contains a classification of finite fields up to isomorphism~\cite{karayel2022finitefields}. The tower of \cref{sec:smallfield} is a single concrete construction, the binomial extension used by the implementation in its exact coordinates, built as a type so that the general theorems can be instantiated at it.

\section{Conclusion}
\label{sec:conclusion}

A GKR verifier's acceptance is a statement about a circuit, and a system that acts on it needs a statement about an update. For sparse-Merkle membership, non-membership and update we have proved the route between the two: a soundness bound for the GKR protocol with a cube gate, its instantiation over the exact quartic extension of a 31-bit field, the exactness of a verifier model that encodes the witness, with a soundness theorem for inputs supplied by the prover, and their composition into a per-input bound, to which position binding of the tree adds that accepted runs for a key and a root assert one leaf. The chain bound, its extension-field instance, the arithmetization theorem and, over one field, the existential composition are machine-checked in Isabelle/HOL; the remaining steps are proved on paper in this text. The general results are independent of the workload: another application inherits the composition by proving its own arithmetization theorem in the form of \cref{thm:a-sound}, and another challenge field needs a field construction and the lifting lemmas of \cref{sec:smallfield}, with no second probabilistic proof.

\section*{Acknowledgments}
We thank Bernhard M{\"u}ller for his independent review of the verifier's challenge-generation assumptions and the implementation-facing interpretation of the extension-field soundness bound. His analysis clarified the distinction between the mechanized field and model results and the remaining multi-round Fiat--Shamir and input-binding arguments. Responsibility for the implementation and the remaining security arguments rests with the author.

\section*{Declarations}
\paragraph{Use of AI tools.} AI-use designation under the IACR policy on generative AI: \AIdesignation. Claude (Anthropic) and GPT (OpenAI) drafted the manuscript, including the paper proofs, wrote the Isabelle proof scripts, the Rust code and its contracts, and were used to check the paper proofs against the formal development and to review the manuscript. All AI-assisted work was directed by the author, who defined the research questions, the system design and the scope of the verification, decided which statements are claimed, and reviewed the manuscript. The Isabelle kernel checks every mechanized proof independently of how the proof was written. The AI tools are not authors, and the author takes full responsibility for the work.
\paragraph{Competing interests.} Jinwook Kim is a co-founder, an officer (Chief Product and Technology Officer) and a shareholder of Oraclizer Labs, Inc., which develops the StateSync-GKR software studied here and offers commercial licenses for it under the Business Source License~1.1.
\paragraph{Funding.} No external funding was received for this work. It was carried out with internal funds of Oraclizer Labs, Inc.
\paragraph{Code availability.} The Isabelle sessions, the Rust workspace and the verification scripts are available as described in \cref{sec:mechanization}, under the Business Source License~1.1, which permits non-production research and reproduction; rights previously granted under MIT or Apache-2.0 for v1.1.0 remain valid.

\bibliographystyle{plain}
\bibliography{references}

\appendix
\crefalias{section}{appendix}

\section{Proofs for the GKR Soundness Bound}
\label{app:gkr-proofs}

This appendix gives the complete paper proofs of Theorems~\ref{thm:sumcheck-afp} and~\ref{thm:sumcheck-std}, Lemma~\ref{lem:chain-bound} and Theorem~\ref{thm:b}. In the last three, $q = |\F|$.

\subsection{Proof of \texorpdfstring{Theorem~\ref{thm:sumcheck-afp}}{the sumcheck bound in the form of the AFP}}

The mechanized proof is \isaf{soundness\_inductive} of the AFP entry~\cite{sumcheckafp2024}, which our development imports. We give a proof from the axioms of abstract polynomials listed in Section~\ref{sec:prelim}.

\begin{proof}[Proof of Theorem~\ref{thm:sumcheck-afp}]
For a finite set $U$ of variables and $p$ with $\mathrm{vars}(p)\subseteq U$ write $\Sigma_U(p) = \sum_{\sigma\in H^U}p(\sigma)$. We show by induction on $m$ that for all $p$, $v$, $\mathrm{pr}$, $\mathit{st}$ and $r'$ with $\mathrm{vars}(p)\subseteq\{x_1,\dots,x_m\}$, $\deg p\le d$ and $v\neq\Sigma_{\{x_1,\dots,x_m\}}(p)$, the acceptance event has probability at most $md/|\F|$.

For $m = 0$, $\mathrm{vars}(p) = \emptyset$ and $H^{\emptyset} = \{\emptyset\}$, so $v\neq p(\emptyset)$ and acceptance, which requires $v = p(\emptyset)$, is impossible.

For $m\geq 1$, let $U' = \{x_2,\dots,x_m\}$ and $(q,\mathit{st}') = \mathrm{pr}((H,p,v),x_1,(x_2,\dots,x_m),r',\mathit{st})$, and let $g = \sum_{\sigma\in H^{U'}}p[\sigma]$ be the true round polynomial. By the axioms on variables, $\mathrm{vars}(g)\subseteq\{x_1\}$; by the axioms on degree, $\deg g\le\deg p\le d$; and by the axioms on evaluation, $g([x_1\mapsto h]) = \Sigma_{U'}(p[x_1\mapsto h])$ for every $h\in\F$, so that $\sum_{h\in H}g([x_1\mapsto h]) = \Sigma_{\{x_1,\dots,x_m\}}(p)\neq v$. If the run accepts, the first round gives $\mathrm{vars}(q)\subseteq\{x_1\}$, $\deg q\le\deg p\le d$ and $\sum_{h\in H}q([x_1\mapsto h]) = v$, hence $q\neq g$; and the remaining rounds accept the instance $(H,p[x_1\mapsto r_1],q([x_1\mapsto r_1]))$, whose polynomial has variables in $U'$, degree at most $d$ and true sum $\Sigma_{U'}(p[x_1\mapsto r_1]) = g([x_1\mapsto r_1])$. So acceptance implies one of two events. In the first, $q([x_1\mapsto r_1]) = g([x_1\mapsto r_1])$ with $q\neq g$; $q$ and $g$ are fixed before $r_1$ is drawn, and by the root bound this event has probability at most $d/|\F|$ over $r_1$. In the second, $q([x_1\mapsto r_1])\neq g([x_1\mapsto r_1])$ and the remaining rounds accept a false claim; for each fixed $r_1$ the induction hypothesis bounds its probability over $r_2,\dots,r_m$ by $(m-1)d/|\F|$, and averaging over $r_1$ by Lemma~\ref{lem:split}(a) preserves the bound. The union bound gives $md/|\F|$.
\end{proof}

\subsection{Proof of \texorpdfstring{Theorem~\ref{thm:sumcheck-std}}{the per-variable sumcheck bound}}

This is the standard argument~\cite{lfkn1992,thaler2013}, recalled for completeness; it is not mechanized in our development.

\begin{proof}[Proof of Theorem~\ref{thm:sumcheck-std}]
Fix the strategy $\psi$. For $\mathbf{r}\in\F^m$ let $g_t = \psi(\mathbf{r}_{[0,t)})$ and $v_t$ be as in Definition~\ref{def:sumcheck-D}, and define the true round polynomials and true partial sums
\[
H_t(X) = \sum_{\mathbf{b}\in\cube{m-t-1}} g(\mathbf{r}_0,\dots,\mathbf{r}_{t-1},X,\mathbf{b})\quad (t\in[m]),\qquad
S_t = \sum_{\mathbf{b}\in\cube{m-t}} g(\mathbf{r}_0,\dots,\mathbf{r}_{t-1},\mathbf{b})\quad (t\in[m+1]).
\]
Then $S_0 = \sum_{\mathbf{b}\in\cube{m}}g(\mathbf{b})$, $S_m = g(\mathbf{r})$, $S_t = H_t(0)+H_t(1)$ and $S_{t+1} = H_t(\mathbf{r}_t)$. Both $g_t$ and $H_t$ depend only on $\mathbf{r}_{[0,t)}$, and $\deg H_t\le\deg_{X_t}g\le D_t$.

Suppose $\cSC^{\mathbf{D}}(g,v,\psi,\mathbf{r})$ holds. Then $v_0 = v\neq S_0$ and $v_m = g(\mathbf{r}) = S_m$. Let $t$ be the least index with $v_{t+1} = S_{t+1}$; then $v_t\neq S_t$. The round check gives $g_t(0)+g_t(1) = v_t\neq S_t = H_t(0)+H_t(1)$, so $g_t\neq H_t$, and $g_t(\mathbf{r}_t) = v_{t+1} = S_{t+1} = H_t(\mathbf{r}_t)$. Hence $\mathbf{r}$ lies in
\[
F_t = \bigl\{\mathbf{r} : g_t\neq H_t,\ \deg g_t\le D_t,\ g_t(\mathbf{r}_t) = H_t(\mathbf{r}_t)\bigr\}.
\]
$F_t$ depends only on $\mathbf{r}_{[0,t+1)}$. For a fixed prefix $\mathbf{u} = \mathbf{r}_{[0,t)}$, the polynomials $g_t$ and $H_t$ are fixed; if $g_t = H_t$ or $\deg g_t>D_t$, no extension of $\mathbf{u}$ lies in $F_t$, and otherwise $g_t - H_t$ is a nonzero polynomial of degree at most $D_t$ with at most $D_t$ roots. So by Lemma~\ref{lem:split}(b) and then (a),
\[
\Pr_{\mathbf{r}\leftarrow\F^m}[F_t] = q^{-t}\sum_{\mathbf{u}\in\F^t}\Pr_{\rho\leftarrow\F}\bigl[\mathbf{u}\,\|\,\rho\in F_t'\bigr] \le \frac{D_t}{q},
\]
where $F_t'$ is $F_t$ read on the first $t+1$ coordinates. The union bound over $t\in[m]$ gives $\sum_{t}D_t/q$.
\end{proof}

\subsection{Proof of \texorpdfstring{Lemma~\ref{lem:chain-bound}}{the chain bound}}

\begin{proof}[Proof of Lemma~\ref{lem:chain-bound}]
We prove the statement by induction on the length $n$ of the chain, for all claim lists $C$, values $v$ and adversaries $A$ at once. Let $(s,L,b) = (s_0,L_0,b_0)$, $M = \sum_{i\in[n]}(2s_i+1)$ and $(\mathrm{pr},\mathit{st},r',\cdot,\cdot) = A(0,())$, and for $\mathbf{u}\in\F^{2s}$ let
\[
E_1(\mathbf{u}) \iff \cSC\bigl(\mathrm{pr},\mathit{st},(\{0,1\},P(L,s,C,b),v),r',((j,\mathbf{u}_j))_{j\in[2s]}\bigr).
\]
Disjunct (i) of Definition~\ref{def:chainbad} at $\mathbf{r}$ is $E_1(\mathbf{r}_{[0,2s)})$.

\emph{Bound for disjunct (i).} By Lemma~\ref{lem:split}(b) with $k = 2s$,
\[
\Pr_{\mathbf{r}\leftarrow\F^M}\bigl[E_1(\mathbf{r}_{[0,2s)})\bigr] = \Pr_{\mathbf{u}\leftarrow\F^{2s}}\bigl[E_1(\mathbf{u})\bigr] \le \frac{2s\,\dbnd(s)}{q},
\]
by Lemma~\ref{lem:layer-sumcheck} applied to $p = P(L,s,C,b)$ and $d = \dbnd(s)$: the representation premise holds by hypothesis, every gate of $L$ reads below $2^s$ by the chain conditions, and $v\neq\ccs_L(C,b)$.

\emph{Case $n = 1$.} Disjunct (ii) requires a nonempty remaining chain, so the event is disjunct (i), and its probability is at most $2s\,\dbnd(s)/q\le(2s\,\dbnd(s)+1)/q$.

\emph{Case $n\geq 2$.} Let $(s',L',b') = (s_1,L_1,b_1)$, let $\mathcal{C}'$ be the remaining chain and $M' = M-(2s+1)$. For $\mathbf{r}\in\F^M$ let $(t_x,t_y)$ be the carried values of $A(0,\mathbf{r}_{[0,2s)})$, $\xi = \mathbf{r}_{[0,s)}$, $\eta = \mathbf{r}_{[s,2s)}$, $\gamma = \mathbf{r}_{2s}$, and
\begin{gather*}
\mathrm{TF}(\mathbf{r}) \iff t_x\neq\widetilde b(\xi)\ \text{or}\ t_y\neq\widetilde b(\eta),\qquad
C'(\mathbf{r}) = ((\xi,1),(\eta,\gamma)),\\
v'(\mathbf{r}) = \bigl(t_x - \widetilde{\mathrm{const}}_{L'}(\xi)\bigr) + \gamma\bigl(t_y - \widetilde{\mathrm{const}}_{L'}(\eta)\bigr),\qquad
A'_{\mathbf{r}}(j,\pi) = A(j+1,\mathbf{r}_{[0,2s+1)}\,\|\,\pi).
\end{gather*}
All four depend only on $\mathbf{r}_{[0,2s+1)}$. Define
\begin{align*}
S_1 &= \{\mathbf{r} : E_1(\mathbf{r}_{[0,2s)})\},\\
S_2 &= \{\mathbf{r} : \mathrm{TF}(\mathbf{r})\ \wedge\ v'(\mathbf{r}) = \ccs_{L'}(C'(\mathbf{r}),b')\},\\
S_3 &= \{\mathbf{r} : \mathrm{TF}(\mathbf{r})\ \wedge\ v'(\mathbf{r})\neq\ccs_{L'}(C'(\mathbf{r}),b')\ \wedge\ \cBad_P(A'_{\mathbf{r}},\mathcal{C}',C'(\mathbf{r}),v'(\mathbf{r}),\mathbf{r}_{\geq 2s+1})\}.
\end{align*}
By Definition~\ref{def:chainbad}, $\cBad_P(A,\mathcal{C},C,v,\mathbf{r})$ holds iff $\mathbf{r}\in S_1$, or $\mathrm{TF}(\mathbf{r})$ and $\cBad_P(A'_{\mathbf{r}},\mathcal{C}',C'(\mathbf{r}),v'(\mathbf{r}),\mathbf{r}_{\geq 2s+1})$. Splitting the second alternative on whether $v'(\mathbf{r}) = \ccs_{L'}(C'(\mathbf{r}),b')$ shows that the event is contained in $S_1\cup S_2\cup S_3$.

$S_1$ has probability at most $2s\,\dbnd(s)/q$ by the bound for disjunct (i).

$S_2$ depends only on $\mathbf{r}_{[0,2s+1)}$. By Lemma~\ref{lem:split}(b) and then (a) with $k = 2s$ and $m = 1$,
\[
\Pr[S_2] = q^{-2s}\sum_{\mathbf{u}\in\F^{2s}}\Pr_{\gamma\leftarrow\F}\bigl[\mathbf{u}\,\|\,\gamma\in S_2\bigr].
\]
For fixed $\mathbf{u}$ the values $t_x,t_y$, the points $\xi,\eta$ and $\mathrm{TF}$ are fixed. If $\mathrm{TF}$ fails, the inner probability is $0$. Otherwise, the chain conditions give $s = w(L')$ and $b = \mathrm{eval}(L',b')$, so $\widetilde b$ is the $w(L')$-variate extension $\widetilde V$ of $V = \mathrm{eval}(L',b')$, $\mathrm{TF}$ says $t_x\neq\widetilde V(\xi)$ or $t_y\neq\widetilde V(\eta)$, and the inner event is the event of Lemma~\ref{lem:carry}. The inner probability is therefore at most $1/q$ for every $\mathbf{u}$, and $\Pr[S_2]\le 1/q$.

For $S_3$, Lemma~\ref{lem:split}(a) with $k = 2s+1$ and $m = M'$ gives
\[
\Pr[S_3] = q^{-(2s+1)}\sum_{\mathbf{u}\in\F^{2s+1}}\Pr_{\mathbf{r}''\leftarrow\F^{M'}}\bigl[\mathbf{u}\,\|\,\mathbf{r}''\in S_3\bigr].
\]
For fixed $\mathbf{u}$, the quantities $\mathrm{TF}$, $C'$, $v'$ and $A'_{\mathbf{u}}$ are fixed, and $(\mathbf{u}\,\|\,\mathbf{r}'')_{\geq 2s+1} = \mathbf{r}''$. If $\mathrm{TF}(\mathbf{u})$ and $v'(\mathbf{u})\neq\ccs_{L'}(C'(\mathbf{u}),b')$ do not both hold, the inner probability is $0$. Otherwise the inner event is $\cBad_P(A'_{\mathbf{u}},\mathcal{C}',C'(\mathbf{u}),v'(\mathbf{u}),\mathbf{r}'')$ for the layer chain $\mathcal{C}'$, whose first table is $b'$, with the false claim value $v'(\mathbf{u})\neq\ccs_{L'}(C'(\mathbf{u}),b')$ and the fixed adversary $A'_{\mathbf{u}}$. The induction hypothesis bounds its probability by $\sum_{i=1}^{n-1}(2s_i\,\dbnd(s_i)+1)/q$, and averaging over $\mathbf{u}$ preserves the bound.

The union bound gives
\[
\Pr\bigl[\cBad_P(A,\mathcal{C},C,v,\mathbf{r})\bigr]\le\frac{2s\,\dbnd(s)}{q}+\frac{1}{q}+\frac{1}{q}\sum_{i=1}^{n-1}\bigl(2s_i\,\dbnd(s_i)+1\bigr) = \frac{1}{q}\sum_{i\in[n]}\bigl(2s_i\,\dbnd(s_i)+1\bigr).
\]
\end{proof}

The mechanized proof of \isaf{gkr\_chain\_bad\_bound} in \texttt{GKR\_Assembly.thy} follows the same steps.

\subsection{Proof of \texorpdfstring{Theorem~\ref{thm:b}}{the chain bound}}

\begin{proof}[Proof of Theorem~\ref{thm:b}]
Let $n_0 = \sout$, $M = \sum_{i\in[n]}(2s_i+1)$, $V_0 = \mathrm{eval}(L_0,b_0)$ and, for $\mathbf{r}\in\F^{n_0+M}$, $z = \mathbf{r}_{[0,n_0)}$ and $\mathrm{val}(z) = \widetilde o(z) - \widetilde{\mathrm{const}}_{L_0}(z)$. Define
\begin{align*}
\mathrm{SEED} &= \{\mathbf{r} : \widetilde o(z) = \widetilde{V_0}(z)\},\\
\mathrm{REST} &= \{\mathbf{r} : \mathrm{val}(z)\neq\ccs_{L_0}(((z,1)),b_0)\ \wedge\ \cBad_P(A_z,\mathcal{C},((z,1)),\mathrm{val}(z),\mathbf{r}_{\geq n_0})\}.
\end{align*}
Since $\ccs_{L_0}(((z,1)),b_0) = \widetilde{V_0}(z) - \widetilde{\mathrm{const}}_{L_0}(z)$, an outcome outside $\mathrm{SEED}$ has a false seeded claim value, so the seeded chain event is contained in $\mathrm{SEED}\cup\mathrm{REST}$.

$\mathrm{SEED}$ depends only on $z$. By Lemma~\ref{lem:split}(b) and Lemma~\ref{lem:seed}, applied to the tables $o$ and $V_0$, which differ at some $u\in[2^{n_0}]$ by hypothesis, $\Pr[\mathrm{SEED}] = \Pr_{z\leftarrow\F^{n_0}}[\widetilde o(z) = \widetilde{V_0}(z)]\le n_0/q$.

By Lemma~\ref{lem:split}(a) with $k = n_0$,
\[
\Pr[\mathrm{REST}] = q^{-n_0}\sum_{z\in\F^{n_0}}\Pr_{\mathbf{r}'\leftarrow\F^{M}}\bigl[z\,\|\,\mathbf{r}'\in\mathrm{REST}\bigr].
\]
For fixed $z$, if $\mathrm{val}(z) = \ccs_{L_0}(((z,1)),b_0)$ the inner probability is $0$. Otherwise the inner event is $\cBad_P(A_z,\mathcal{C},((z,1)),\mathrm{val}(z),\mathbf{r}')$ with a false claim value and the fixed adversary $A_z$, and Lemma~\ref{lem:chain-bound} bounds it by $\sum_{i}(2s_i\,\dbnd(s_i)+1)/q$. Averaging preserves the bound, and the union bound over $\mathrm{SEED}$ and $\mathrm{REST}$ gives the theorem.
\end{proof}

The mechanized proof of \isaf{gkr\_assembly\_soundness} in \texttt{GKR\_Assembly.thy} has the same structure.


\section{Proofs for the Quartic Extension}
\label{app:ext-proofs}

This appendix gives the two routine verifications used in \cref{sec:smallfield}.

\paragraph{The quadratic construction.} Let $R$ be a commutative ring with $0\neq 1$ and $\delta\in R$. On $R\times R$ define
\[
(a,b)+(c,e) = (a+c,\ b+e),\qquad (a,b)\,(c,e) = (ac+\delta\,be,\ ae+bc).
\]
We show that this is a commutative ring with zero $(0,0)$ and unit $(1,0)$, as used in the proof of \cref{prop:ext-field}; both levels of the tower of \cref{def:ext-tower} are of this form, with $(R,\delta) = (\F_p,3)$ and $(R,\delta) = (\extE,\theta)$. In the mechanization the two levels are verified separately, as the ring structures on \isaf{kb\_quad} and \isaf{koala\_bear\_ext4} in \texttt{KoalaBear\_Ext4\_Field.thy}.

Addition is componentwise, so the additive group axioms hold with $-(a,b) = (-a,-b)$. Multiplication is commutative because $(ac+\delta\,be,\ ae+bc)$ is unchanged when $(a,b)$ and $(c,e)$ are exchanged. The unit law is $(1,0)(c,e) = (c+\delta\cdot 0\cdot e,\ e+0\cdot c) = (c,e)$. Both coordinates of the product are bilinear in $(a,b)$ and $(c,e)$, which gives distributivity. For associativity, let $(f,h)$ be a third element. Then
\begin{align*}
\bigl((a,b)(c,e)\bigr)(f,h) &= \bigl((ac+\delta\,be)f+\delta\,(ae+bc)h,\ (ac+\delta\,be)h+(ae+bc)f\bigr),\\
(a,b)\bigl((c,e)(f,h)\bigr) &= \bigl(a(cf+\delta\,eh)+\delta\,b(ch+ef),\ a(ch+ef)+b(cf+\delta\,eh)\bigr),
\end{align*}
and expanding shows that both equal $\bigl(acf+\delta\,(aeh+bch+bef),\ ach+aef+bcf+\delta\,beh\bigr)$. Finally, $(0,0)\neq(1,0)$ because $0\neq 1$ in $R$.

\paragraph{The two forms of chain well-formedness.} The mechanized definition \isaf{gkr\_chain\_wf} (in \texttt{GKR\_Assembly.thy}) is recursive. The empty chain is well-formed. A chain $\bigl((s_0),(L_0),(b_0)\bigr)$ of length one is well-formed if and only if every gate of $L_0$ has input wires below $2^{s_0}$. A chain of length $n\ge 2$ is well-formed if and only if every gate of $L_0$ has input wires below $2^{s_0}$, $s_0 = w(L_1)$, $b_0 = \mathrm{eval}(L_1,b_1)$, and the chain obtained by deleting the first entry of each of the three lists is well-formed. Triples of lists of different lengths are not well-formed.

For lists of a common length $n\ge 1$ this is equivalent to conditions (Ch1) and (Ch2) of \cref{sec:smallfield}, by induction on $n$. For $n = 1$, (Ch2) is vacuous and (Ch1) is the defining condition. For $n\ge 2$, the defining conditions are (Ch1) and (Ch2) at index $0$ together with well-formedness of the shortened chain, which by the induction hypothesis is (Ch1) at indices $1,\dots,n-1$ and (Ch2) at indices $1,\dots,n-2$. Together these are (Ch1) and (Ch2) at all indices.

\section{Proofs for the Sparse-Merkle Sections}
\label{app:smt-proofs}

This appendix gives the index-level details used in the proofs of Lemmas~\ref{lem:circuit-cons} and~\ref{lem:smt-honest-enc}, and the induction that identifies the accumulator chain of Definition~\ref{def:smt}~(f) with the recursive path root of the formalization. Positions are counted from $0$; a block of $W$ entries starting at position $a$ occupies $a, \ldots, a + W - 1$.

\paragraph{Accumulator chains.} The formal path root is defined by recursion on the sibling list: $\smPathRoot(a, k, ()) = a$, and $\smPathRoot(a, k, s_0 \cdot \vec{s}\,') = \smPathRoot(a', \lfloor k/2 \rfloor, \vec{s}\,')$ with $a' = \hnode(a, s_0)$ if $k$ is even and $a' = \hnode(s_0, a)$ if $k$ is odd (\isaf{path\_root}).

\begin{lemma}[Accumulator chain; mechanized]
\label{lem:smt-acc-chain}
For every digest $a$, key $k$ and sibling list $\vec{s}$ of length $n$, the recursive path root equals the last entry $a_n$ of the accumulator chain of Definition~\ref{def:smt}~(f), and that chain satisfies $a_0 = a$ and the step rule of Definition~\ref{def:smt}~(f).
\end{lemma}

\begin{proof}
By induction on $n$. For $n = 0$ both sides are $a$. For $\vec{s} = s_0 \cdot \vec{s}\,'$, the first step of the recursion computes $a_1$, since $k \bmod 2 = \smBit_0(k)$. The recursion then continues from $a_1$ with the key $\lfloor k/2 \rfloor$ and the list $\vec{s}\,'$, and $\smBit_j(\lfloor k/2 \rfloor) = \lfloor k/2^{j+1} \rfloor \bmod 2 = \smBit_{j+1}(k)$. So the chain of $(a_1, \lfloor k/2 \rfloor, \vec{s}\,')$ is $(a_1, \ldots, a_n)$, and by induction the recursion returns $a_n$.
\end{proof}

\paragraph{Positions.} Table~\ref{tab:smt-input} lists the positions of the segments of an input vector (\isaf{ro\_lp}, \ldots, \isaf{up\_width} in \texttt{Compiler\_Model.thy}); they are the partial sums of the block lengths in Definition~\ref{def:encoding}. Table~\ref{tab:smt-ev} lists the positions of the entries of $\smEv_\tau(x)$ (\isaf{ev\_hl}, \ldots, \isaf{ev2\_vd}).

\begin{table}[htbp]
\centering
\small
\caption{Positions of the segments of an input vector of length at least $n_\tau$.}
\label{tab:smt-input}
\begin{tabular}{@{}lll@{}}
\toprule
Segment & Read-only kinds & Updates\\
\midrule
$P$ & $0$ ($\smWl$ entries) & $0$ ($\smWl$ entries)\\
$P'$ & & $\smWl$ ($\smWl$ entries)\\
$A_j$, $j \leq d$ & $\smWl + Wj$ & $2\smWl + Wj$\\
$A'_j$, $j \leq d$ & & $2\smWl + W(d+1) + Wj$\\
$S_j$, $j < d$ & $\smWl + W(d+1) + Wj$ & $2\smWl + 2W(d+1) + Wj$\\
$\beta_j$, $j < d$ & $\smWl + W(2d+1) + j$ & $2\smWl + W(3d+2) + j$\\
$R$ & $\smWl + W(2d+1) + d$ & $2\smWl + W(3d+2) + d$\\
$R'$ & & $2\smWl + W(3d+3) + d$\\
$V$ & $\smWl + W(2d+2) + d$ & $2\smWl + W(3d+4) + d$\\
\midrule
$n_\tau$ & $\smWl + W(2d+3) + d$ & $2\smWl + W(3d+5) + d$\\
\bottomrule
\end{tabular}
\end{table}

\begin{table}[htbp]
\centering
\small
\caption{Positions of the entries of $\smEv_\tau(x)$; every block other than $\beta_j$ and $x_0$ has $W$ entries.}
\label{tab:smt-ev}
\begin{tabular}{@{}lll@{}}
\toprule
Entry & Read-only kinds & Updates\\
\midrule
$\hwl(P)$ & $0$ & $0$\\
$\hwl(P')$ & & $W$\\
$M_j$, $j < d$ & $W + Wj$ & $2W + Wj$\\
$M'_j$, $j < d$ & & $2W + Wd + Wj$\\
$A_j$, $j \leq d$ & $W(d+1) + Wj$ & $W(2d+2) + Wj$\\
$A'_j$, $j \leq d$ & & $W(3d+3) + Wj$\\
$\beta_j$, $j < d$ & $W(2d+2) + j$ & $W(4d+4) + j$\\
$R$ & $W(2d+2) + d$ & $W(4d+4) + d$\\
$R'$ & & $W(4d+5) + d$\\
$V$ & $W(2d+3) + d$ & $W(4d+6) + d$\\
$x_0$ & $W(2d+4) + d$ & \\
\midrule
$e_\tau$ & $W(2d+4) + d + 1$ & $W(4d+7) + d$\\
\bottomrule
\end{tabular}
\end{table}

\begin{lemma}[Entries of the stack output; mechanized]
\label{lem:smt-ev-entries}
If $x$ has length at least $n_\tau$, then $\smEv_\tau(x)$ has length $2^{e_\tau}$ and holds at the positions of Table~\ref{tab:smt-ev} the entries named there.
\end{lemma}

\begin{proof}
Every segment of Table~\ref{tab:smt-input} lies in the first $n_\tau$ positions, so the segments of $x$ are defined. The unpadded list is the concatenation of the blocks in the order of Table~\ref{tab:smt-ev}, with $W$ entries for each hash output, accumulator and digest block (the codomain of $\hwl$ and $\hwn$ is $\F^{W}$) and one entry for each key bit and for $x_0$. The position of a block is the sum of the lengths of the blocks before it, which gives the table, and the total length is $e_\tau$. The zero padding is appended after position $e_\tau$ and leaves these entries unchanged.
\end{proof}

\paragraph{Residual lists.} Writing $\mathrm{D}(a, b)$ for the residual $y_a - y_b$ of a vector $y$, $\mathrm{B}(a)$ for $y_a^2 - y_a$ and $\mathrm{T}(a)$ for $y_a^2 - 2 y_a$, and using the positions of Table~\ref{tab:smt-ev}, the list $\Delta_\tau$ for a read-only kind is
\[
\begin{aligned}
&\mathrm{D}(i,\ W(d+1) + i) && (i < W) && \text{(C1)}\\
&\mathrm{D}(W + Wj + i,\ W(d+1) + W(j+1) + i) && (j < d,\ i < W) && \text{(C2)}\\
&\mathrm{D}(W(2d+1) + i,\ W(2d+2) + d + i) && (i < W) && \text{(C3)}\\
&\mathrm{D}(W(d+1) + i,\ W(2d+3) + d + i) && (i < W) && \text{(C4)}\\
&\mathrm{B}(W(2d+2) + j) && (j < d) && \text{(C5)}\\
&\mathrm{T}(W(2d+4) + d) && \text{only for } \tau = \smKN && \text{(C6)}
\end{aligned}
\]
in this order (\isaf{descs\_ro}), and for updates it is the analogous list over the update column of Table~\ref{tab:smt-ev} for (C1) twice, (C2) twice, (C3) twice, (C4$'$) and (C5) (\isaf{descs\_up}).

\begin{lemma}[Residual layer; mechanized]
\label{lem:smt-res-layer}
Let $\Delta$ be a residual list and $y$ a vector containing every position that $\Delta$ reads. Then $\mathrm{eval}(\smRes(\Delta), y)$ has length $2^{|\Delta|}$; its entry $z$ is the value at $y$ of the $z$-th residual for $z < |\Delta|$, and $0$ for $|\Delta| \leq z < 2^{|\Delta|}$. Consequently all its entries are zero iff every residual of $\Delta$ vanishes at $y$.
\end{lemma}

\begin{proof}
By Definition~\ref{def:circuit}, entry $z$ of a layer evaluation is the sum of the constants of wire $z$ plus, over the gates that write wire $z$, the contributions $c \cdot g(y_{\mathit{in}_1}, y_{\mathit{in}_2})$, where $g(u, t) = u$ for a linear gate and $g(u, t) = u t$ for a product gate. The layer $\smRes(\Delta)$ has no constants, and its gate list is the concatenation over $z < |\Delta|$ of the two gates of the $z$-th residual, both writing wire $z$. The gates writing wire $z$ are therefore exactly these two for $z < |\Delta|$, and there are none for $z \geq |\Delta|$. Their contributions sum to $1 \cdot y_a + (-1) \cdot y_b$ for $\mathrm{D}(a, b)$, to $1 \cdot y_a y_a + (-1) \cdot y_a$ for $\mathrm{B}(a)$, and to $y_a y_a - 2 y_a$ for $\mathrm{T}(a)$. Since $z < |\Delta| < 2^{|\Delta|}$, each residual sits on a wire of the layer, which gives the last statement.
\end{proof}

\begin{proof}[Proof of Lemma~\ref{lem:circuit-cons} in full]
Let $x$ have length $2^{n_\tau} \geq n_\tau$ and let $y = \smEv_\tau(x)$. By (S), the output of $C_\tau$ on $x$ is $\mathrm{eval}(\smRes(\Delta_\tau), y)$, so by Lemma~\ref{lem:smt-res-layer}, $C_\tau$ accepts $x$ iff every residual of $\Delta_\tau$ vanishes at $y$. By Lemma~\ref{lem:smt-ev-entries}, for a read-only kind the residuals of the first line are $\hwl(P)_i - (A_0)_i$ for $i < W$; they all vanish iff $\hwl(P) = A_0$, both vectors having $W$ entries, which is (C1). In the same way the second line gives $(M_j)_i - (A_{j+1})_i$ and (C2), the third $(A_d)_i - R_i$ and (C3), and the fourth $(A_0)_i - V_i$ and (C4). The fifth line gives $\beta_j^2 - \beta_j = \beta_j(\beta_j - 1)$ and (C5), and the last $x_0^2 - 2x_0 = x_0(x_0 - 2)$ and (C6). The update case reads the update column in the same way (\isaf{res\_sem\_descs\_ro\_iff}, \isaf{res\_sem\_descs\_up\_iff}).
\end{proof}

\paragraph{Encoded inputs.} The remaining facts used in the proof of Lemma~\ref{lem:smt-honest-enc} are the following.
\begin{enumerate}[label=(\alph*),leftmargin=*]
\item \emph{Segments of an encoded vector} (\isaf{enc\_ro\_segs}, \isaf{enc\_up\_segs}). For $x = \smEnc_\tau(\ell, k, \vec{s}, r, \mathit{vd})$ with $|\vec{s}| = d$, the segments of Table~\ref{tab:smt-input} are $P = \smPre(\ell)$, $A_j = \smRep(a_j)$, $S_j = \smRep(s_j)$, $\beta_j = \smBit_j(k)$, $R = \smRep(r)$ and $V = \smRep(\mathit{vd})$, and $x_0 = \smTag(\ell)$; similarly for updates, with $P' = \smPre(\ell')$, $A'_j = \smRep(a'_j)$ and $R' = \smRep(r')$. The vector $x^{\circ}$ of Definition~\ref{def:encoding} is the concatenation of blocks of lengths $\smWl$, $W$ ($d+1$ times), $W$ ($d$ times), $1$ ($d$ times), $W$ and $W$, so the block positions are the partial sums of Table~\ref{tab:smt-input}; the padding comes after position $n_\tau$; and $x_0$ is the first entry of $\smPre(\ell)$, which is $\smTag(\ell)$.
\item \emph{Key bits} (\isaf{kbit\_cases}, \isaf{kbit\_bool}). As a field element, $\smBit_j(k) = \nu(\lfloor k/2^j \rfloor \bmod 2) \in \{0, 1\}$, so $\smBit_j(k)(\smBit_j(k) - 1) = 0$.
\item \emph{Multiplexer} (\isaf{mux\_zero}, \isaf{mux\_one}). $\smMuxL(0, u, t) = u$ and $\smMuxR(0, u, t) = t$; $\smMuxL(1, u, t) = u + (t - u) = t$ and $\smMuxR(1, u, t) = t + (u - t) = u$.
\item \emph{Chain} (Lemma~\ref{lem:smt-acc-chain}). The encoded chain satisfies the step rule, and since $|\vec{s}| = d$ its last entry is $a_d = \smPathRoot(\hleaf(\ell), k, \vec{s})$.
\item \emph{Tag} (\isaf{leaf\_tag\_val\_zero\_or\_two}). For every leaf state $\lambda$, $\smTag(\lambda) \in \{0, 2\}$ iff $\lambda \in \{\smEmpty, \smTomb\}$: the tags are $0$, $1$ and $2$, $1 \neq 0$, and $1 = 2$ would give $1 = 0$. In an integral domain, $t(t - 2) = 0$ iff $t \in \{0, 2\}$.
\end{enumerate}
With (a) to (e), the proof of Lemma~\ref{lem:smt-honest-enc} in Section~\ref{sec:compiler} is complete. In the formal proof of \isaf{cons\_ro\_enc}, the unconditional constraints are the steps \texttt{c1}, \texttt{c2} and \texttt{c5}, and the three equivalences are \texttt{c3}, \texttt{c4} and \texttt{c6}; \isaf{cons\_up\_enc} has the same structure, with \isaf{up\_chain\_holds} for the compression steps of either chain.

\section{Correspondence with the Formal Development}
\label{app:mechanization}

\Cref{tab:mechanization} lists every numbered result of the paper with its status (\cref{sec:mechanization}) and the Isabelle theorems that state it.

\begin{table}[ht]
\centering
\caption{Status of every numbered result. M: mechanized; an Isabelle theorem states the result up to notation, with the same or weaker hypotheses. M$^{*}$: mechanized, with a difference described after the result. P: proved on paper only. I: imported from the AFP sumcheck entry, with a paper proof in \cref{app:gkr-proofs}. S: standard result, cited. Theory names are given without the suffix \texttt{.thy}.}
\label{tab:mechanization}
\footnotesize
\renewcommand{\arraystretch}{1.0}
\begin{tabular}{@{}>{\raggedright\arraybackslash}p{0.21\textwidth}>{\raggedright\arraybackslash}p{0.17\textwidth}>{\raggedright\arraybackslash}p{0.56\textwidth}@{}}
\toprule
Result & Status & Isabelle theorems (theory) \\
\midrule
\Cref{lem:mle-facts} & M & \isaf{eq\_pi\_\allowbreak bool\_point}, \isaf{mle\_\allowbreak bool\_point}, \isaf{mle\_add}, \isaf{mle\_scale}, \isaf{beta\_\allowbreak bool\_point}, \isaf{sum\_\allowbreak beta\_\allowbreak bool\_points} (Multilinear\_Extension); \isaf{mle\_cons} (GKR\_Assembly) \\
\Cref{lem:schwartz-zippel} & S & used in paper proofs only \\
\Cref{thm:sumcheck-afp} & I & \isaf{soundness\_\allowbreak inductive} (AFP Soundness\_Proof); paper proof in \cref{app:gkr-proofs} \\
\Cref{thm:sumcheck-std} & S & proof in \cref{app:gkr-proofs} \\
\Cref{lem:split} & M & \isaf{prob\_\allowbreak tuples\_\allowbreak split\_seg}, \isaf{prob\_\allowbreak tuples\_\allowbreak prefix} (GKR\_Assembly) \\
\Cref{lem:layer-identity} & M & \isaf{layer\_\allowbreak hypercube\_\allowbreak identity} (Sumcheck\_Instance) \\
\Cref{lem:integrand-poly} & P & \\
\Cref{lem:layer-sumcheck} & M & \isaf{layer\_\allowbreak sumcheck\_\allowbreak soundness} (Sumcheck\_Instance) \\
\Cref{lem:carry} & M & \isaf{carried\_\allowbreak claim\_\allowbreak carry\_\allowbreak sound} (GKR\_Assembly) \\
\Cref{lem:seed} & M & \isaf{mle\_\allowbreak agree\_prob} (GKR\_Assembly) \\
\Cref{lem:chain-bound} & M & \isaf{gkr\_chain\_\allowbreak bad\_bound} (GKR\_Assembly) \\
\Cref{thm:b} & M & \isaf{gkr\_\allowbreak assembly\_\allowbreak soundness} (GKR\_Assembly) \\
\Cref{cor:b-deg4} & M, vacuous, not used & \isaf{gkr\_\allowbreak assembly\_\allowbreak soundness\_deg4} (GKR\_Assembly); likewise \isaf{gkr\_\allowbreak assembly\_\allowbreak soundness\_\allowbreak deg4\_ext4} and \isaf{theorem\_C\_\allowbreak composition\_\allowbreak deg4} \\
\Cref{lem:accept-bad} & P & \\
\Cref{prop:b-faithful} & P & \\
\Cref{lem:ext-base} & M & \isaf{kb\_p\_prime}, \isaf{kb\_\allowbreak inverse\_\allowbreak exists}, \isaf{CARD\_\allowbreak koala\_bear} (KoalaBear\_Field) \\
\Cref{lem:nonsquare} & M & \isaf{kb\_three\_\allowbreak nonsquare}, \isaf{kb\_minus\_\allowbreak three\_\allowbreak nonsquare} (KoalaBear\_Ext4\_Nonsquare); \isaf{kb\_quad\_\allowbreak y\_\allowbreak nonsquare} (KoalaBear\_Ext4\_Field) \\
\Cref{prop:ext-field} & M: (a)--(d), bijectivity in (e); P: formula in (e), (f) & field instances on \isaf{kb\_quad} and \isaf{koala\_\allowbreak bear\_ext4}, \isaf{CARD\_\allowbreak kb\_quad}, \isaf{CARD\_\allowbreak koala\_\allowbreak bear\_ext4}, \isaf{kb4\_\allowbreak embed\_simps}, \isaf{kb4\_\allowbreak embed\_inj}, \isaf{kb4\_\allowbreak embed\_\allowbreak inverse}, \isaf{kb\_alpha\_\allowbreak square}, \isaf{kb\_alpha\_\allowbreak fourth}, \isaf{kb\_alpha\_\allowbreak not\_in\_\allowbreak base} (KoalaBear\_Ext4\_Field); \isaf{coeff4\_bij} (KoalaBear\_Ext4\_GKR) \\
\Cref{lem:pushforward} & M & \isaf{uniform\_\allowbreak coeff\_\allowbreak tuple\_\allowbreak pushforward} (KoalaBear\_Ext4\_GKR) \\
\Cref{thm:b-ext4} & M & \isaf{gkr\_\allowbreak assembly\_\allowbreak soundness\_ext4}, \isaf{ext4\_\allowbreak denominator\_\allowbreak exact} (KoalaBear\_Ext4\_GKR) \\
\Cref{rem:ext-scale} & P & arithmetic only \\
\Cref{lem:lift-layer} & M: (a)--(c), (e); (d), (f) under well-formedness, P as stated & (a) \isaf{kb4\_\allowbreak embed\_\allowbreak sum\_list}, \isaf{kb4\_\allowbreak embed\_\allowbreak prod\_list}, \isaf{kb4\_\allowbreak embed\_\allowbreak sum\_\allowbreak lessThan}, \isaf{kb4\_\allowbreak embed\_\allowbreak prod\_\allowbreak lessThan}, \isaf{kb4\_\allowbreak embed\_\allowbreak power}; (b) \isaf{kb4\_\allowbreak embed\_\allowbreak eq\_pi}, \isaf{kb4\_\allowbreak embed\_\allowbreak mle\_\allowbreak base\_point}, \isaf{lifted\_\allowbreak mle\_\allowbreak boolean}; (c) \isaf{layer\_\allowbreak width\_lift}, \isaf{const\_\allowbreak at\_lift}; (e) \isaf{wf\_\allowbreak layer\_lift}; (d), (f) \isaf{layer\_\allowbreak eval\_embed}, \isaf{false\_\allowbreak output\_\allowbreak preserved}, with \isaf{gate\_\allowbreak contrib\_\allowbreak sum\_lift} (KoalaBear\_Ext4\_Lift) \\
\Cref{prop:lift-chain} & P & \\
\Cref{lem:leaffold} & M$^{*}$ & \isaf{leaf\_\allowbreak fold\_inj}, \isaf{h\_leaf\_\allowbreak collision\_\allowbreak to\_\allowbreak sponge\_\allowbreak collision} (Leaf\_Fold) \\
\Cref{prop:binding} & P; last step M & \isaf{encoded\_\allowbreak state\_\allowbreak collision\_\allowbreak to\_\allowbreak sponge\_\allowbreak collision} (Leaf\_Fold) \\
\Cref{lem:circuit-cons} & M$^{*}$ & \isaf{A\_i\_ro}, \isaf{A\_i\_up} (Compiler\_Correctness) \\
\Cref{lem:smt-honest-enc} & M$^{*}$ & \isaf{cons\_ro\_enc}, \isaf{cons\_up\_enc} (Compiler\_Correctness) \\
\Cref{thm:a-sound} & M & \isaf{theorem\_A\_\allowbreak soundness} (Compiler\_Correctness) \\
\Cref{thm:a-complete} & M & \isaf{theorem\_A\_\allowbreak completeness} (Compiler\_Correctness) \\
\Cref{cor:a-iff} & M & \isaf{theorem\_A\_\allowbreak compilation\_\allowbreak soundness} (Compiler\_Correctness) \\
\Cref{thm:c} & M & \isaf{theorem\_C\_\allowbreak composition} (Composition) \\
\Cref{cor:perinput} & P & \\
\Cref{cor:perinput-verifier} & P & \\
\Cref{cor:perinput-ext} & P & \\
\Cref{thm:a-vector} & P & uses \isaf{A\_i\_ro}, \isaf{A\_i\_up} (Compiler\_Correctness) and \isaf{mux\_zero}, \isaf{mux\_one} (Compiler\_Model) \\
\Cref{rem:vector-constraints} & P & \\
\Cref{cor:perinput-vector} & P & \\
\Cref{cor:perinput-vector-ext} & P & \\
\Cref{cor:one-leaf} & P & \\
\Cref{lem:smt-acc-chain} & M & \isaf{acc\_list\_0}, \isaf{acc\_\allowbreak list\_step}, \isaf{acc\_\allowbreak list\_last} (Compiler\_Model) \\
\Cref{lem:smt-ev-entries} & M & \isaf{length\_\allowbreak ev\_ro}, \isaf{length\_\allowbreak ev\_up} and the \isaf{ev\_ro\_nth} and \isaf{ev\_up\_nth} lemmas (Compiler\_Correctness) \\
\Cref{lem:smt-res-layer} & M & \isaf{mk\_res\_\allowbreak layer\_\allowbreak eval\_in}, \isaf{mk\_res\_\allowbreak layer\_\allowbreak eval\_pad}, \isaf{mk\_res\_\allowbreak layer\_\allowbreak accept\_iff} (Compiler\_Model) \\
\bottomrule
\end{tabular}
\end{table}

\end{document}